\documentclass[conference]{IEEEtran}

\usepackage{amsmath, amsfonts, amssymb}
\usepackage{algorithm, algorithmic}
\usepackage{graphicx, comment}
\usepackage{color,subfigure,cite}

\def\Rank{\mathop{\rm Rank}\limits}

\def\GF{\mathop{\sf GF}\limits}

\def\tran{\mathrm{T}}

\def\linsp{{\mathsf{span}}}

\def\X{{\mathbf{X}}}

\def\Z{{\mathbf{Z}}}

\def\vv{{\mathbf{v}}}

\def\vv{{\mathbf{v}}}
\def\w{{\mathbf{w}}}

\def\0{{\mathbf{0}}}
\def\1{{\mathbf{1}}}

\def\prop{{\mathsf{Prob}}}
\def\EE{{\mathsf{E}}}

\def\lmat{\left(\begin{matrix}}
\def\rmat{\end{matrix}\right)}
\def\eqref#1{(\ref{#1})}

\newenvironment{thmproof}[1]
{\noindent\hspace{2em}{\it #1 }}
{\hspace*{\fill}~\QED\par\endtrivlist\unskip}

\newtheorem{theorem}{Theorem}
\newtheorem{definition}{Definition}
\newtheorem{lemma}{Lemma}

\newtheorem{corollary}{Corollary}
\newtheorem{proposition}{Proposition}

\def\Co{{\mathsf{Co}}}

\def\BibTeX{{\rm B\kern-.05em{\sc i\kern-.025em b}\kern-.08em
    T\kern-.1667em\lower.7ex\hbox{E}\kern-.125emX}}

\def\algtop#1{\vspace{.4cm}\hrule\vspace{.035cm}\hrule\vspace{.1cm}\noindent{\S~\sc #1}\par}
\def\algbot{\hrule\vspace{.035cm}\hrule\vspace{.4cm}}

\begin{document}

\title{Capacity of 1-to-$K$ Broadcast Packet Erasure Channels with Channel Output Feedback}

\author{\IEEEauthorblockN{Chih-Chun Wang}
\IEEEauthorblockA{
Center of Wireless Systems and Applications (CWSA)\\
School of Electrical and Computer Engineering, Purdue University, USA}
}

%

\maketitle

\begin{abstract}
This paper focuses on the 1-to-$K$ broadcast packet erasure channel (PEC), which is a generalization of the 
broadcast binary erasure channel from the binary symbol to that of arbitrary finite fields $\GF(q)$ with sufficiently large $q$. We consider the setting in which the source node has instant feedback of the channel outputs of the $K$ receivers after each transmission. The capacity region of the 1-to-$K$ PEC with COF was previously known only for the case $K=2$.  Such a setting directly models {network coded packet transmission in the downlink direction} with integrated feedback mechanisms (such as Automatic Repeat reQuest (ARQ)).



 The main results of this paper are:  (i) The capacity region for general 1-to-3 broadcast PECs, and (ii) The capacity region for two types of 1-to-$K$ broadcast PECs:
the {symmetric} PECs, and the {spatially independent} PECs with {one-sided fairness constraints}. This paper also develops
(iii) A pair of outer and inner bounds of the capacity region for arbitrary 1-to-$K$ broadcast PECs, which can be easily evaluated by any linear programming solver. The proposed inner bound is proven by a new class of intersession network coding schemes, termed the {\em packet evolution} schemes, which is based on the concept of {\em code alignment} in $\GF(q)$ that is in parallel with the interference alignment techniques for the Euclidean space. Extensive numerical experiments show that the outer and inner bounds meet for almost all broadcast PECs encountered in practical scenarios and thus effectively bracket the  capacity of general 1-to-$K$ broadcast PECs with COF.


%
%
\end{abstract}


\section{Introduction}

Broadcast channels have been actively studied since the inception of network information theory. Although the broadcast capacity region remains unknown for general channel models, significant progress has been made in various sub-directions  (see \cite{Cover98} for a tutorial paper), including but not limited to the degraded broadcast channel models \cite{Bergmans73}, the 2-user capacity with degraded message sets \cite{KornerMarton77} or with message side information \cite{Wu07}. Motivated by wireless broadcast communications, 
the Gaussian broadcast channel (GBC) \cite{WeingartenSteinbergShamai06} is among the most widely studied broadcast channel models. 

In the last decade, the new network coding concept has emerged \cite{LiYeungCai03}, which focuses on achieving the capacity of a communication network. More explicitly, the network-coding-based approaches generally model each hop of a packet-based communication network by a {\em packet erasure channel} (PEC) instead of the classic Gaussian channel. Such simple abstraction allows us to explore the information-theoretic capacity of a much larger network with mathematical rigor and also sheds new insights on the network effects of a communication system. One such example is that when all destinations are interested in the same set of packets, the capacity of any arbitrarily large, multi-hop PEC network can be characterized by the corresponding min-cut/max-flow values \cite{LiYeungCai03,DanaGowaikarPalankiHassibiEffros06}. 
Another example is the broadcast channel capacity with message side information. Unlike the existing GBC-based results that are limited to the simplest 2-user scenario \cite{Wu07}, the capacity region for 1-to-$K$ broadcast PECs with message side information has been derived for $K=3$ and tightly bounded for general $K$ values \cite{Wang10a,WangKhreishahShroff09}.\footnote{The results of 1-to-$K$ broadcast PECs with message side information \cite{Wang10a,WangKhreishahShroff09} is related to the capacity of the ``XOR-in-the-air" scheme \cite{KattiRahulHuKatabiMedardCrowcroft06} in a wireless network.} 
In addition to providing new insights on network communications, this simple PEC-based abstraction in network coding also accelerates the transition from theory to practice.
Many of the capacity-achieving {\em network codes} \cite{HoMedardKoetterKargerEffrosShiLeong06} have since been implemented for either the wireline \cite{ChouWuJain03} or the wireless multi-hop networks \cite{KattiRahulHuKatabiMedardCrowcroft06,KoutsonikolasWangHu10}.

Motivated by the state-of-the-art wireless network coding protocols and the corresponding applications, this paper studies the memoryless 1-to-$K$ broadcast PEC with Channel Output Feedback (COF). Namely, a single source node sends out a stream of packets wirelessly, which carries information of $K$ independent downlink data sessions, one for each receiver $d_k$, $k=1,\cdots, K$, respectively. Due to the randomness of the underlying wireless channel condition, which varies independently for each time slot, each transmitted packet may or may not be heard by a receiver $d_k$. After packet transmission, each $d_k$ then informs the source its own channel output by sending back the ACKnowledgement (ACK) packets periodically (batch feedback) or after each time slot (per-packet instant feedback) \cite{XueYang08}.
\cite{GeorgiadisTassiulas09} derives the capacity region of the memoryless 1-to-2 broadcast PEC with COF. The results show that COF strictly improves the capacity of the memoryless 1-to-2 broadcast PEC, which is in sharp contrast with the classic result that feedback does not increase the capacity for any memoryless 1-to-1 channel. \cite{GeorgiadisTassiulas09} can also be viewed as a mirroring result to the achievability results of GBCs with COF \cite{OzarowCheong84}.  It is worth noting that other than increasing the achievable throughput, COF can also be used for queue and delay management \cite{LiWangLin10,SundararajanShahMedard07} and for rate-control in a wireless network coded system \cite{KoutsonikolasWangHu10}.

The main contribution of this work includes: (i) The capacity region for general 1-to-3 broadcast PECs with COF; 
(ii) The capacity region for two types of 1-to-$K$ broadcast PECs with COF:
the {\em symmetric} PECs, and the {\em spatially independent} PECs with {\em one-sided fairness constraints}; and
(iii) A pair of outer and inner bounds of the capacity region for general 1-to-$K$ broadcast PECs with COF, which can be easily evaluated by any linear programming solver. Extensive numerical experiments show that the outer and inner bounds meet for almost all broadcast PECs encountered in practical scenarios and thus effectively bracket the exact capacity region.

The capacity outer bound in this paper is derived by generalizing the degraded channel argument first proposed in \cite{OzarowCheong84}. For the achievability part of (i), (ii), and (iii), we devise a new class of inter-session network coded schemes, termed the {\em packet evolution method}. The packet evolution method is based on a novel concept of {\em network code alignment}, which is the PEC-counterpart of the interference alignment method originally proposed for Gaussian interference channels \cite{CadambeJafar08,DasVishwanathJafarMarkopoulou10}. It is worth noting that in addition to the random PEC model in this paper, there are other promising channel models that also greatly facilitate capacity analysis for larger networks. One such example is the deterministic wireless channel model proposed in \cite{AvestimehrDiggaviTse07}, which can also be viewed as a deterministic degraded binary erasure channel.

The rest of this paper is organized as follows. Section~\ref{sec:setting} contains the basic setting as well as the detailed comparison to the existing results in \cite{GeorgiadisTassiulas09,LarssonJohansson06,RoznerIyerMehtaQiuJafry07} via an illustrating example. Section~\ref{sec:main} describes the main theorems of this paper and the proof of the converse theorem. In particular, Section~\ref{subsec:general} focuses on the capacity results for arbitrary broadcast PEC parameters while Section~\ref{subsec:special} considers two special types of broadcast PECs: the symmetric and the spatially independent PECs, respectively.
 Section~\ref{sec:achievability} introduces
a new class of network coding schemes, termed the packet evolution (PE) method. Based on the PE method, Section~\ref{sec:pf-of-achievability}
outlines the proofs of the achievability results in Section~\ref{sec:main}. Some theoretic implications and discussions are included in Section~\ref{sec:discussion}. Section~\ref{sec:conclusion} concludes this paper.


\section{Problem Setting \& Existing Results\label{sec:setting}}

\subsection{The Memoryless 1-to-$K$ Broadcast Packet Erasure Channel} 
For any positive integer $K$, we use $[K]\stackrel{\Delta}{=}\{1,2,\cdots,K\}$ to denote the set of integers from 1 to $K$, and use $2^{[K]}$ to denote the collection of all subsets of $[K]$.

Consider a 1-to-$K$ broadcast PEC from a single source $s$ to $K$ destinations $d_k$, $k\in [K]$. For each channel usage, the 1-to-$K$ broadcast PEC takes an input symbol $Y\in\GF(q)$ from $s$ and outputs a $K$-dimensional vector $\Z\stackrel{\Delta}{=}(Z_1,\cdots, Z_K)\in (\{Y\}\cup \{*\})^K$,  where the $k$-th coordinate $Z_k$ being ``$*$" denotes that the transmitted symbol $Y$ does not reach the $k$-th receiver $d_k$ (thus being erased). We also assume that there is no other type of noise, 
i.e.,
the individual output is either equal to the input $Y$ or an erasure ``$*$." The {\em success probabilities} of a 1-to-$K$ PEC are described by $2^K$ non-negative parameters: $p_{S\overline{[K]\backslash S}}$ for all $S\in 2^{[K]}$ such that $\sum_{S\in 2^{[K]}}p_{S\overline{[K]\backslash S}}=1$ and for all $y\in \GF(q)$, 
\begin{align}
\prop\left(\left.\{k\in[K]:Z_k=y\}=S\right|Y=y\right)=p_{S\overline{[K]\backslash S}}.\nonumber
\end{align}

That is, $p_{S\overline{[K]\backslash S}}$ denotes the probability that the transmitted symbol $Y$ is received {\em by and only by} the receivers $\{d_k:k\in S\}$. In addition to the joint probability mass function $p_{S\overline{[K]\backslash S}}$ of the success events, the following notation will be used frequently in this work. For all $S\in 2^{[K]}$, we define
\begin{align}
p_{\cup S}=\sum_{\forall S'\in 2^{[K]}: S'\cap S\neq \emptyset} p_{S'\overline{[K]\backslash S'}}.\label{eq:pS-def}
\end{align}
That is, $p_{\cup S}$ is the probability that {\em at least one of the receiver $d_k$ in $S$} successfully receives the transmitted symbol $Y$. For example, when $K=2$,
\begin{align}
p_{\cup\{1,2\}}=p_{\{1\}\overline{\{2\}}}+p_{\{2\}\overline{\{1\}}}+p_{\{1,2\}\overline{\emptyset}}\nonumber
\end{align} is the probability that at least one of $d_1$ and $d_2$ receives the transmitted symbol $Y$. We sometimes use $p_{k}$ as shorthand for $p_{\cup \{k\}}$, which is the marginal probability that the $k$-th receiver $d_k$ receives $Y$ successfully.

We can repeatedly use the channel for $n$ time slots and let $Y(t)$ and $\Z(t)$ denote the input and output for the $t$-th time slot. We assume that the 1-to-$K$ broadcast PEC is memoryless and time-invariant, i.e., for any given function $y(\cdot):[n]\mapsto \GF(q)$,
\begin{align}
\prop\left(\forall t\in [n], \{k:Z_k(t)=y(t)\}=S(t)\right.\hspace{2.5cm}\nonumber\\
\left|\forall t\in [n],Y(t)=y(t)\right)=\prod_{t=1}^np_{S(t)\overline{[K]\backslash S(t)}}.\nonumber
\end{align}
Note that this setting allows the success events among different receivers to be dependent, also defined as {\em spatial dependence}. For example, when two logical receivers $d_{k_1}$ and $d_{k_2}$ are situated in the same physical node, we simply set the $p_{S\overline{[K]\backslash S}}$ parameters to allow perfect correlation between the success events of $d_{k_1}$ and $d_{k_2}$. Throughout this paper, we consider memoryless 1-to-$K$ broadcast PECs that may or may not be spatially dependent.


\subsection{Broadcast PEC Capacity with Channel Output Feedback}

We consider the following broadcast scenario from $s$ to $\{d_k:k\in[K]\}$. Assume slotted transmission. Source $s$ is allowed to use the 1-to-$K$ PEC exactly $n$ times and would like to carry information for $K$ independent downlink data sessions, one for each $d_k$, respectively. For each $k\in[K]$, the $k$-th session (from $s$ to $d_k$) contains $nR_k$ information symbols $\X_k\stackrel{\Delta}{=}\{X_{k,j}\in\GF(q), j\in [nR_k]\}$, where $R_k$ is the data rate for the $(s,d_k)$ session. All the information symbols $X_{k,j}$ for all $k\in[K]$ and $j\in[nR_k]$ are independently and uniformly distributed in $\GF(q)$.

We consider the setting with instant channel output feedback (COF). That is, for the $t$-th time slot, source $s$ sends out a symbol
\begin{align}
Y(t)=
f_t\left(\{\X_k:\forall k\in [K]\}, \{\Z(\tau):\tau\in[t-1]\}\right),\nonumber
\end{align}
which is a function $f_t(\cdot)$ based on the information symbols $\{X_{k,j}\}$ and the COF $\{\Z(\tau):\tau\in[t-1]\}$ of the previous transmissions. In the end of the $n$-th time slot, each $d_k$ outputs the decoded symbols
\begin{align}
\hat{\X}_k\stackrel{\Delta}{=}\{\hat{X}_{k,j}:j\in[nR_k]\}=g_k(\{Z_{k}(t):\forall t\in[n]\}),\nonumber
\end{align}
where $g_k(\cdot)$ is the decoding function of $d_k$ based on the corresponding observation $Z_k(t)$ for $t\in[n]$. Note that we assume that the PEC channel parameters $\left\{p_{S\overline{[K]\backslash S}}:\forall S\in 2^{[K]}\right\}$ are available at $s$ before transmission. See Fig.~\ref{fig:B-PEC} for illustration.

\begin{figure}
\centering
\includegraphics[width=7cm]{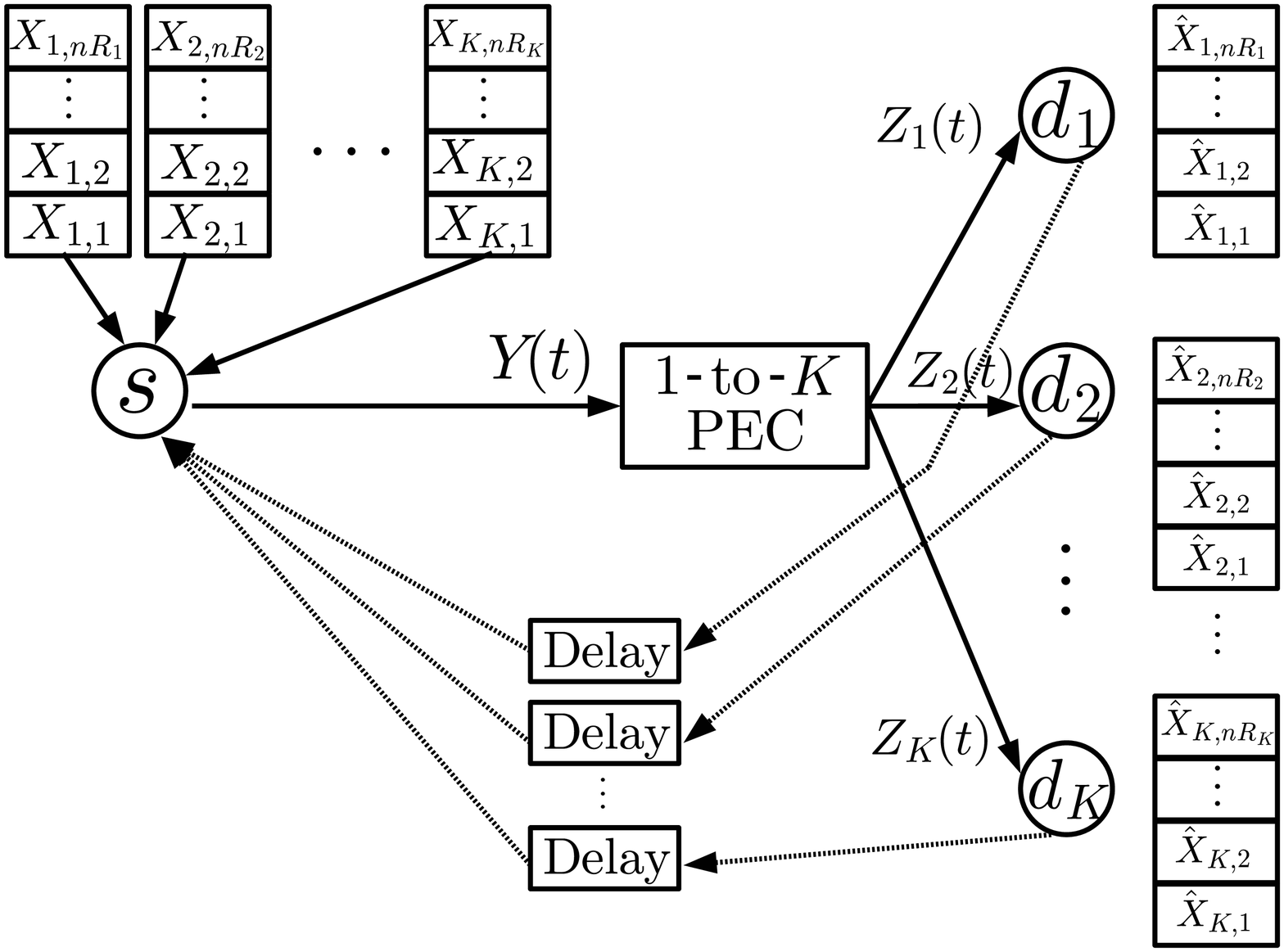}
\caption{Illustration of a 1-to-$K$ broadcast PEC with COF. \label{fig:B-PEC}}
\end{figure}

We now define the achievable rate of a 1-to-$K$ broadcast PEC with COF.

\begin{definition}A rate vector $(R_1,\cdots, R_K)$ is achievable if for any $\epsilon>0$, there exist sufficiently large $n$ and sufficiently large underlying finite field $\GF(q)$ such that
\begin{align}
\forall k\in[K],~\prop\left( \hat{\X}_k\neq \X_k\right)<\epsilon.\nonumber
\end{align}
\end{definition}
\begin{definition}The capacity region of a 1-to-$K$ broadcast PEC with COF is defined as the closure of all achievable rate vectors $(R_1,\cdots, R_K)$.
\end{definition}

\subsection{Existing Results\label{subsec:existing}}
The capacity of  1-to-2 broadcast PECs with COF has been characterized in \cite{GeorgiadisTassiulas09}:
\begin{theorem}[Theorem~3 in \cite{GeorgiadisTassiulas09}] The capacity region $(R_1,R_2)$ of a 1-to-2 broadcast PEC with COF is described by
\begin{align}
\begin{cases}\frac{R_1}{p_{1}}+\frac{R_2}{p_{\cup\{1,2\}}}\leq 1\\
\frac{R_1}{p_{\cup\{1,2\}}}+\frac{R_2}{p_{2}}\leq 1
\end{cases}.\label{eq:2cap}\end{align}
\end{theorem}

One scheme that achieves the above capacity region in \eqref{eq:2cap} is the 2-phase approach in \cite{GeorgiadisTassiulas09}. That is, for any $(R_1,R_2)$ in the interior of \eqref{eq:2cap}, perform the following coding operations.

In Phase~1, the source $s$ sends out uncoded information packets $X_{1,j_1}$ and $X_{2,j_2}$ for all $j_1\in[nR_1]$ and $j_2\in[nR_2]$ until each packet is received by at least one receiver. Those $X_{1,j_1}$ packets that are received by $d_1$ have already reached their intended receiver and thus will not be retransmitted in the second phase. Those $X_{1,j_1}$ packets that are received by $d_2$ but not by $d_1$ need to be retransmitted in the second phase, and are thus stored in a separate queue $Q_{1;2\overline{1}}$. Symmetrically, the $X_{2,j_2}$ packets that are received by $d_1$ but not by $d_2$ need to be retransmitted, and are stored in another queue $Q_{2;1\overline{2}}$.  Since those ``overheard" packets in queues $Q_{1;2\overline{1}}$ and $Q_{2;1\overline{2}}$ are perfect candidates for intersession network coding \cite{KattiRahulHuKatabiMedardCrowcroft06}, they can be linearly mixed together in Phase~2. Each single coded packet in Phase~2 can now serve both $d_1$ and $d_2$ simultaneously. The intersession network coding gain in Phase~2 allows us to achieve the capacity region in \eqref{eq:2cap}.

Based on the same logic, \cite{LarssonJohansson06} derives an achievability region for 1-to-$K$ broadcast PECs with COF under a {\em perfectly symmetric setting.}  The main idea can be viewed as an extension of the above 2-phase approach.  That is, for Phase~1, the source $s$ sends out all $X_{k,j}$, $\forall k\in[K], j\in[nR_k]$, until each of them is received by at least one of the receivers $\{d_k:k\in[K]\}$. Those $X_{k,j}$ packets that are received by $d_k$ have already reached their intended destination and will not be transmitted in Phase~2. Those $X_{k,j}$ packets that are received by some other $d_i$ but not by $d_k$ are the ``overheard packets," and could potentially be mixed with packets of the $i$-th session. In Phase~2, source $s$ takes advantage of all the coding opportunities created in Phase~1 and mixes the packets of different sessions to capitalize the network coding gain. \cite{RoznerIyerMehtaQiuJafry07} implements such 2-phase approach while taking into account of various practical considerations, such as time-out and network synchronization.

\subsection{The Suboptimality of The 2-Phase Approach\label{subsec:example}}

Although being throughput optimal for the simplest $K=2$ case, the above 2-phase approach does not achieve the capacity for the cases in which $K>2$. To illustrate this point, consider the example in Fig.~\ref{fig:example}.

\begin{figure}

\hspace{.5cm}\subfigure[Sending the first Phase-2 packet {$[X_1+X_2]$}.\label{subfig:example1}]{\includegraphics[height=2.75cm]{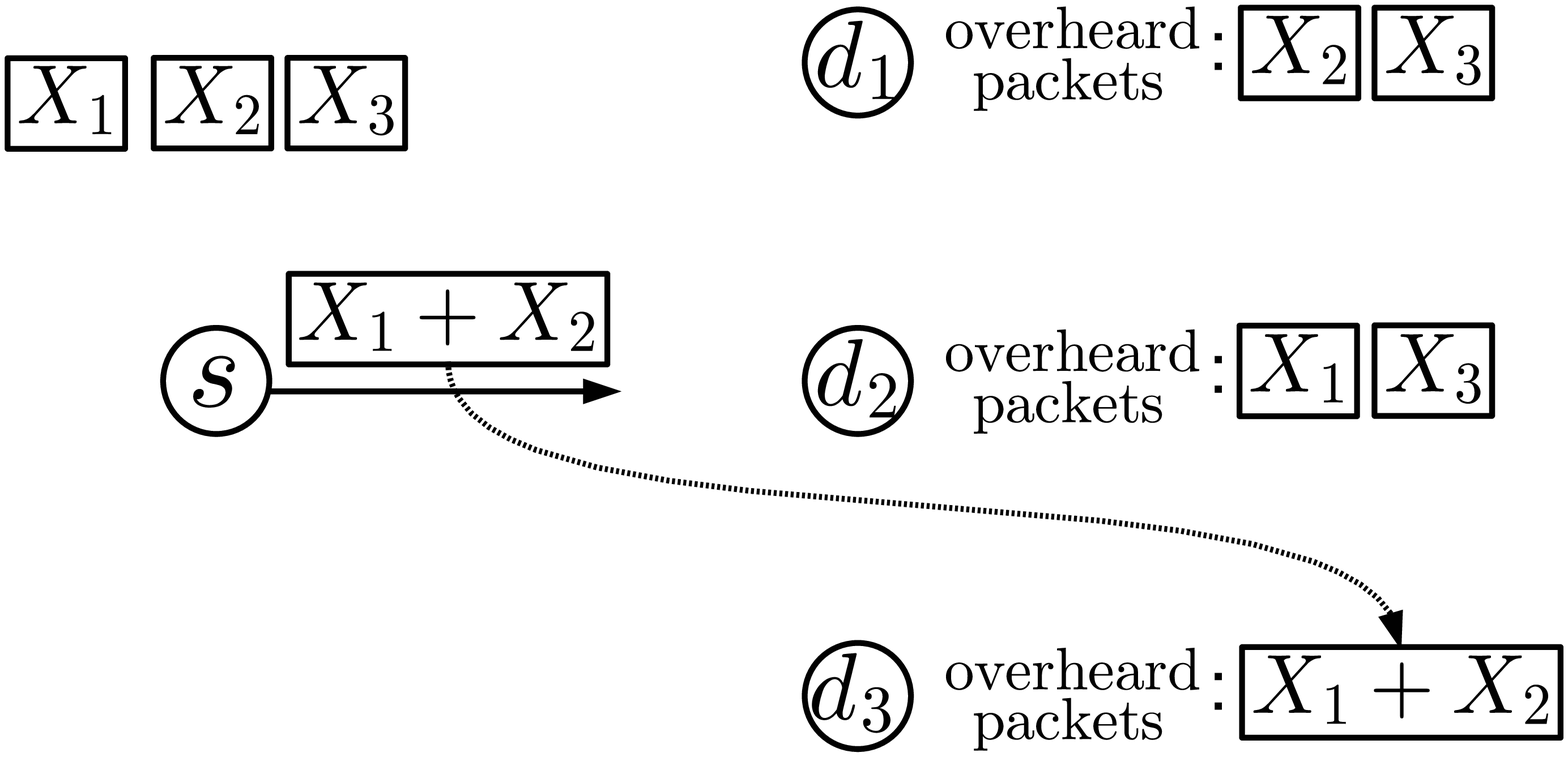}}

\vspace{.2cm}
\hspace{.5cm}\subfigure[The optimal coding operation after sending the {$[X_1+X_2]$}.\label{subfig:example2}]{\includegraphics[height=2.75cm]{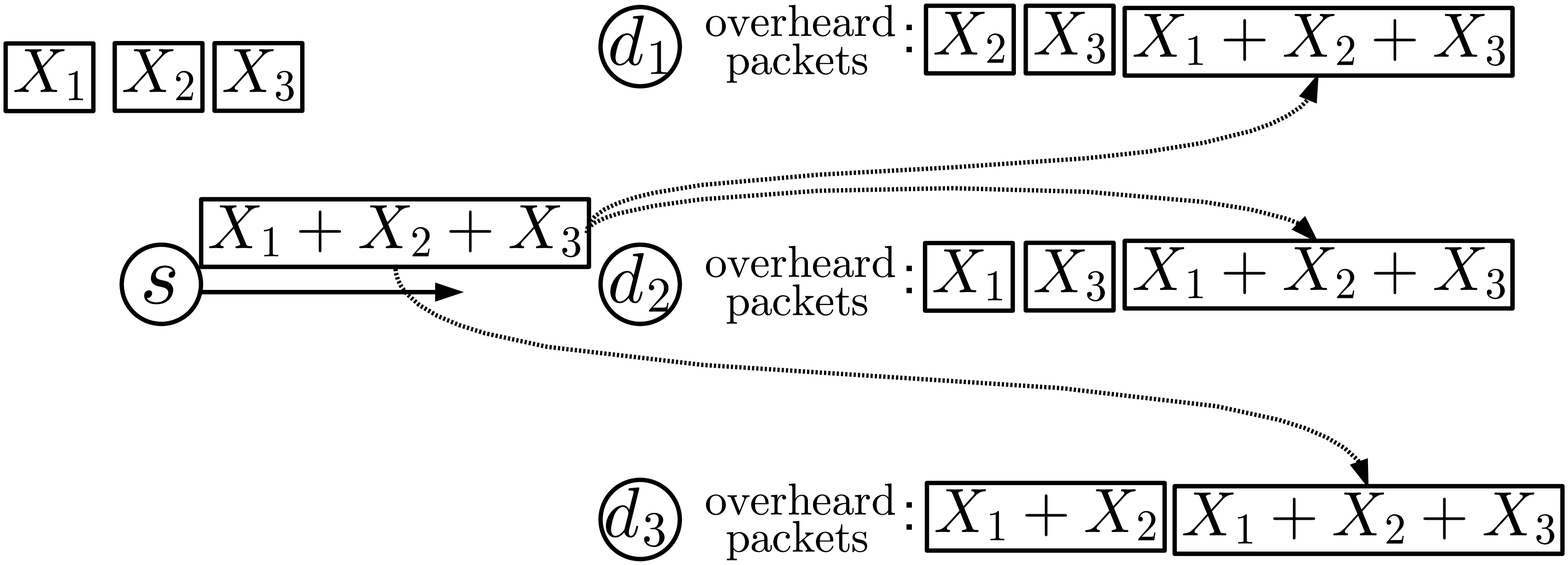}}

\caption{Example of the suboptimality of the 2-phase approach. \label{fig:example}}
\end{figure}

In Fig.~\ref{subfig:example1}, source $s$ would like to serve three receivers $d_1$ to $d_3$. Each $(s,d_k)$ session contains a single information packet $X_k$, and the goal is to convey each $X_k$ to the intended receiver $d_k$ for all $k=1,2,3$. Suppose the 2-phase approach in Section~\ref{subsec:existing} is used. During Phase~1, each packet is sent repeatedly until it is received by at least one receiver, which either conveys the packet to the intended receiver or creates an overheard packet that can be used in Phase~2. Suppose after Phase~1, $d_1$ has received $X_2$ and $X_3$, $d_2$ has received $X_1$ and $X_3$, and $d_3$ has not received any packet (Fig.~\ref{subfig:example1}). Since each packet has reached at least one receiver, source $s$ moves to Phase~2.

One can easily check that if $s$ sends out a coded packet $[X_1+X_2]$ in Phase~2, such packet can serve both $d_1$ and $d_2$. That is, $d_1$ (resp.\ $d_2$) can decode $X_1$ (resp.\ $X_2$)  by subtracting $X_2$ (resp.\ $X_1$) from $[X_1+X_2]$. Nonetheless, since the broadcast PEC is random, the coded packet $[X_1+X_2]$ may or may not reach $d_1$ or $d_2$. Suppose that due to random channel realization, $[X_1+X_2]$ reaches only $d_3$, see Fig.~\ref{subfig:example1}. The remaining question is what $s$ should send for the next time slot. For the following, we compare the existing 2-phase approach and a new optimal decision.

{\bf The existing 2-phase approach:} We first note that since $d_3$ received neither $X_1$ nor $X_2$ in the past, the newly received $[X_1+X_2]$ cannot be used by $d_3$ to decode any information packet. In the existing results \cite{LarssonJohansson06,RoznerIyerMehtaQiuJafry07,GeorgiadisTassiulas09}, $d_3$ thus discards the overheard $[X_1+X_2]$, and $s$ would continue sending $[X_1+X_2]$ for the next time slot in order to capitalize this coding opportunity created in Phase~1.

{\bf The optimal decision:} It turns out that the broadcast system can actually benefit from the fact that $d_3$ overhears the coded packet $[X_1+X_2]$ even though neither $X_1$ nor $X_2$ can be decoded by $d_3$. More explicitly, instead of sending $[X_1+X_2]$, $s$ should send a new packet $[X_1+X_2+X_3]$ that mixes all three sessions together. With the new $[X_1+X_2+X_3]$ (see Fig.~\ref{subfig:example2} for illustration), $d_1$ can decode the desired $X_1$ by subtracting both $X_2$ and $X_3$ from $[X_1+X_2+X_3]$. $d_2$ can decode the desired $X_2$ by subtracting both $X_1$ and $X_3$ from $[X_1+X_2+X_3]$. For $d_3$, even though $d_3$ does not know the values of  $X_1$ and $X_2$, $d_3$ can still use the previously overheard $[X_1+X_2]$ packet to subtract the interference $(X_1+X_2)$ from $[X_1+X_2+X_3]$ and decode its desired packet $X_3$. As a result, the new coded packet $[X_1+X_2+X_3]$ serves all destinations $d_1$, $d_2$, and $d_3$, simultaneously. This new coding decision thus strictly outperforms the existing 2-phase approach.

Two critical observations can be made for this example. First of all, when $d_3$ overhears a coded $[X_1+X_2]$ packet, even though $d_3$ can decode neither $X_1$ nor $X_2$, such new side information can still be used for future decoding. More explicitly, as long as $s$ sends packets that are of the form $\alpha(X_1+X_2)+\beta X_3$, the ``aligned interference" $\alpha(X_1+X_2)$ can be completely removed by $d_3$ without decoding individual $X_1$ and $X_2$. This technique is thus termed  ``{\em code alignment}," which is in parallel with the interference alignment method used in Gaussian interference channels \cite{CadambeJafar08}. Second of all, in the existing 2-phase approach, Phase~1 has the dual roles of sending uncoded packets to their intended receivers, and, at the same time,  creating new coding opportunities (the overheard packets) for Phase~2. It turns out that this dual-purpose Phase-1 operation is indeed optimal (as will be seen in Sections~\ref{sec:achievability} and~\ref{sec:pf-of-achievability}). The suboptimality of the 2-phase approach for $K>2$ is actually caused by the Phase-2 operation, in which source $s$ only capitalizes the coding opportunities created in Phase~1 but does not create any new coding opportunities for subsequent packet mixing. One can thus envision that for the cases $K>2$, an optimal policy should be a multi-phase policy, say an $M$-phase policy, such that for all $i\in[M-1]$ (not only for the first phase) the packets sent in the $i$-th phase have dual roles of  sending the information packets to their intended receivers and simultaneously creating new coding opportunities for the subsequent Phases $(i+1)$ to $M$. These two observations will be the building blocks of our achievability results. 


\section{The Main Results\label{sec:main}}

We have two groups of results: one is for general 1-to-$K$ broadcast PECs with arbitrary values of the PEC parameters, and the other is for 1-to-$K$ broadcast PECs with some restrictive conditions on the values of the PEC parameters.

\subsection{Capacity Results For General 1-to-$K$ Broadcast PECs\label{subsec:general}}

We define any bijective function $\pi:[K]\mapsto[K]$ as a $K$-permutation and we sometimes just say that $\pi$ is a permutation whenever it is clear from the context that we are focusing on $[K]$.  There are totally $K!$ distinct $K$-permutations. Given any $K$-permutation $\pi$, for all $j\in [K]$ we define $S^\pi_j\stackrel{\Delta}{=}\{\pi(l):\forall l\in[j]\}$ as the set of the first $j$ elements according to the permutation $\pi$. We then have the following capacity outer bound for any 1-to-$K$ broadcast PEC with COF.

\begin{proposition}\label{prop:outer}
Recall the definition of $p_{\cup S}$ in \eqref{eq:pS-def}. Any achievable rates  $(R_1,\cdots, R_K)$ must satisfy the following $K!$ inequalities:
\begin{align}
\forall \pi,~\sum_{j=1}^K\frac{R_{\pi(j)}}{p_{\cup S^\pi_j}}\leq 1.\label{eq:pi-outer}
\end{align}
\end{proposition}

\begin{figure}
\centering
\includegraphics[width=6cm]{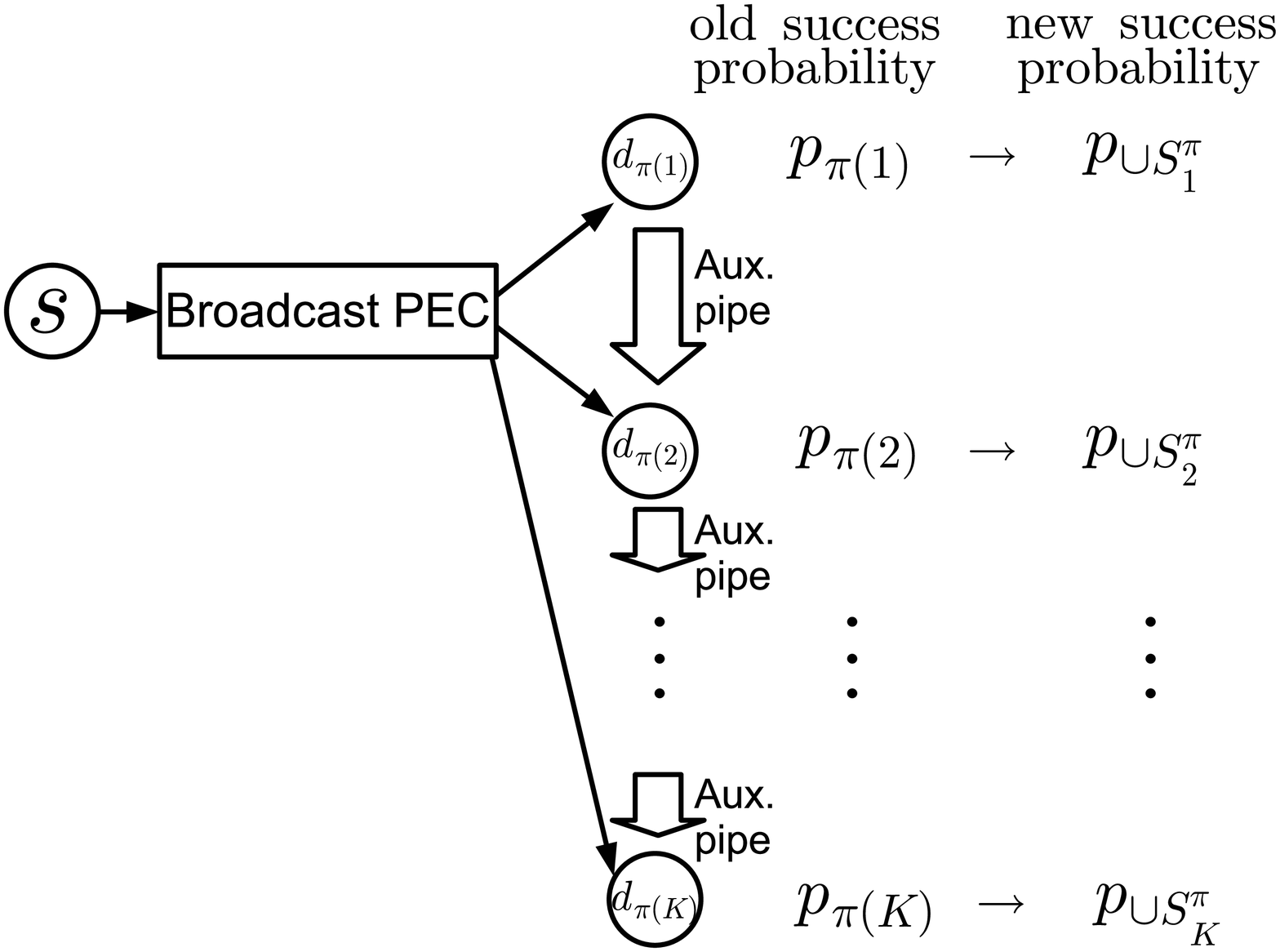}
\caption{Illustration of the proof of Proposition~\ref{prop:outer}. \label{fig:pipes}}
\end{figure}

\begin{thmproof}{Proof:} Proposition~\ref{prop:outer} can be proven by a simple extension of the outer bound arguments used in \cite{OzarowCheong84,GeorgiadisTassiulas09}. 
(Note that when $K=2$, Proposition~\ref{prop:outer} collapses to Theorem~3 of \cite{GeorgiadisTassiulas09}.)

For any given permutation $\pi$, consider a new broadcast channel with $(K-1)$ artificially created information pipes connecting all the receivers $d_1$ to $d_K$. More explicitly, for all $j\in[K-1]$, create an auxiliary pipe from $d_{\pi(j)}$ to $d_{\pi(j+1)}$. See Fig.~\ref{fig:pipes} for illustration. With the auxiliary pipes, any destination $d_{\pi(j)}$, $j\in[K]$, not only observes the corresponding output $Z_{\pi(j)}$ of the broadcast PEC but also has all the information $Z_{\pi(l)}$ of its ``upstream receivers" $d_{\pi(l)}$ for all $l\in [j-1]$. Since we only create new pipes, any achievable rates of the original 1-to-$K$ broadcast PEC with COF must also be achievable in the new 1-to-$K$ broadcast PEC with COF in Fig.~\ref{fig:pipes}. The capacity of the new 1-to-$K$ broadcast PEC with COF is thus an outer bound on the capacity of the original 1-to-$K$ broadcast PEC with COF.

 On the other hand, the new 1-to-$K$ broadcast PEC in  Fig.~\ref{fig:pipes} is a physically degraded broadcast channel with the new success probability of $d_k$ being $p_{\cup S_k^\pi}$ instead of $p_{\pi(k)}$ (see Fig.~\ref{fig:pipes}). \cite{Gamal78} shows that COF does not increase the capacity of any physically degraded broadcast channel. Therefore the capacity of the new 1-to-$K$ broadcast PEC with COF is identical to the capacity of the new 1-to-$K$ broadcast PEC without COF. Since \eqref{eq:pi-outer} is the capacity of the new 1-to-$K$ broadcast PEC without COF, \eqref{eq:pi-outer} must be an outer bound of the capacity of the original 1-to-$K$ PEC with COF. By considering different permutation $\pi$, the proof of Proposition~\ref{prop:outer} is complete.
\end{thmproof}

For the following, we first provide the capacity results for general 1-to-3 broadcast PECs. We then state an achievability inner bound for general 1-to-$K$ broadcast PECs with COF for arbitrary $K$ values, which, together with the outer bound in Proposition~\ref{prop:outer} can effectively bracket the capacities for the cases in which $K\geq 4$.

\begin{proposition}\label{prop:cap3}
For any parameter values $\left\{p_{S\overline{\{1,2,3\}\backslash S}}:\forall S\in 2^{\{1,2,3\}}\right\}$ of a 1-to-3 broadcast PEC, the capacity outer bound in Proposition~\ref{prop:outer} is indeed the capacity region of a 1-to-3 broadcast PEC with COF.
\end{proposition}

To state the capacity inner bound, we need to define an additional function: $f_p(S\overline{T})$, which takes an input $S\overline{T}$ of two disjoint sets $S,T\in 2^{[K]}$. More explicitly, we define $f_p(S\overline{T})$ as the probability that a packet $Y$, transmitted through the 1-to-$K$ PEC, is received by all those $d_i$ with $i\in S$ and not received by any $d_j$ with $j\in T$. For example, $f_p(S\overline{[K]\backslash S})=p_{S\overline{[K]\backslash S}}$ for all $S\in 2^{[K]}$. For arbitrary disjoint $S$ and $T$, we thus have
\begin{align}
&f_p(S\overline{T})\stackrel{\Delta}{=}\sum_{\forall S_1: S\subseteq S_1, T\subseteq ([K]\backslash S_1)}p_{S_1\overline{[K]\backslash S_1}}.\label{eq:fp-def}
\end{align}
We also say that a {strict total ordering} ``$\prec$" on $2^{[K]}$ is {\em cardinality-compatible} if
\begin{align}
\forall S_1, S_2\in 2^{[K]},\quad |S_1|<|S_2|\Rightarrow S_1\prec S_2.\label{eq:card-ordering}
\end{align}
For example, for $K=3$, the following strict total ordering
\begin{align}
\emptyset\prec \{2\}\prec\{1\}\prec\{3\}\prec \{1,2\}\prec\{1,3\}\prec\{2,3\}\prec\{1,2,3\}\nonumber
\end{align}
is cardinality-compatible. 

\begin{proposition}\label{prop:ach2} Fix any arbitrary  cardinality-compatible, strict total ordering $\prec$.
For any general 1-to-$K$ broadcast PEC with COF, a rate vector $(R_1,\cdots, R_K)$ can be achieved by a {\em linear network code}
if
there exist $2^K$ non-negative $x$ variables, indexed by $S\in 2^{[K]}$:
\begin{align}
\left\{x_S\geq 0:\forall S\in 2^{[K]}\right\},\label{eq:xs}
\end{align}
and $K3^{K-1}$ non-negative $w$ variables, indexed by $(k;S\rightarrow T)$ satisfying $T\subseteq S\subseteq ([K]\backslash k)$:
\begin{align}&\left\{w_{k;S\rightarrow T}\geq 0: \forall k\in[K],\forall S,T\in 2^{[K]},\right.\nonumber\\
&\hspace{3cm}\left.\text{satisfying }T\subseteq S\subseteq ([K]\backslash k)\right\},\label{eq:ws}
\end{align}
such that jointly the following linear inequalities\footnote{
There are totally $(1+K2^{K-1}+K3^{K-1})$ inequalities. More explicitly, \eqref{eq:total-x} describes one inequality. There are $K2^{K-1}$ inequalities having the form of \eqref{eq:coding-len}. There are totally $K3^{K-1}$ inequalities having the form of one of \eqref{eq:ind-length-0}, \eqref{eq:ind-length-1}, and \eqref{eq:ind-length-2}. For comparison, the outer bound in Proposition~\ref{prop:outer} actually has more inequalities asymptotically ($K!$ of them) than those in Proposition~\ref{prop:ach2}.} are satisfied:
\begin{align}
&\sum_{\forall S:S\in 2^{[K]}}x_S< 1\label{eq:total-x}\\
&\forall T\in 2^{[K]},\forall k\in T,\nonumber\\
&\hspace{2.5cm} x_T\geq \sum_{\forall S:(T\backslash k)\subseteq S\subseteq([K]\backslash k)}w_{k;S\rightarrow (T\backslash k)}\label{eq:coding-len}\\
&\forall k\in[K],\quad w_{k;\emptyset\rightarrow \emptyset}\cdot p_{\cup [K]}\geq R_k\label{eq:ind-length-0}\\
&\forall k\in[K], \forall S\subseteq ([K]\backslash k), S\neq \emptyset, \nonumber\\
&\hspace{0cm}
\left(\sum_{\forall T_1: T_1\subseteq S} w_{k;S\rightarrow T_1}\right) p_{\cup ([K]\backslash S)}\geq \nonumber\\
&\hspace{.7cm}\sum_{\scriptsize \begin{array}{c}\forall S_1,T_1:\text{such that}\\
T_1\subseteq S_1\subseteq ([K]\backslash k),\\
T_1\subseteq S,S\nsubseteq S_1 \end{array}}w_{k;S_1\rightarrow T_1}\cdot f_p\left((S\backslash T_1)\overline{([K]\backslash S)}\right)\label{eq:ind-length-1}\\
&\forall k\in[K], S,T\in 2^{[K]} \text{ satisfying } T\subseteq S\subseteq ([K]\backslash k), T\neq S,\nonumber\\
&\hspace{0cm}\left(w_{k;S\rightarrow T}+\sum_{\scriptsize\begin{array}{c}\forall T_1\subseteq S:\\
(T_1\cup\{k\})\prec (T\cup\{k\})\end{array}}w_{k; S\rightarrow T_1}\right)p_{\cup ([K]\backslash S)} \leq\nonumber\\
&\hspace{.5cm}
\sum_{\scriptsize\begin{array}{c}\forall S_1: S_1\prec S,\\
T\subseteq S_1\subseteq ([K]\backslash k)\end{array}}w_{k;S_1\rightarrow T} \cdot f_p\left((S\backslash T)\overline{([K]\backslash S)}\right)+\nonumber\\
&\hspace{.7cm}\sum_{\scriptsize \begin{array}{c}\forall S_1,T_1:\text{such that}\\
T_1\subseteq S_1\subseteq ([K]\backslash k),\\
(T_1\cup\{k\})\prec (T\cup\{k\}),\\
T_1\subseteq S,S\nsubseteq S_1 \end{array}}w_{k;S_1\rightarrow T_1}\cdot f_p\left((S\backslash T_1)\overline{([K]\backslash S)}\right).\label{eq:ind-length-2}
\end{align}
\end{proposition}

Since Proposition~\ref{prop:ach2} holds for any cardinality-compatible, strict total ordering $\prec$. We can easily derive the following corollary:

To distinguish different strict total orderings, we append a subscript $l$ to $\prec$. For example, $\prec_1$ and $\prec_2$ correspond to two distinct strict total orderings. Overall, there are $L\stackrel{\Delta}{=}\prod_{k=0}^K\left({K\choose k}!\right)$ distinct strict total ordering $\prec_l$, $\forall l\in [L]$, that are cardinality-compatible.
\begin{corollary}
For any given cardinality-compatible strict total ordering $\prec_l$, we use $\Lambda_l$ to denote the collection of all $(R_1,\cdots, R_K)$ rate vectors satisfying Proposition~\ref{prop:ach2}. Then the convex hull of $\Co\left(\{\Lambda_l:\forall l\in[L]\}\right)$ is an achievable region of the given 1-to-$K$ broadcast PEC with COF.
\end{corollary}

{\em Remark:} For some general classes of PEC parameters, one can prove that the inner bound of Proposition~\ref{prop:ach2} is indeed the capacity region for arbitrary $K\geq 4$ values. Two such classes are discussed in the next subsection.

\subsection{Capacity Results For Two Classes of 1-to-$K$ Broadcast PECs\label{subsec:special}}

We first provide the capacity results for {\em symmetric} broadcast PECs.
\begin{definition}
A 1-to-$K$ broadcast PEC is {\em symmetric} if the channel parameters $\left\{p_{S\overline{[K]\backslash S}}:\forall S\in 2^{[K]}\right\}$ satisfy
\begin{align}
\forall S_1,S_2\in 2^{[K]}\text{ with }|S_1|=|S_2|,~p_{S_1\overline{[K]\backslash S_1}}=p_{S_2\overline{[K]\backslash S_2}}.\nonumber
\end{align}
That is, the success probability $p_{S\overline{[K]\backslash S}}$ depends only on $|S|$, the size of $S$, and does not depend on which subset of receivers being considered.
\end{definition}

\begin{proposition}\label{prop:cap-sym}
For any symmetric 1-to-$K$ broadcast PEC with COF, the  capacity outer bound in Proposition~\ref{prop:outer} is indeed the corresponding capacity region.
\end{proposition}

The perfect channel symmetry condition in Proposition~\ref{prop:cap-sym}  may be a bit restrictive for real environments as most broadcast channels are non-symmetric. A more realistic setting is to allow channel asymmetry while assuming spatial independence between different destinations $d_i$. 

\begin{definition}
A 1-to-$K$ broadcast PEC is {\em spatially independent} if the channel parameters $\left\{p_{S\overline{[K]\backslash S}}:\forall S\in 2^{[K]}\right\}$ satisfy
\begin{align}
\forall S\in2^{[K]},~ p_{S\overline{[K]\backslash S}}=\left(\prod_{k\in S}p_k\right)\left(\prod_{k\in [K]\backslash S} (1-p_k)\right),\nonumber
\end{align}
where $p_k$ is the marginal success probability of destination $d_k$.
\end{definition}

{\em Note:} A symmetric 1-to-$K$ broadcast PEC needs not be spatially independent. A spatially independent PEC is symmetric if $p_1=p_2=\cdots=p_K$.

To describe the capacity results for spatially independent 1-to-$K$ PECs, we need the following additional definition.

\begin{definition}
Consider a 1-to-$K$ broadcast PEC with marginal success probabilities $p_1$ to $p_K$. Without loss of generality, assume $p_1\leq p_2\leq \cdots\leq p_K$, which can be achieved by relabeling. We say a rate vector $(R_1,\cdots, R_K)$ is {\em one-sidedly fair} if
\begin{align}
\forall i<j,~R_i(1-p_i)\geq R_j(1-p_j). \nonumber
\end{align}
We use $\Lambda_{\text{osf}}$ to denote the collection of all one-sidedly fair rate vectors.
\end{definition}

The one-sided fairness contains many practical scenarios of interest. For example, the perfectly fair rate vector $(R,R,\cdots, R)$ by definition is also one-sidedly fair. Another example is when $\min(p_1,\cdots, p_K)>\frac{1}{2}$ and we allow the rate $R_k$ to be proportional to the corresponding marginal success probability $p_k$, i.e., $R_k=p_k R$,  then the rate vector
$(p_1R, p_2R,\cdots, p_KR)$ is also one-sidedly fair.

For the following, we provide the capacity of spatially independent 1-to-$K$ PECs with COF under the condition of one-sided fairness.

\begin{proposition}\label{prop:cap-osf}
Suppose the 1-to-$K$ PEC of interest is spatially independent with marginal success probabilities  $0<p_1\leq p_2\leq \cdots\leq p_K$. Any one-sidedly fair rate vector $(R_1,\cdots, R_K)\in \Lambda_\text{osf}$ is in the capacity region if and only if $(R_1,\cdots, R_K)\in \Lambda_\text{osf}$ satisfies
\begin{align}
\sum_{k=1}^K\frac{R_k}{1-\prod_{l=1}^k(1-p_l)}\leq  1.\label{eq:osf-cap}
\end{align}
\end{proposition}

Proposition~\ref{prop:cap-osf} implies that Proposition~\ref{prop:outer} is indeed the capacity region when focusing on the one-sidedly fair rate region $\Lambda_\text{osf}$.


\section{The Packet Evolution Schemes\label{sec:achievability}}

For the following, we describe a new class of coding schemes, termed the {\em packet evolution} (PE) scheme, which embodies the concept of code alignment and achieves (near) optimal throughput. The PE scheme is the building block of the capacity /  achievability results in Section~\ref{sec:main}.


\subsection{Description Of The Packet Evolution Scheme\label{subsec:PE}}
The packet evolution scheme is described as follows. Recall that each $(s,d_k)$ session has $nR_k$ information packets $X_{k,1}$ to $X_{k,nR_k}$. We associate each of the $\sum_{k=1}^KnR_k$ information packets  with {\em an intersession coding vector} $\vv$ and a set $S\subseteq [K]$. An intersession coding vector is a $\left(\sum_{k=1}^KnR_k\right)$-dimensional row vector with each coordinate being a scalar in $\GF(q)$. Before the start of the broadcast, for any $k\in[K]$ and $j\in[nR_k]$ we initialize the corresponding vector $\vv$ of $X_{k,j}$ in a way that the only nonzero coordinate of $\vv$ is the coordinate corresponding to $X_{k,j}$ and all other coordinates are zero. Without loss of generality, we set the value of the only non-zero coordinate to one. That is, initially the coding vectors $\vv$ are set to the elementary basis vectors of the entire $\left(\sum_{k=1}^KnR_k\right)$-dimensional message space.

 For any $k\in[K]$ and $j\in[nR_k]$ the set $S$ of $X_{k,j}$ is initialized to $\emptyset$. As will be clear shortly after, we call  $S$ the {\em overhearing set}\footnote{Unlike the existing results \cite{KattiRahulHuKatabiMedardCrowcroft06}, in this work the overhearing set does not mean that the receivers $d_i$ in $S(X_{k,j})$ have known the value of $X_{k,j}$. Detailed discussion of the overhearing set $S(X_{k,j})$ are provided in Lemma~\ref{lem:non-interfering}.} of the packet $X_{k,j}$. For easier reference, we use $\vv(X_{k,j})$ and $S(X_{k,j})$ to denote the intersession coding vector and the overhearing set of $X_{k,j}$.

Throughout the $n$ broadcast time slots, source $s$ constantly updates the $S(X_{k,j})$ and $\vv(X_{k,j})$ according to the COF. The main structure of a packet evolution scheme can now be described as follows.

\algtop{The Packet Evolution Scheme}
\begin{algorithmic}[1]

\STATE Source $s$ maintains a single flag ${\mathsf{f}}_{\text{change}}$. Initially, set ${\mathsf{f}}_{\text{change}}\leftarrow 1$.
\FOR{$t=1,\cdots, n$, }
\STATE In the beginning of the $t$-th time slot, do Lines~\ref{line:begin1} to~\ref{line:transmit-v}.
\IF{${\mathsf{f}}_{\text{change}}= 1$\label{line:begin1}}
\STATE Choose a non-empty subset $T\subseteq [K]$.
\STATE Run a subroutine {\sc Packet Selection}, which takes $T$ as input and outputs a collection of $|T|$ packets $\{X_{k,j_k}: \forall k\in T\}$, termed the {\em target packets}, for which all $X_{k,j_k}$ satisfy $(S(X_{k,j_k})\cup \{k\})\supseteq T$.\label{line:targeting-set}

\STATE Generate $|T|$ uniformly random coefficients $c_k\in \GF(q)$ for all $k\in T$ and construct an intersession coding vector $\vv_\text{tx}\leftarrow \sum_{k\in T}c_k\cdot \vv(X_{k,j_k})$.\label{line:vv-construct}

\STATE Set ${\mathsf{f}}_{\text{change}}\leftarrow 0$.
\ENDIF

\STATE Sends out a linearly intersession coded packet according to the coding vector $\vv_\text{tx}$. That is, we send
\begin{align}Y_\text{tx}=\vv_\text{tx}\cdot (X_{1,1},\cdots, X_{K,nR_K})^\tran\nonumber
\end{align}  where $(X_{1,1},\cdots, X_{K,nR_K})^\tran$ is a column vector consisting of all information symbols.\footnote{It is critical to note that the coding operation is based purely on $\vv_\text{tx}$ rather than on the list of the target packets $X_{k,j_k}$. Once $\vv_\text{tx}$ is decided, we create a new coded packet based on the coordinates of $\vv_\text{tx}$. It is possible that
$\vv_\text{tx}$ has non-zero coordinates corresponding to some $X_{k',j}$ that are not one of the target packets $X_{k,j_k}$. Those $X_{k',j}$ will participate in creating the coded packet.}
\label{line:transmit-v}



\STATE In the end of the $t$-th time slot, use a subroutine {\sc Update} to revise the $\vv(X_{k,j_k})$ and $S(X_{k,j_k})$ values of all target packets $X_{k,j_k}$ based on the COF.

\IF{the $S(X_{k,j_k})$ value changes for at least one target packet $X_{k,j_k}$ after the {\sc Update}}
\STATE Set ${\mathsf{f}}_{\text{change}}\leftarrow 1$.
\ENDIF

\ENDFOR

\end{algorithmic} \algbot

In summary, a group of target packets $\{X_{k,j_k}\}$ are selected according to the choice of the subset $T$. The corresponding vectors $\{\vv(X_{k,j_k})\}$ are used to construct a coding vector $\vv_\text{tx}$. The same coded packet $Y_\text{tx}$, corresponding to $\vv_\text{tx}$, is then sent repeatedly for many time slots until one of the target packets $X_{k,j_k}$ {\em evolves} (when the corresponding $S(X_{k,j_k})$ changes). Then a new subset $T$ is chosen and the process is repeated until we use up all $n$ time slots. Three subroutines are used as the building blocks of a packet evolution method: (i) How to choose the non-empty $T\subseteq [K]$; (ii) For each $k\in[K]$, how to select a single target packets $X_{k,j_k}$ among all $X_{k,j}$ satisfying $(S(X_{k,j})\cup \{k\})\supseteq T$; and (iii) How to update the coding vectors $\vv(X_{k,j_k})$ and the overhearing sets $S(X_{k,j_k})$. For the following, we first describe the detailed update rules.

\algtop{Update of $S(X_{k,j_k})$ and $\vv(X_{k,j_k})$}
\begin{algorithmic}[1]
\STATE {\bf Input:} The $T$ and $\vv_\text{tx}$ used for transmission in the current time slot; And $S_\text{rx}$, the set of destinations $d_i$ which receive the transmitted coded packet in the current time slot. ($S_\text{rx}$ is obtained through the COF in the end of the current time slot.)
\FOR{all $k\in T$}\label{line:update-4-K}
\IF{$S_\text{rx}\nsubseteq S(X_{k,j_k})$}
\STATE Set $S(X_{k,j_k})\leftarrow (T\cap S(X_{k,j_k}))\cup S_\text{rx}$.\label{line:S-update}
 \STATE Set $\vv(X_{k,j_k})\leftarrow \vv_\text{tx}$.\label{line:v-update}
\ENDIF

\ENDFOR

\end{algorithmic} \algbot

\vspace{.3cm}
\noindent {\em An Illustrative Example Of The PE Scheme:}

\vspace{.2cm}
Let us revisit the optimal coding scheme of the example in Fig.~\ref{fig:example} of Section~\ref{subsec:example}. Before broadcast, the three information packets $X_{1}$ to $X_3$ have the corresponding $\vv$ and $S$: $\vv(X_1)=(1,0,0)$, $\vv(X_2)=(0,1,0)$, and $\vv(X_3)=(0,0,1)$, and $S(X_1)=S(X_2)=S(X_3)=\emptyset$. We use the following table
for summary.

\begin{center}\begin{tabular}{|c|c|c|}
\hline
$X_1$: (1,0,0),$\emptyset$ &$X_2$: (0,1,0),$\emptyset$ & $X_3$: (0,0,1),$\emptyset$ \\
\hline
\end{tabular}
\end{center}

Consider a duration of 5 time slots.

Slot 1: Suppose that $s$ chooses $T=\{1\}$. Since $(\emptyset\cup \{1\})\supseteq T$, {\sc Packet Selection} outputs $X_1$. The coding vector $\vv_\text{tx}$ is thus a scaled version of $\vv(X_1)=(1,0,0)$. Without loss of generality, we choose $\vv_\text{tx}=(1,0,0)$. Based on $\vv_\text{tx}$, $s$ transmits a packet $1 X_1+0 X_2+0 X_3=X_1$. Suppose $[X_1]$ is received by $d_2$, i.e., $S_\text{rx}=\{2\}$. Then during {\sc Update}, $S_\text{rx}=\{2\}\nsubseteq S(X_1)=\emptyset$. {\sc Update} thus sets $S(X_1)=\{2\}$ and $\vv(X_1)=\vv_\text{tx}=(1,0,0)$. The packet summary becomes

\begin{center}\begin{tabular}{|c|c|c|}
\hline
$X_1$: (1,0,0),$\{2\}$ &$X_2$: (0,1,0),$\emptyset$ & $X_3$: (0,0,1),$\emptyset$ \\
\hline
\end{tabular}.
\end{center}

Slot 2: Suppose that $s$ chooses $T=\{2\}$. Since $(\emptyset\cup \{2\})\supseteq T$,  {\sc Packet Selection} outputs $X_2$. The coding vector $\vv_\text{tx}$ is thus a scaled version of $\vv(X_2)=(0,1,0)$. Without loss of generality, we choose $\vv_\text{tx}=(0,1,0)$ and accordingly $[X_2]$ is sent. Suppose $[X_2]$ is received by $d_1$, i.e., $S_\text{rx}=\{1\}$. Since $S_\text{rx}\nsubseteq S(X_2)$, after {\sc Update} the packet summary becomes

\begin{center}\begin{tabular}{|c|c|c|}
\hline
$X_1$: (1,0,0),$\{2\}$ &$X_2$: (0,1,0),$\{1\}$ & $X_3$: (0,0,1),$\emptyset$ \\
\hline
\end{tabular}.
\end{center}

Slot 3: Suppose that $s$ chooses $T=\{3\}$ and {\sc Packet Selection} outputs $X_3$. The coding vector $\vv_\text{tx}$ is thus a scaled version of $\vv(X_3)=(0,0,1)$, and we choose $\vv_\text{tx}=(0,0,1)$. Accordingly $[X_3]$ is sent. Suppose $[X_3]$ is received by $d_1$ and $d_2$, i.e., $S_\text{rx}=\{1,2\}$. Then after {\sc Update}, the packet summary becomes

\begin{center}\begin{tabular}{|c|c|c|}
\hline
$X_1$: (1,0,0),$\{2\}$ &$X_2$: (0,1,0),$\{1\}$ & $X_3$: (0,0,1),$\{1,2\}$ \\
\hline
\end{tabular}.
\end{center}

Slot 4: Suppose that $s$ chooses $T=\{1,2\}$. Since $(S(X_1)\cup \{1\})\supseteq T$ and $(S(X_2)\cup \{2\})\supseteq T$, {\sc Packet Selection} outputs $\{X_1,X_2\}$. $\vv_\text{tx}$ is thus a linear combination of $\vv(X_1)=(1,0,0)$ and $\vv(X_2)=(0,1,0)$. Without loss of generality, we choose $\vv_\text{tx}=(1,1,0)$ and accordingly $[X_1+X_2]$ is sent. Suppose $[X_1+X_2]$ is received by $d_3$, i.e., $S_\text{rx}=\{3\}$. Then during {\sc Update}, for $X_1$,  $S_\text{rx}=\{3\}\nsubseteq S(X_1)=\{2\}$. {\sc Update} thus sets $S(X_1)=\{2,3\}$ and $\vv(X_1)=\vv_\text{tx}=(1,1,0)$. For $X_2$, $S_\text{rx}=\{3\} \nsubseteq S(X_2)=\{1\}$. {\sc Update} thus sets $S(X_2)=\{1,3\}$ and $\vv(X_2)=\vv_\text{tx}=(1,1,0)$.  The packet summary becomes


\begin{center}\begin{tabular}{|c|c|}
\hline
$X_1$: (1,1,0),$\{2,3\}$ &$X_2$: (1,1,0),$\{1,3\}$\\
\hline
 $X_3$: (0,0,1),$\{1,2\}$  &~ \\
\hline
\end{tabular}.
\end{center}


Slot 5: Suppose that $s$ chooses $T=\{1,2,3\}$. By Line~\ref{line:targeting-set} of {\sc The Packet Evolution Scheme}, the subroutine {\sc Packet Selection} outputs $\{X_1,X_2,X_3\}$. $\vv_\text{tx}$ is thus a linear combination of $\vv(X_1)=(1,1,0)$, $\vv(X_2)=(1,1,0)$, and $\vv(X_3)=(0,0,1)$, which is of the form $\alpha(X_1+X_2)+\beta X_3$. {\em Note that the packet evolution scheme automatically achieves code alignment}, which is the key component of the optimal coding policy in Section~\ref{subsec:example}. Without loss of generality, we choose $\alpha=\beta=1$ and $\vv_\text{tx}=(1,1,1)$. $Y_\text{tx}=[X_1+X_2+X_3]$ is sent accordingly. Suppose $[X_1+X_2+X_3]$ is received by $\{d_1,d_2,d_3\}$, i.e., $S_\text{rx}=\{1,2,3\}$. Then after {\sc Update}, the summary of the packets becomes
\begin{center}\begin{tabular}{|c|c|}
\hline
$X_1$: (1,1,1),$\{1,2,3\}$ &$X_2$: (1,1,1),$\{1,2,3\}$\\
\hline
  $X_3$: (1,1,1),$\{1,2,3\}$ &~ \\
\hline
\end{tabular}.
\end{center}

From the above step-by-step illustration, we see that the optimal coding policy in Section~\ref{subsec:example} is a special case of a packet evolution scheme.

\subsection{Properties of A Packet Evolution Scheme}

We term the packet evolution (PE) scheme in Section~\ref{subsec:PE} a {\em generic} PE method since it does not depend on how to choose $T$ and the target packets $X_{k,j_k}$ and only requires the output of {\sc Packet Selection} satisfying
$(S(X_{k,j_k})\cup \{k\})\supseteq T, \forall k\in T$. In this subsection, we state some key properties for any generic PE scheme. The intuition of the PE scheme is based on these key properties and will be discussed further in Section~\ref{subsec:intuition-PE}.

We first define the following notation for any linear network codes. (Note that the PE scheme is a linear network code.)


\begin{definition}
Consider any linear network code. For any destination $d_k$, each of the received packet $Z_k(t)$ can be represented by a vector $\w_k(t)$, which is a $\left(\sum_{k=1}^KnR_k\right)$-dimensional vector containing the coefficients used to generate $Z_k(t)$. That is, $Z_k(t)=\w_k(t)\cdot (X_{1,1},\cdots, X_{K,nR_K})^\tran$. If $Z_k(t)$ is an erasure, we simply set $\w_k(t)$ to be an all-zero vector. The {\em knowledge space} of destination $d_k$ in the end of time $t$ is denoted by $\Omega_{\text{Z},k}(t)$, which is the linear span of $\w_k(\tau)$, $\tau\leq t$. That is, $\Omega_{\text{Z},k}(t)\stackrel{\Delta}{=}\linsp(\w_k(\tau):\forall \tau\in[t])$.
\end{definition}

\begin{definition} For any non-coded information packet $X_{k,j}$, the corresponding intersession coding vector is a $\left(\sum_{k=1}^KnR_k\right)$-dimensional vector with a single one in the corresponding coordinate and all other coordinates being zero. We use $\delta_{k,j}$ to denote such a delta vector. The message space of $d_k$ is then defined as $\Omega_{M,k}=\linsp(\delta_{k,j}:\forall j\in[nR_k])$.
\end{definition}

With the above definitions, we have the following straightforward lemma:
\begin{lemma}\label{lem:simple-dec}
In the end of time $t$, destination $d_k$ is able to decode all the desired information packets $X_{k,j}$, $\forall j\in [nR_k]$, if and only if $\Omega_{M,k}\subseteq \Omega_{Z,k}(t)$.
\end{lemma}

We now define ``non-interfering vectors" from the perspective of a destination $d_k$.
\begin{definition}
In the end of time $t$ (or in the beginning of time $(t+1)$), a vector $\vv$ (and thus the corresponding coded packet) is ``non-interfering" from the perspective of $d_k$ if
\begin{align}
\vv\in \linsp(\Omega_{Z,k}(t),\Omega_{M,k}).\nonumber
\end{align}
\end{definition}

We note that any non-interfering vector $\vv$ can always be expressed as the sum of two vectors $\vv'$ and $\w$, where $\vv'\in \Omega_{M,k}$ is a linear combination of all information vectors for $d_k$ and $\w\in \Omega_{Z,k}(t)$ is a linear combination of all the packets received by $d_k$. If $\vv'=0$, then $\vv=\w$ is a {\em transparent} packet from $d_k$'s perspective since $d_k$ can compute the value of $\w\cdot (X_{1,1},\cdots, X_{K,nR_K})^\tran$ from its current knowledge space $\Omega_{Z,k}(t)$.
If $\vv'\neq 0$, then $\vv=\vv'+\w$ can be viewed as a pure information packet $\vv'\in \Omega_{M,k}$ after subtracting the unwanted $\w$ vector. In either case, $\vv$ is {\em not interfering} with the transmission of the $(s,d_k)$ session, which gives the name of ``non-interfering vectors."

The following Lemmas~\ref{lem:non-interfering} and \ref{lem:decodability} discuss the time dynamics of the PE scheme. To distinguish different time instants, we add a time subscript and use $S_{t-1}(X_{k,j_k})$ and $S_{t}(X_{k,j_k})$ to denote the overhearing set of $X_{k,j_k}$ in the end of time $(t-1)$ and $t$, respectively. Similarly, $\vv_{t-1}(X_{k,j_k})$ and $\vv_{t}(X_{k,j_k})$ denote the coding vectors in the end of time $(t-1)$ and $t$, respectively.

\begin{lemma}\label{lem:non-interfering}In the end of the $t$-th time slot, consider any $X_{k,j}$ out of all the information packets $X_{1,1}$ to $X_{K,nR_K}$. Its assigned vector $\vv_t(X_{k,j})$ is non-interfering from the perspective of $d_i$ for all $i\in (S_t(X_{k,j})\cup\{k\})$.
\end{lemma}


To illustrate Lemma~\ref{lem:non-interfering}, consider our 5-time-slot example. In the end of Slot~4, we have $\vv(X_1)=(1,1,0)$ and $S(X_1)\cup \{1\}=\{1,2,3\}$. From $d_1$'s perspective, $\Omega_{Z,1}(4)=\linsp((0,1,0),(0,0,1))$ and $\Omega_{M,1}=\linsp((1,0,0))$. $\vv(X_1)\in \linsp(\Omega_{Z,1}(4),\Omega_{M,1})$ is indeed non-interfering from $d_1$'s perspective. The same reasoning can be applied to $d_2$ to show that $\vv(X_1)$ is non-interfering from $d_2$'s perspective. For $d_3$, $\Omega_{Z,3}(4)=\linsp((1,1,0))$ and $\Omega_{M,3}=\linsp((0,0,1))$.  $\vv(X_1)\in \linsp(\Omega_{Z,3}(4),\Omega_{M,3})$ is indeed non-interfering from $d_3$'s perspective. Lemma~\ref{lem:non-interfering} holds for our illustrative example.

\begin{lemma}\label{lem:decodability} In the end of the $t$-th time slot, we use $\Omega_{R,k}(t)$ to denote the {\em remaining space} of the PE scheme:
\begin{align}
&\Omega_{R,k}(t)\stackrel{\Delta}{=}\nonumber\\
&\linsp(\vv_t(X_{k,j}):\forall j\in [nR_k] \text{ satisfying } k\notin S_t(X_{k,j})).\nonumber
\end{align}

For any $n$ and any $\epsilon>0$, there exists a sufficiently large finite field $\GF(q)$ such that for all $k\in[K]$ and $t\in[n]$,
\begin{align}
&\prop\left(\linsp(\Omega_{Z,k}(t), \Omega_{R,k}(t))=\linsp(\Omega_{Z,k}(t), \Omega_{M,k})\right)\nonumber\\
&>1-\epsilon.\nonumber
\end{align}
\end{lemma}

Intuitively, Lemma~\ref{lem:decodability} says that if in the end of time $t$ we directly transmit all the {\em remaining} coded packets $\left\{\vv_t(X_{k,j}):\forall j\in [nR_k], k\notin S_t(X_{k,j})\right\}$ from $s$ to $d_k$ through a noise-free information pipe, then with high probability, $d_k$ can successfully decode all the desired information packets $X_{k,1}$ to $X_{k,nR_k}$ (see Lemma~\ref{lem:simple-dec}) by the knowledge space $\Omega_{Z,k}(t)$ and the new information of the remaining space $\Omega_{R,k}(t)$.

Lemma~\ref{lem:decodability} directly implies the following corollary.

\begin{corollary}\label{cor:dec}
For any $n$ and any $\epsilon>0$, there exists a sufficiently large finite field $\GF(q)$ such that the following statement holds. If in the end of the $n$-th time slot, all information packets $X_{k,j}$ have $S_n(X_{k,j})\ni k$, then
\begin{align}
\prop(\forall k, d_k\text{ can decode all its desired } \{X_{k,j}\})>1-\epsilon.\nonumber
\end{align}
\end{corollary}
\begin{proof} If in the end of the $n$-th time slot, all $X_{k,j}$ have $S_n(X_{k,j})\ni k$, then the corresponding $\Omega_{R,k}(n)=\{0\}$ contains only the origin for all $k\in[K]$. Therefore, Corollary~\ref{cor:dec} is simply a restatement of Lemmas~\ref{lem:simple-dec} and~\ref{lem:decodability}.
\end{proof}

To illustrate Corollary~\ref{cor:dec}, consider our 5-time-slot example. In the end of Slot~5, since $k\in S(X_k)$ for all $k\in\{1,2,3\}$, Corollary~\ref{cor:dec} guarantees that with high probability all $d_k$ can decode the desired $X_k$, which was first observed in the example of Section~\ref{subsec:example}.

The proofs of Lemmas~\ref{lem:non-interfering} and~\ref{lem:decodability} are relegated to Appendices~\ref{app:lem-non-interfering} and~\ref{app:lem-decodability}, respectively.

\subsection{The Intuitions Of The Packet Evolution Scheme\label{subsec:intuition-PE}}
Lemmas~\ref{lem:non-interfering} and~\ref{lem:decodability} are the key properties of a PE scheme. In this subsection, we discuss the corresponding intuitions.

{\sf Receiving} {\bf the information packet $X_{k,j}$:}\quad
Each information packet keeps a coding vector $\vv(X_{k,j})$. Whenever we would like to communicate $X_{k,j}$ to destination $d_k$, instead of sending a non-coded packet $X_{k,j}$ directly, we send an intersession coded packet according to the coding vector $\vv(X_{k,j})$. Lemma~\ref{lem:decodability} shows that if we send all the coded vectors $\vv(X_{k,j})$ that have not been heard by $d_k$ (with $k\notin S(X_{k,j})$) through a noise-free information pipe, then $d_k$ can indeed decode all the desired packets $X_{k,j}$ with close-to-one probability. It also implies, although in an implicit way, that once a $\vv(X_{k,j_0})$ is heard by $d_k$ for some $j_0$ (therefore $k\in S(X_{k,j_0})$), there is no need to transmit this particular $\vv(X_{k,j_0})$ in the later time slots. Jointly, these two implications show that we can indeed use the coded packet $\vv(X_{k,j})$ as a substitute for $X_{k,j}$ without losing any information. In the broadest sense, we can say that $d_k$ {\sf receives} a packet $X_{k,j}$ if the corresponding $\vv(X_{k,j})$ successfully arrives $d_k$ in some time slot $t$.

For each $X_{k,j}$, the set $S(X_{k,j})$ serves two purposes: (i) Keep track of whether its intended destination $d_k$ has {\sf received} this $X_{k,j}$ (through the $\vv(X_{k,j})$), and (ii) Keep track of whether $\vv(X_{k,j})$ is non-interfering to other destinations $d_i$, $i\neq k$. We discuss these two purposes separately.

{\bf Tracking the reception of the intended $d_k$:}\quad We first note that in the end of time 0, $d_k$ has not received any packet and we indeed have $k\notin S(X_{k,j})=\emptyset$. We then notice that for any given $X_{k,j}$, the set $S(X_{k,j})$ evolves over time. By Line~\ref{line:S-update} of the {\sc Update}, we can prove that as time proceeds, the first time $t_0$ such that $k\in S(X_{k,j})$ must be the first time when $X_{k,j}$ is {\sf received} by $d_k$ (i.e., $X_{k,j}$ is chosen in the beginning of time $t$ and $k\in S_\text{rx}$ in the end of time $t$). One can also show that for any $X_{k,j}$ once $k\in S_{t_0}(X_{k,j})$ in the end of time $t_0$ for some $t_0$, we will have $k\in S_t(X_{k,j})$ for all $t\geq t_0$. By the above reasonings, checking whether $k\in S(X_{k,j})$ indeed tells us whether the intended receiver $d_k$ has {\sf received} $X_{k,j}$.

{\bf Tracking the non-interference from the perspective of $d_i\neq d_k$:}\quad
Lemma~\ref{lem:non-interfering} also ensures that $\vv(X_{k,j})$ is non-interfering from $d_i$'s perspective for any $i\in S(X_{k,j})$, $i\neq k$. Therefore $S(X_{k,j})$ successfully tracks whether $\vv(X_{k,j})$ is non-interfering from the perspectives of $d_i$, $i\neq k$.

{\bf Serving multiple destinations simultaneously by mixing non-interfering packets:}\quad The above discussion ensures that when we would like to send an information packet $X_{k,j_k}$ to $d_k$, we can send a coded packet $\vv(X_{k,j_k})$ as an information-lossless substitute. On the other hand, by Lemma~\ref{lem:non-interfering}, such $\vv(X_{k,j_k})$ is non-interfering from $d_i$'s perspective for all $i\in (S(X_{k,j_k})\cup\{k\})$. Therefore, instead of sending a single packet $\vv(X_{k,j_k})$, it is beneficial to {\em combine} the transmission of two packets $\vv(X_{k,j_k})$ and $\vv(X_{l,j_l})$ together, as long as $l\in S(X_{k,j_k})$ and $k\in S(X_{l,j_l})$. More explicitly, suppose we simply add the two packets together and transmit a packet corresponding to $[\vv(X_{k,j_k})+\vv(X_{l,j_l})]$. Since $\vv(X_{k,j_k})$ is non-interfering from $d_l$'s perspective, it is as if $d_l$ directly receives  $\vv(X_{l,j_l})$ without any interference. Similarly, since $\vv(X_{l,j_l})$ is non-interfering from $d_k$'s perspective, it is as if $d_k$ directly receives  $\vv(X_{k,j_k})$ without any interference. By generalizing this idea, a PE scheme first selects a $T\subseteq [K]$ and then choose all $X_{k,j_k}$ such that $k\in T$ and $\vv(X_{k,j_k})$ are non-interfering from $d_l$'s perspective for all $l\in T\backslash k$ (see Line~\ref{line:targeting-set} of the PE scheme). This thus ensures that the coded packet $\vv_\text{tx}$ in Line~\ref{line:vv-construct} of the PE scheme can serve all destinations $k\in T$ simultaneously.

{\bf Creating new coding opportunities while exploiting the existing coding opportunities:}\quad As discussed in the example of Section~\ref{subsec:example}, the suboptimality of the existing 2-phase approach for $K\geq 3$ destinations is due to the fact that it fails to create new coding opportunities while exploiting old coding opportunities. The PE scheme was designed to solve this problem. More explicitly, for each $X_{k,j}$ the $\vv(X_{k,j})$ is non-interfering for all $d_i$ satisfying $i\in (S(X_{k,j})\cup\{k\})$. Therefore, the larger the set $S(X_{k,j})$ is, the larger the number of sessions that can be coded together with $\vv(X_{k,j})$. To create more coding opportunities, we thus need to be able to enlarge the $S(X_{k,j})$ set over time. Let us assume that the {\sc Packet Selection} in Line~\ref{line:targeting-set} chooses the $X_{k,j}$ such that $S(X_{k,j})=T\backslash k$. That is, we choose the $X_{k,j}$ that can be mixed with those $(s,d_l)$ sessions with $l\in S(X_{k,j})\cup\{k\}=T$. Then Line~\ref{line:S-update} of the {\sc Update} guarantees that if some other $d_i$, $i\notin T$, overhears the coded transmission, we can update $S(X_{k,j})$ with a strictly larger set $S(X_{k,j})\cup S_\text{rx}$. Therefore, new coding opportunity is created since we can now mix more sessions together with $X_{k,j}$. Note that the coding vector $\vv(X_{k,j})$ is also updated accordingly. The new $\vv(X_{k,j})$ represents the necessary ``code alignment" in order to utilize this newly created coding opportunity. The (near-) optimality of the PE scheme is rooted deeply in the concept of code alignment, which aligns the ``non-interfering subspaces" through the joint use of $S(X_{k,j})$ and $\vv(X_{k,j})$.

\section{Quantify The Achievable Rates of PE Schemes\label{sec:pf-of-achievability}}
In this section, we describe how to use the PE schemes to attain the capacity of 1-to-3 broadcast PECs with COF (Proposition~\ref{prop:cap3}), the achievability results for general 1-to-$K$ broadcast PEC with COF (Proposition~\ref{prop:ach2}), the capacity results for symmetric broadcast PECs (Proposition~\ref{prop:cap-sym}) and for spatially independent PECs with one-sided fairness constraints (Proposition~\ref{prop:cap-osf}).

We first describe a detailed construction of a capacity-achieving PE scheme for general 1-to-3 broadcast PECs with COF in Section~\ref{subsec:detailed3} and then discuss the corresponding high-level intuition in Section~\ref{subsec:high-level-cap3}. The high-level discussion will later be used to prove the achievability results for general 1-to-$K$ broadcast PEC with COF in Section~\ref{subsec:1-to-M-ach}. The proofs of the capacity results of two special classes of PECs are provided in Section~\ref{subsec:2spec-class}.

\subsection{Achieving the Capacity of 1-to-3 Broadcast PECs With COF --- Detailed Construction~\label{subsec:detailed3}}
Consider a 1-to-3 broadcast PEC with arbitrary channel parameters $\{p_{S\overline{\{1,2,3\}\backslash S}}\}$. Without loss of generality, assume that the marginal success probability $p_k>0$ for $k=1,2,3$. For the cases in which  $p_k=0$ for some $k$, such $d_k$ cannot receive any packet. The 1-to-3 broadcast PEC thus collapses to a 1-to-2 broadcast PEC, the capacity of which was proven in \cite{GeorgiadisTassiulas09}.

 Given any arbitrary rate vector $(R_1,R_2,R_3)$ that is in the interior of the capacity outer bound of Proposition~\ref{prop:outer}, our goal is to design a PE scheme for which each $d_k$ can successfully decode its desired packets $\{X_{k,j}:\forall j\in[nR_k]\}$, for $k\in\{1,2,3\}$, after $n$ usages of the broadcast PEC. Before describing such a PE scheme, we introduce a new definition and the corresponding lemma.

Given a rate vector $(R_1,R_2,R_3)$ and the PEC channel parameters $\{p_{S\overline{\{1,2,3\}\backslash S}}\}$,
we say that destination $d_i$ {\em dominates} another $d_k$, $i\neq k$ if
\begin{align}
&R_i\left(\frac{1}{p_{\cup (\{1,2,3\}\backslash k)}}-\frac{1}{p_{\cup \{1,2,3\}}}\right)\nonumber\\
&\hspace{2cm}\geq R_k\left(\frac{1}{p_{\cup (\{1,2,3\}\backslash i)}}-\frac{1}{p_{\cup \{1,2,3\}}}\right).\label{eq:dominance-cond}
\end{align}
\begin{lemma} \label{lem:dominance}For distinct values of $i,k,l\in\{1,2,3\}$, if $d_i$ dominates $d_k$, and $d_k$ dominates $d_l$, then we must have $d_i$ dominates $d_l$.
\end{lemma}
\begin{proof} Suppose this lemma is not true and we have $d_i$ dominates $d_k$, $d_k$ dominates $d_l$, and $d_l$ dominates $d_i$. By definition, we must have
\begin{align}
&R_i\left(\frac{1}{p_{\cup (\{1,2,3\}\backslash k)}}-\frac{1}{p_{\cup \{1,2,3\}}}\right)\nonumber\\
&\hspace{2cm}\geq R_k\left(\frac{1}{p_{\cup (\{1,2,3\}\backslash i)}}-\frac{1}{p_{\cup \{1,2,3\}}}\right),\label{eq:i-k-dominate}\\
&R_k\left(\frac{1}{p_{\cup (\{1,2,3\}\backslash l)}}-\frac{1}{p_{\cup \{1,2,3\}}}\right)\nonumber\\
&\hspace{2cm}\geq R_l\left(\frac{1}{p_{\cup (\{1,2,3\}\backslash k)}}-\frac{1}{p_{\cup \{1,2,3\}}}\right),\label{eq:k-l-dominate}\\
&R_l\left(\frac{1}{p_{\cup (\{1,2,3\}\backslash i)}}-\frac{1}{p_{\cup \{1,2,3\}}}\right)\nonumber\\
&\hspace{2cm}\geq R_i\left(\frac{1}{p_{\cup (\{1,2,3\}\backslash l)}}-\frac{1}{p_{\cup \{1,2,3\}}}\right).\label{eq:l-i-dominate}
\end{align}
We then notice that the product of the left-hand sides of \eqref{eq:i-k-dominate}, \eqref{eq:k-l-dominate}, and \eqref{eq:l-i-dominate} equals the product of the right-hand side of  \eqref{eq:i-k-dominate}, \eqref{eq:k-l-dominate}, and \eqref{eq:l-i-dominate}. As a result, all three inequalities of \eqref{eq:i-k-dominate}, \eqref{eq:k-l-dominate}, and \eqref{eq:l-i-dominate} must also be equalities. Since \eqref{eq:l-i-dominate} is an equality, we can also say that $d_i$ dominates $d_l$. The proof of Lemma~\ref{lem:dominance} is complete.
\end{proof}

By Lemma~\ref{lem:dominance}, we can assume that $d_1$ dominates $d_2$, $d_2$ dominates $d_3$, and $d_1$ dominates $d_3$, which can be achieved by relabeling the destinations $d_k$. We then describe a detailed capacity-achieving PE scheme, which has four major phases. The dominance relationship is a critical part in the proposed PE scheme. The high-level discussion of this capacity-achieving PE scheme will be provided in Section~\ref{subsec:high-level-cap3}

{\bf Phase~1} contains 3 sub-phases. In {\bf Phase~1.1}, we always choose $T=\{1\}$ for the PE scheme. In the beginning of time 1,  we first select $X_{1,1}$. We keep transmitting the uncoded packet according to $\vv(X_{1,1})=\delta_{1,1}$ until it is received by at least one of the three destinations $\{d_1,d_2,d_3\}$. Update its $S(X_{1,1})$ and $\vv(X_{1,1})$ according to the {\sc Update} rule.  Then we move to packet $X_{1,2}$. Keep transmitting the uncoded packet according to $\vv(X_{1,2})=\delta_{1,2}$ until it is received by at least one of the three receivers $\{d_1,d_2,d_3\}$. Update its $S(X_{1,2})$ and $\vv(X_{1,2})$ according to the {\sc Update} rule. Repeat this process until all $X_{1,j}$, $j\in [nR_1]$ is received by at least one receiver. By the law of large numbers, Phase~1.1 will continue for
\begin{align}
\approx \frac{nR_1}{p_{\cup\{1,2,3\}}}\text{ time slots.} \label{eq:phase11}
\end{align}

{\bf Phase~1.2}: After Phase~1.1 we move to Phase~1.2. In {Phase~1.2}, we always choose $T=\{2\}$ for the PE scheme. In the beginning of Phase~1.2,  we first select $X_{2,1}$. We keep transmitting the uncoded packet according to $\vv(X_{2,1})=\delta_{2,1}$ until it is received by at least one of the three destinations $\{d_1,d_2,d_3\}$. Update its $S(X_{2,1})$ and $\vv(X_{2,1})$.  Repeat this process until all $X_{2,j}$, $j\in [nR_2]$ is received by at least one receiver. By the law of large numbers, Phase~1.2 will continue for
\begin{align}\approx \frac{nR_2}{p_{\cup\{1,2,3\}}}\text{ time slots.} \label{eq:phase12}\end{align}

{\bf Phase~1.3}: After Phase~1.2 we move to Phase~1.3. In {Phase~1.3},  we always choose $T=\{3\}$ for the PE scheme. We repeat the same process as in Phases~1.1 and~1.2 until
 all $X_{3,j}$, $j\in [nR_3]$ is received by at least one receiver. By the law of large numbers, Phase~1.3 will continue for
 \begin{align}\approx \frac{nR_3}{p_{\cup\{1,2,3\}}}\text{ time slots.} \label{eq:phase13}\end{align}

{\bf Phase~2}: After Phase~1.3, we move to Phase~2. {Phase~2} contains 3 sub-phases. In  {\bf Phase~2.1},  we always choose $T=\{2,3\}$ for the PE scheme. Consider all the packets $X_{2,j}$ that have $S(X_{2,j})=\{3\}$ in the end of Phase~1.3, which was resulted/created in Phase 1.2 when a Phase-1.2 packet was received by $d_3$ only. Totally there are $\approx \frac{nR_2p_{\{3\}\overline{\{1,2\}}}}{p_{\cup\{1,2,3\}}}$ such packets, which are termed the queue $Q_{2;3\overline{1}}$ packets.  Consider all the packets $X_{3,j}$ that have $S(X_{3,j})=\{2\}$ in the end of Phase~1.3, which was resulted/created in Phase 1.3 when a Phase-1.3 packet was received by $d_2$ only. Totally there are $\approx \frac{nR_3p_{\{2\}\overline{\{1,3\}}}}{p_{\cup\{1,2,3\}}}$ such packets, which are termed the queue $Q_{3;2\overline{1}}$ packets.

We order all the $Q_{2;3\overline{1}}$ packets in any arbitrary sequence and order all the $Q_{3;2\overline{1}}$ packets in any arbitrary sequence. In the beginning of Phase~2.1,  we first select the head-of-the-line $X_{2,j_2}$ and the head-of-line  $X_{3,j_3}$ from these two queues $Q_{2;3\overline{1}}$ and $Q_{3;2\overline{1}}$, respectively. Since
\begin{align}
S(X_{2,j_2})\cup\{2\}=T=\{2,3\}=S(X_{3,j_3})\cup\{3\},\nonumber
\end{align}
these two packets can be linearly combined together. Let $\vv_\text{tx}$ denote the overall coding vector generated from these two packets (see Line~\ref{line:vv-construct} of the main PE scheme). As discussed in Line~\ref{line:transmit-v} of the main PE scheme, we keep transmitting the same coded packet $\vv_\text{tx}$ until at least one of the two packets $X_{2,j_2}$ and $X_{3,j_3}$ has a new $S(X_{2,j_2})$ (or a new $S(X_{3,j_3})$). In the end, we thus have three subcases: (i) only $X_{2,j_2}$ has a new $S(X_{2,j_2})$, (ii) only $X_{3,j_3}$ has a new $S(X_{3,j_3})$, and (iii) both $X_{2,j_2}$ has a new $S(X_{2,j_2})$ and $X_{3,j_3}$ has a new $S(X_{3,j_3})$. In Case (i), we keep the same $T=\{2,3\}$ and the same $X_{3,j_3}$ but switch to the next-in-line $Q_{2;3\overline{1}}$ packet $X_{2,j_2'}$. The new $X_{2,j_2'}$ will be then be used, together with the existing $X_{3,j_3}$ to generate new $\vv_\text{tx}$ in Line~\ref{line:vv-construct} of the main PE scheme for the next time slot(s). In Case~(ii), we keep the same $T=\{2,3\}$ and the same $X_{2,j_2}$ but switch to the next-in-line $Q_{3;2\overline{1}}$ packet $X_{3,j_3'}$. The new $X_{3,j_3'}$ will then be used, together with the existing $X_{2,j_2}$, to generate new $\vv_\text{tx}$ in Line~\ref{line:vv-construct} of the main PE scheme for the next time slot(s). In Case~(iii), we keep the same $T=\{2,3\}$ and switch to the next-in-line packets $X_{2,j_2'}$ and $X_{3,j_3'}$. The new pair $X_{2,j_2'}$ and $X_{3,j_3'}$ will then be used to generate new $\vv_\text{tx}$ in Line~\ref{line:vv-construct} of the main PE scheme for the next time slot(s). We repeat the above process until we have used up all $Q_{3;2\overline{1}}$ packets $X_{3,j}$.

{\em Remark 1:} One critical observation of the PE scheme is that when two packets $X_{2,j_2}$ or $X_{3,j_3}$ are mixed together to generate $\vv_\text{tx}$, each packet still keeps its own identity $X_{2,j_2}$ and $X_{3,j_3}$, its own associated sets $S(X_{2,j_2})$ and $S(X_{3,j_3})$ and coding vectors $\vv(X_{2,j_2})$ and $\vv(X_{3,j_3})$. Even the decision whether to update $S(X)$ or $\vv(X)$ is made separately (Line~\ref{line:update-4-K} of the {\sc Update}) for each of the two packets $X_{2,j_2}$ or $X_{3,j_3}$. Therefore,
it is as if the two packets $X_{2,j_2}$ or $X_{3,j_3}$ are sharing the single time slot in a non-interfering way (like carpooling together). Following this observation, in Phase~2.1, whether we decide to switch the current $X_{2,j_2}$ to the next-in-line $Q_{2;3\overline{1}}$ packet $X_{2,j_2'}$ is also completely independent from the decision whether to switch the current $X_{3,j_3}$ to the next-in-line $Q_{3;2\overline{1}}$ packet $X_{3,j_3'}$.

{\em Remark 2:} We first take a closer look at when a $Q_{3;2\overline{1}}$ packet $X_{3,j_3}$ will be switched to the next-in-line packet $X_{3,j_3'}$. By Line~\ref{line:S-update} of the {\sc Update}, we switch to the next-in-line  $X_{3,j_3'}$ if and only if one of $\{d_1,d_3\}$ has received the current packet $\vv_{\text{tx}}$, in which $X_{3,j_3}$ participates. Therefore, in average each $X_{3,j_3}$ will stay in Phase~2.1 for $\frac{1}{p_{\cup\{1,3\}}}$ time slots. Since we have $\approx \frac{nR_3p_{\{2\}\overline{\{1,3\}}}}{p_{\cup\{1,2,3\}}}$  number of $Q_{3;2\overline{1}}$ packets to begin with, it takes
\begin{align}
\approx \frac{nR_3p_{\{2\}\overline{\{1,3\}}}}{p_{\cup\{1,2,3\}}}\frac{1}{p_{\cup\{1,3\}}}=
nR_3\left(\frac{1}{p_{\cup\{1,3\}}}-\frac{1}{p_{\cup\{1,2,3\}}}\right)\label{eq:d3-phase2.1}
\end{align}
to completely finish the $Q_{3;2\overline{1}}$ packets. By similar arguments, it takes
\begin{align}
\approx \frac{nR_2p_{\{3\}\overline{\{1,2\}}}}{p_{\cup\{1,2,3\}}}\frac{1}{p_{\cup\{1,2\}}}
=nR_2\left(\frac{1}{p_{\cup\{1,2\}}}-\frac{1}{p_{\cup\{1,2,3\}}}\right)\label{eq:d2-phase2.1}
\end{align}
to completely use up the $Q_{2;3\overline{1}}$ packets. Since we assume that $d_2$ dominates $d_3$, the dominance inequality in  \eqref{eq:dominance-cond} implies that \eqref{eq:d2-phase2.1} is no smaller than \eqref{eq:d3-phase2.1}. Therefore we indeed can finish the $Q_{3;2\overline{1}}$ packets before exhausting the $Q_{2;3\overline{1}}$ packets.

{\em Remark~3:} Overall it takes roughly \eqref{eq:d3-phase2.1} of time slots to finish Phase~2.1.

{\bf Phase~2.2}: After Phase~2.1, we move to Phase~2.2. In  {Phase~2.2},  we always choose $T=\{1,3\}$ for the PE scheme. Consider all the packets $X_{1,j}$ that have $S(X_{1,j})=\{3\}$ in the end of Phase~2.1, which was resulted/created in Phase 1.1 when a Phase-1.1 packet was received by $d_3$ only. Totally there are $\approx \frac{nR_1p_{\{3\}\overline{\{1,2\}}}}{p_{\cup\{1,2,3\}}}$ such packets, which are termed the queue $Q_{1;3\overline{2}}$ packets.  Consider all the packets $X_{3,j}$ that have $S(X_{3,j})=\{1\}$ in the end of Phase~2.1, which was resulted/created in Phase 1.3 when a Phase-1.3 packet was received by $d_1$ only. We note that there are some $Q_{3;2\overline{1}}$ packets being transmitted in Phase~2.1. Before the transmission of Phase~2.1, those packets have $S(X_{3,j})=\{2\}$ and after the transmission of Phase~2.1, those packets will have their $S(X_{3,j})$ being one of the three forms $\{1,2\}$, $\{2,3\}$, and $\{1,2,3\}$ (see Line~\ref{line:S-update} of the {\sc Update}). Therefore, Phase~2.1 does not contribute to any $X_{3,j}$ packets considered in Phase~2.2 (those with $S(X_{3,j})=\{1\}$).
Totally there are $\approx \frac{nR_3p_{\{1\}\overline{\{2,3\}}}}{p_{\cup\{1,2,3\}}}$ packets considered in Phase~2.2, which are termed the queue $Q_{3;1\overline{2}}$ packets.

We order all the $Q_{1;3\overline{2}}$ packets in any arbitrary sequence and order all the $Q_{3;1\overline{2}}$ packets in any arbitrary sequence. Following similar steps  as in Phase~2.1, we first mix the head-of-the-line packets $X_{1,j_1}$ and $X_{3,j_3}$ of $Q_{1;3\overline{2}}$ and $Q_{3;1\overline{2}}$, respectively, and then make the decisions of switching to the next-in-line packets $X_{1,j_1'}$ and $X_{3,j_3'}$ independently for the two queues $Q_{1;3\overline{2}}$ and $Q_{3;1\overline{2}}$. We repeat the above process until we have used up all $Q_{3;1\overline{2}}$ packets $X_{3,j}$.

{\em Remark:} We take a closer look at when a $Q_{3;1\overline{2}}$ packet $X_{3,j_3}$ will be switched to the next-in-line packet $X_{3,j_3'}$. By Line~\ref{line:S-update} of the {\sc Update}, we switch to the next-in-line  $X_{3,j_3'}$ if and only if one of $\{d_2,d_3\}$ has received the current packet $\vv_{\text{tx}}$, in which $X_{3,j_3}$ participates. Therefore, in average each $X_{3,j_3}$ will stay in Phase~2.2 for $\frac{1}{p_{\cup\{2,3\}}}$ time slots. Since we have $\approx \frac{nR_3p_{\{1\}\overline{\{2,3\}}}}{p_{\cup\{1,2,3\}}}$  number of $Q_{3;1\overline{2}}$ packets to begin with, it takes
\begin{align}
\approx \frac{nR_3p_{\{1\}\overline{\{2,3\}}}}{p_{\cup\{1,2,3\}}}\frac{1}{p_{\cup\{2,3\}}}
=nR_3\left(\frac{1}{p_{\cup\{2,3\}}}-\frac{1}{p_{\cup\{1,2,3\}}}\right)\label{eq:d3-phase2.2}
\end{align}
to completely finish the $Q_{3;1\overline{2}}$ packets. By similar arguments, it takes
\begin{align}
\approx \frac{nR_1p_{\{3\}\overline{\{1,2\}}}}{p_{\cup\{1,2,3\}}}\frac{1}{p_{\cup\{1,2\}}}
=nR_1\left(\frac{1}{p_{\cup\{1,2\}}}-\frac{1}{p_{\cup\{1,2,3\}}}\right)\label{eq:d1-phase2.2}
\end{align}
to completely use up the $Q_{1;3\overline{2}}$ packets. Since we assume that $d_1$ dominates $d_3$, the dominance inequality in  \eqref{eq:dominance-cond} implies that \eqref{eq:d1-phase2.2} is no smaller than \eqref{eq:d3-phase2.2}. Therefore we indeed can finish the $Q_{3;1\overline{2}}$ packets before exhausting the $Q_{1;3\overline{2}}$ packets.
Overall it takes roughly \eqref{eq:d3-phase2.2} number of time slots to finish Phase~2.2.

{\bf Phase~2.3}: After Phase~2.2, we move to Phase~2.3. In {Phase~2.3}, we always choose $T=\{1,2\}$ for the PE scheme. Consider all the packets $X_{1,j}$ that have $S(X_{1,j})=\{2\}$ in the end of Phase~2.2, which was resulted/created in Phase 1.1 when a Phase-1.1 packet was received by $d_2$ only. Note that the transmission in Phase~2.2 does not create any new such packets.
Totally there are thus $\approx \frac{nR_1p_{\{2\}\overline{\{1,3\}}}}{p_{\cup\{1,2,3\}}}$ such packets, which are termed the queue $Q_{1;2\overline{3}}$ packets. Consider all the packets $X_{2,j}$ that have $S(X_{2,j})=\{1\}$ in the end of Phase~2.2, which was resulted/created in Phase 1.2 when a Phase-1.2 packet was received by $d_1$ only. Note that the transmission in Phase~2.1 does not create any new such packets.
Totally there are thus $\approx \frac{nR_2p_{\{1\}\overline{\{2,3\}}}}{p_{\cup\{1,2,3\}}}$ such packets, which are termed the queue $Q_{2;1\overline{3}}$ packets.

We order all the $Q_{1;2\overline{3}}$ packets in any arbitrary sequence and order all the $Q_{2;1\overline{3}}$ packets in any arbitrary sequence. Following similar steps  as in Phases~2.1 and~2.2, we first mix the head-of-the-line packets $X_{1,j_1}$ and $X_{2,j_2}$ of $Q_{1;2\overline{3}}$ and $Q_{2;1\overline{3}}$, respectively, and then make the decisions of switching to the next-in-line packets $X_{1,j_1'}$ and $X_{2,j_2'}$ independently for the two queues $Q_{1;2\overline{3}}$ and $Q_{2;1\overline{3}}$. We repeat the above process until we have used up all $Q_{2;1\overline{3}}$ packets $X_{2,j}$. By the assumption that $d_1$ dominates $d_2$ and  by the same arguments as in Phases~2.1 and~2.2,  we indeed can finish the $Q_{2;1\overline{3}}$ packets before exhausting the $Q_{1;2\overline{3}}$ packets.
Overall it takes roughly
\begin{align}
\approx \frac{nR_2p_{\{1\}\overline{\{2,3\}}}}{p_{\cup\{1,2,3\}}}\frac{1}{p_{\cup\{2,3\}}}
=nR_2\left(\frac{1}{p_{\cup\{2,3\}}}-\frac{1}{p_{\cup\{1,2,3\}}}\right)\label{eq:d2-phase2.3}
\end{align}
 of time slots to finish Phase~2.3.

{\bf Phase 3:} Before the description of Phase-3 operations, we first summarize the status of all the packets in the end of Phase~2.3. For $d_3$, all $X_{3,j}$ packets that have $S(X_{3,j})=\emptyset$ have been used up in Phase~1.3. All $X_{3,j}$ packets that have $S(X_{3,j})=\{1\}$ have been used up in Phase~2.2. All $X_{3,j}$ packets that have $S(X_{3,j})=\{2\}$ have been used up in Phase~2.1. As a result, all the $X_{3,j}$ packets are either {\sf received} by $d_3$ (i.e., having $3\in S(X_{3,j})$) or have $S(X_{3,j})=\{1,2\}$. For Phase 3, we will focus on the latter type of $X_{3,j}$ packets, which are termed the $Q_{3;12}$ packets. Recall the definition of $f_p(S\overline{T})$ in \eqref{eq:fp-def}. Totally, we have
\begin{align}
nR_3\left(\frac{p_{12\overline{3}}}{p_{\cup \{1,2,3\}}}+\frac{p_{1\overline{23}}}{p_{\cup \{1,2,3\}}}\frac{f_p(2\overline{3})}{p_{\cup \{2,3\}}}+\frac{p_{2\overline{13}}}{p_{\cup \{1,2,3\}}}\frac{f_p(1\overline{3})}{p_{\cup \{1,3\}}}\right)\label{eq:d3-phase3}
\end{align}
number of $Q_{3;12}$ packets in the beginning of Phase~3, where the first, second, and the third terms correspond to the $Q_{3;12}$ packets generated in Phase~1.3, Phase~2.2, and Phase~2.1, respectively. We can further simplify \eqref{eq:d3-phase3} as
\begin{align}
\eqref{eq:d3-phase3}=nR_3p_3\left(\frac{1}{p_3}-\frac{1}{p_{\cup\{1,3\}}}-\frac{1}{p_{\cup\{2,3\}}}
+\frac{1}{p_{\cup\{1,2,3\}}}\right).\label{eq:d3-phase3-sim}
\end{align}

For $d_2$, all $X_{2,j}$ packets that have $S(X_{2,j})=\emptyset$ have been used up in Phase~1.2. All $X_{2,j}$ packets that have $S(X_{2,j})=\{1\}$ have been used up in Phase~2.3. As a result, all the $X_{2,j}$ packets must satisfy one of the following: (i) $X_{2,j}$ are {\sf received} by $d_2$ (i.e., having $2\in S(X_{2,j})$), or (ii) have $S(X_{2,j})=\{3\}$, or (iii) have $S(X_{2,j})=\{1,3\}$. For Phase 3, we will focus on the latter two types of $X_{2,j}$ packets, which are termed the $Q_{2;3\overline{1}}$ and the $Q_{2;13}$ packets, respectively. There are
\begin{align}nR_2\frac{p_{3\overline{12}}}{p_{\cup\{1,2,3\}}}
-nR_3\frac{p_{2\overline{13}}}{p_{\cup\{1,2,3\}}}\frac{p_{\cup\{1,2\}}}{p_{\cup \{1,3\}}}\label{eq:d23-phase3}
\end{align}
number of $Q_{2;3\overline{1}}$ packets in the beginning of Phase~3, where the first term is the number of $Q_{2;3\overline{1}}$ packets generated in Phase~1.2 and the second term corresponds to the number of $Q_{2;3\overline{1}}$ packets that are used up in Phase~2.1. \eqref{eq:d23-phase3} can be simplified to
\begin{align}
\eqref{eq:d23-phase3}=p_{\cup\{1,2\}}&\left(nR_2\left(\frac{1}{p_{\cup\{1,2\}}}-\frac{1}{p_{\cup\{1,2,3\}}}\right)\right.\nonumber\\
&\hspace{.8cm}\left.-nR_3\left(\frac{1}{p_{\cup\{1,3\}}}-\frac{1}{p_{\cup \{1,2,3\}}}\right)\right).\label{eq:d23-phase3-sim}
\end{align}
 There are
\begin{align}nR_2\frac{p_{13\overline{2}}}{p_{\cup\{1,2,3\}}}
+nR_2\frac{p_{1\overline{23}}}{p_{\cup\{1,2,3\}}}\frac{f_p(3\overline{2})}{p_{\cup\{2,3\}}}
+nR_3\frac{p_{2\overline{13}}}{p_{\cup\{1,2,3\}}}\frac{f_p(1\overline{2})}{p_{\cup \{1,3\}}}\label{eq:d213-phase3}
\end{align}
number of $Q_{2;13}$ packets in the beginning of Phase~3, where the first, second, and third terms correspond to the number of $Q_{2;13}$ packets generated in Phase~1.2, Phase~2.3, and Phase 2.1, respectively.

For $d_1$, all $X_{1,j}$ packets that have $S(X_{1,j})=\emptyset$ have been used up in Phase~1.1. As a result, all the $X_{1,j}$ packets must satisfy one of the following: (i) $X_{1,j}$ are {\sf received} by $d_1$ (i.e., having $1\in S(X_{1,j})$), or (ii) have $S(X_{1,j})=\{2\}$, (iii) have $S(X_{1,j})=\{3\}$, or (iv) have $S(X_{1,j})=\{2,3\}$. For Phase 3, we will focus on the types (ii) and (iii), which are termed the $Q_{1;2\overline{3}}$ and the $Q_{1;3\overline{2}}$ packets, respectively. There are
\begin{align}&nR_1\frac{p_{2\overline{13}}}{p_{\cup\{1,2,3\}}}
-nR_2\frac{p_{1\overline{23}}}{p_{\cup\{1,2,3\}}}\frac{p_{\cup\{1,3\}}}{p_{\cup \{2,3\}}}\label{eq:d12-phase3}
\end{align}
number of $Q_{1;2\overline{3}}$ packets in the beginning of Phase~3, where the first term is the number of $Q_{1;2\overline{3}}$ packets generated in Phase~1.1 and the second term corresponds to the number of $Q_{1;2\overline{3}}$ packets that are used up in Phase~2.3. \eqref{eq:d12-phase3} can be simplified to
\begin{align}
\eqref{eq:d12-phase3}=p_{\cup\{1,3\}}&\left(nR_1\left(\frac{1}{p_{\cup\{1,3\}}}-\frac{1}{p_{\cup\{1,2,3\}}}\right)\right.\nonumber\\
&\hspace{.8cm}\left.-nR_2\left(\frac{1}{p_{\cup\{2,3\}}}-\frac{1}{p_{\cup \{1,2,3\}}}\right)\right).\label{eq:d12-phase3-sim}
\end{align}
There are
\begin{align}&nR_1\frac{p_{3\overline{12}}}{p_{\cup\{1,2,3\}}}
-nR_3\frac{p_{1\overline{23}}}{p_{\cup\{1,2,3\}}}\frac{p_{\cup\{1,2\}}}{p_{\cup \{2,3\}}}\label{eq:d13-phase3}
\end{align}
number of $Q_{1;3\overline{2}}$ packets in the beginning of Phase~3, where the first term is the number of $Q_{1;3\overline{2}}$ packets generated in Phase~1.1 and the second term corresponds to the number of $Q_{1;3\overline{2}}$ packets that are used up in Phase~2.2. \eqref{eq:d13-phase3} can be simplified to
\begin{align}
\eqref{eq:d13-phase3}=p_{\cup\{1,2\}}&\left(nR_1\left(\frac{1}{p_{\cup\{1,2\}}}-\frac{1}{p_{\cup\{1,2,3\}}}\right)\right.\nonumber\\
&\hspace{.8cm}\left.-nR_3\left(\frac{1}{p_{\cup\{2,3\}}}-\frac{1}{p_{\cup \{1,2,3\}}}\right)\right).\label{eq:d13-phase3-sim}
\end{align}

We are now ready to describe Phase~3, which contains 3 sub-phases.

{\bf Phase 3.1}: Similar to Phase~2.1, we choose $T=\{2,3\}$ for the PE scheme. In  Phase~2.1, we chose the $Q_{2;3\overline{1}}$ packets $X_{2,j_2}$ and the $Q_{3;2\overline{1}}$ packets $X_{3,j_3}$ satisfying $S(X_{2,j_2})=\{3\}$ and $S(X_{3,j_3})=\{2\}$. Since we have already used up all $Q_{3;2\overline{1}}$ packets in Phase~2.1, in Phase~3.1, we choose the $Q_{2;3\overline{1}}$ packets $X_{2,j_2}$ and the new $Q_{3;12}$ packets $X_{3,j_3}$ instead, such that the packets satisfy $S(X_{2,j_2})=\{3\}$ and $S(X_{3,j_3})=\{1,2\}$. Similar to Phase~2.1, we switch to the next-in-line packet as long as the $S(X_{2,j_2})$ (or $S(X_{3,j_3})$) is changed.  Again, the decision whether to switch from $X_{2,j_2}$ to the next-in-line packet $X_{2,j_2'}$ is independent from the decision whether to switch from $X_{3,j_3}$ to the next-in-line packet $X_{3,j_3'}$.

Note that, by Line~\ref{line:S-update} of the {\sc Update}, the $S(X_{2,j_2})$ of a $Q_{2;3\overline{1}}$ packet $X_{2,j_2}$ will change if and only if it is {\sf received} by any one of $\{d_1,d_2\}$. Therefore, in average each $Q_{2;3\overline{1}}$ packet $X_{2,j_2}$ will take $\frac{1}{p_{\cup\{1,2\}}}$ number of time slots before we switch to the next-in-line packet $X_{2,j_2'}$. For comparison,
the $S(X_{3,j_3})$ of a $Q_{3;12}$ packet $X_{3,j_3}$ will change if and only if it is {\sf received} by $\{d_3\}$. Therefore, in average each $Q_{3;12}$ packet $X_{3,j_3}$ will take $\frac{1}{p_3}$ number of time slots before we switch to the next-in-line packet $X_{3,j_3'}$. We continue Phase~3.1 until we have finished all $Q_{2;3\overline{1}}$ packets. It is possible that we finish the $Q_{3;12}$ packets before finishing the $Q_{2;3\overline{1}}$ packets. In this case, we do not need to transmitting any $Q_{3;12}$ packets anymore and we use a degenerate $T=\{2\}$ instead and continue Phase~3.1 by only choosing $Q_{2;3\overline{1}}$ packets $X_{2,j_2}$. Intuitively, Phase~3.1 is a clean-up phase that finishes the $Q_{2;3\overline{1}}$ packets that have not been used in Phase~2.1. While finishing up $Q_{2;3\overline{1}}$ packets, we also piggyback some $Q_{3;12}$ packets through network coding. If all $Q_{3;12}$ packets have been used up, then we continue sending pure $Q_{2;3\overline{1}}$ packets without mixing together any $Q_{3;12}$ packets.

Since we have \eqref{eq:d23-phase3-sim} number of $Q_{2;3\overline{1}}$ packets to begin with, it will take
\begin{align}
\frac{\eqref{eq:d23-phase3-sim}}{p_{\cup\{1,2\}}}=&nR_2\left(\frac{1}{p_{\cup\{1,2\}}}-\frac{1}{p_{\cup\{1,2,3\}}}\right)\nonumber\\
&-nR_3\left(\frac{1}{p_{\cup\{1,3\}}}-\frac{1}{p_{\cup \{1,2,3\}}}\right)\label{eq:Phase3-1}
\end{align}
number of time slots to finish Phase~3.1.

{\em Remark:} When we transmit a $Q_{2;3\overline{1}}$ packet $X_{2,j_2}$, the new $S(X_{2,j_2})$ becomes $\{1,3\}$ if and only if $S_\text{rx}=\{1\}$ (i.e., only $d_1$ {\sf receives} $X_{2,j_2}$). Therefore Phase~3.1 will also create some new $Q_{2;13}$ packets. After Phase~3.1, the number of $Q_{2;13}$ packets is changed from \eqref{eq:d213-phase3} to
\begin{align}nR_2\left(\frac{p_{13\overline{2}}}{p_{\cup\{1,2,3\}}}
+\frac{p_{1\overline{23}}}{p_{\cup\{1,2,3\}}}\frac{f_p(3\overline{2})}{p_{\cup\{2,3\}}}
+\frac{p_{3\overline{12}}}{p_{\cup\{1,2,3\}}}\frac{f_p(1\overline{2})}{p_{\cup \{1,2\}}}\right),\label{eq:d213-phase3-new}
\end{align}
where the first, second, and the third terms correspond to the $Q_{2;13}$ packets generated in Phase~1.3, Phase~2.3, and Phase~2.1 plus Phase~3.1, respectively. We can further simplify \eqref{eq:d213-phase3-new} as
\begin{align}
\eqref{eq:d213-phase3-new}=nR_2p_2\left(\frac{1}{p_2}-\frac{1}{p_{\cup\{1,2\}}}-\frac{1}{p_{\cup\{2,3\}}}
+\frac{1}{p_{\cup\{1,2,3\}}}\right).\label{eq:d213-phase3-sim}
\end{align}

{\bf Phase 3.2}: After Phase~3.1, we move to Phase~3.2. Similar to Phase~3.1, Phase~3.2 serves the role of cleaning up the $Q_{1;3\overline{2}}$ packets that have not been used in Phase~2.2. More explicitly,
we choose $T=\{1,3\}$, and use the $Q_{1;3\overline{2}}$ packets $X_{1,j_1}$ and the new $Q_{3;12}$ packets $X_{3,j_3}$, such that the packets satisfy $S(X_{1,j_1})=\{3\}$ and $S(X_{3,j_3})=\{1,2\}$. It is possible that all $Q_{3;12}$ packets have been used up in Phase~3.1.  In this case, we do not need to transmitting any $Q_{3;12}$ packets anymore and we use a degenerate $T=\{1\}$ instead and continue Phase~3.1 by only choosing $Q_{1;3\overline{2}}$ packets $X_{1,j_1}$.

Similar to all previous phases, we switch to the next-in-line packet as long as the $S(X_{1,j_1})$ (or $S(X_{3,j_3})$) is changed, and the decision whether to switch from $X_{1,j_1}$ to the next-in-line packet $X_{1,j_1'}$ is independent from the decision whether to switch from $X_{3,j_3}$ to the next-in-line packet $X_{3,j_3'}$. We continue Phase~3.2 until we have finished all $Q_{1;3\overline{2}}$ packets. Again, if we finish the $Q_{3;12}$ packets before finishing the $Q_{1;3\overline{2}}$ packets, then we stop transmitting any $Q_{3;12}$ packets, use a degenerate $T=\{1\}$ instead, and continue Phase~3.2 by only choosing $Q_{1;3\overline{2}}$ packets $X_{1,j_1}$.

By Line~\ref{line:S-update} of the {\sc Update}, the $S(X_{1,j_1})$ of a $Q_{1;3\overline{2}}$ packet $X_{1,j_1}$ will change if and only if it is {\sf received} by any one of $\{d_1,d_2\}$. Therefore, in average each $Q_{1;3\overline{2}}$ packet $X_{1,j_1}$ will take $\frac{1}{p_{\cup\{1,2\}}}$ number of time slots before we switch to the next-in-line packet $X_{1,j_1'}$.
Since we have \eqref{eq:d13-phase3-sim} number of $Q_{1;3\overline{2}}$ packets to begin with, it will take
\begin{align}
\frac{\eqref{eq:d13-phase3-sim}}{p_{\cup\{1,2\}}}=&nR_1\left(\frac{1}{p_{\cup\{1,2\}}}-\frac{1}{p_{\cup\{1,2,3\}}}\right)\nonumber\\
&-nR_3\left(\frac{1}{p_{\cup\{2,3\}}}-\frac{1}{p_{\cup \{1,2,3\}}}\right)\label{eq:Phase3-2}
\end{align}
number of time slots to finish Phase~3.2.

{\bf Phase 3.3}: After Phase~3.2, we move to Phase~3.3. Similar to Phases~3.1 and~3.2, Phase~3.3 serves the role of cleaning up the $Q_{1;2\overline{3}}$ packets that have not been used in Phase~2.3. More explicitly,
we choose $T=\{1,2\}$, and use the $Q_{1;2\overline{3}}$ packets $X_{1,j_1}$ and the new $Q_{2;13}$ packets $X_{2,j_2}$, such that the packets satisfy $S(X_{1,j_1})=\{2\}$ and $S(X_{2,j_2})=\{1,3\}$. Recall that in the beginning of Phase~3.3, we have \eqref{eq:d213-phase3-sim} number of $Q_{2;13}$ packets.

Similar to all previous phases, we switch to the next-in-line packet as long as the $S(X_{1,j_1})$ (or $S(X_{2,j_2})$) is changed, and the decision whether to switch from $X_{1,j_1}$ to the next-in-line packet $X_{1,j_1'}$ is independent from the decision whether to switch from $X_{2,j_2}$ to the next-in-line packet $X_{2,j_2'}$. We continue Phase~3.3 until we have finished all $Q_{1;2\overline{3}}$ packets. If we finish the $Q_{2;13}$ packets before finishing the $Q_{1;2\overline{3}}$ packets, then we stop transmitting any $Q_{2;13}$ packets, use a degenerate $T=\{1\}$ instead, and continue Phase~3.3 by only choosing $Q_{1;2\overline{3}}$ packets $X_{1,j_1}$.

By Line~\ref{line:S-update} of the {\sc Update}, the $S(X_{1,j_1})$ of a $Q_{1;2\overline{3}}$ packet $X_{1,j_1}$ will change if and only if it is {\sf received} by any one of $\{d_1,d_3\}$. Therefore, in average each $Q_{1;2\overline{3}}$ packet $X_{1,j_1}$ will take $\frac{1}{p_{\cup\{1,3\}}}$ number of time slots before we switch to the next-in-line packet $X_{1,j_1'}$.
Since we have \eqref{eq:d12-phase3-sim} number of $Q_{1;2\overline{3}}$ packets to begin with, it will take
\begin{align}
\frac{\eqref{eq:d12-phase3-sim}}{p_{\cup\{1,3\}}}=&nR_1\left(\frac{1}{p_{\cup\{1,3\}}}-\frac{1}{p_{\cup\{1,2,3\}}}\right)\nonumber\\
&-nR_2\left(\frac{1}{p_{\cup\{2,3\}}}-\frac{1}{p_{\cup \{1,2,3\}}}\right)\label{eq:Phase3-3}
\end{align}
number of time slots to finish Phase~3.3.

{\bf Phase~4:} We first summarize the status of all the packets in the end of Phase~3.3. For $d_3$, all the $X_{3,j}$ packets are either {\sf received} by $d_3$ (i.e., having $3\in S(X_{3,j})$) or have $S(X_{3,j})=\{1,2\}$, the $Q_{3;12}$ packets. 
By  Line~\ref{line:S-update} of the {\sc Update}, the $S(X_{3,j_3})$ of a $Q_{3;12}$ packet $X_{3,j_3}$ will change if and only if it is {\sf received} by $d_3$. Therefore, in average each $Q_{3;12}$ packet $X_{3,j_3}$ will take $\frac{1}{p_{3}}$ number of time slots before we switch to the next-in-line packet $X_{3,j_3'}$.
 Since the $Q_{3;12}$ packets participate in Phases~3.1 and~3.2, in the end of Phase~3.3, the total number of $Q_{3;12}$ packets becomes
\begin{align}
\left(\text{Eq.}\eqref{eq:d3-phase3-sim}-p_3\cdot \text{Eq.}\eqref{eq:Phase3-1}-p_3\cdot \text{Eq.}\eqref{eq:Phase3-2}\right)^+,\label{eq:d3-phase4}
\end{align}
where $(\cdot)^+=\max(\cdot, 0)$ is the projection to the non-negative reals.

For $d_2$, all $X_{2,j}$ packets that have $S(X_{2,j})=\emptyset$ or $S(X_{2,j})=\{1\}$ have been used up in Phase~1.2 or Phase~2.3, respectively. All $X_{2,j}$ packets that have $S(X_{2,j})=\{3\}$ have been used up in Phases~2.1 and~3.1. As a result,
all the $X_{2,j}$ packets
are either {\sf received} by $d_2$ (i.e., having $2\in S(X_{2,j})$) or have $S(X_{2,j})=\{1,3\}$, the $Q_{2;13}$ packets.
By Line~\ref{line:S-update} of the {\sc Update}, the $S(X_{2,j_2})$ of a $Q_{2;13}$ packet $X_{2,j_2}$ will change if and only if it is {\sf received} by $d_2$. Therefore, in average each $Q_{2;13}$ packet $X_{2,j_2}$ will take $\frac{1}{p_{2}}$ number of time slots before we switch to the next-in-line packet $X_{2,j_2'}$.
 Since the $Q_{2;13}$ packets also participate in Phase~3.3, in the end of Phase~3.3, the total number of $Q_{2;13}$ packets becomes
\begin{align}
\left(\text{Eq.}\eqref{eq:d213-phase3-sim}-p_2\cdot \text{Eq.}\eqref{eq:Phase3-3}\right)^+.\label{eq:d2-phase4}
\end{align}

For $d_1$, all $X_{1,j}$ packets that have $S(X_{1,j})=\emptyset$, $S(X_{1,j})=\{2\}$, and $S(X_{1,j})=\{3\}$ have been used up in Phases~1.1, 2.3+3.3, and 2.2+3.2, respectively. As a result,
all the $X_{1,j}$ packets
are either {\sf received} by $d_1$ (i.e., having $1\in S(X_{1,j})$) or have $S(X_{1,j})=\{2,3\}$, the $Q_{1;23}$ packets.
In the end of Phase~3.3, the total number of $Q_{1;23}$ packets is
\begin{align}nR_1\left(\frac{p_{23\overline{1}}}{p_{\cup\{1,2,3\}}}
+\frac{p_{2\overline{13}}}{p_{\cup\{1,2,3\}}}\frac{f_p(3\overline{1})}{p_{\cup\{1,3\}}}
+\frac{p_{3\overline{12}}}{p_{\cup\{1,2,3\}}}\frac{f_p(2\overline{1})}{p_{\cup \{1,2\}}}\right),\label{eq:d123-phase4}
\end{align}
where the first, second, and the third terms correspond to the $Q_{1;23}$ packets generated in Phase~1.1, 2.3+3.3, and 2.2+3.2, respectively. We can further simplify \eqref{eq:d123-phase4} as
\begin{align}
\eqref{eq:d123-phase4}=nR_1p_1\left(\frac{1}{p_1}-\frac{1}{p_{\cup\{1,2\}}}-\frac{1}{p_{\cup\{1,3\}}}
+\frac{1}{p_{\cup\{1,2,3\}}}\right).\label{eq:d1-phase4}
\end{align}

In Phase~4, since the only remaining packets (that still need to be retransmitted, see Lemma~\ref{lem:decodability}) are the $Q_{1;23}$, $Q_{2;13}$, and $Q_{3;12}$ packets, we always choose $T=\{1,2,3\}$ and randomly and linearly mix the $Q_{1;23}$, $Q_{2;13}$, and $Q_{3;12}$ packets (one from each queue) for each time slot. That is, we use Phase~4 to clean up the remaining packets. Since in average a $Q_{i;\{1,2,3\}\backslash i}$ packet $X_{i,j}$ takes $\frac{1}{p_i}$ amount of time before it is received by $d_i$, Phase~4 thus takes
\begin{align}
\max\left(\frac{\text{Eq.\eqref{eq:d1-phase4}}}{p_1},\frac{\text{Eq.\eqref{eq:d2-phase4}}}{p_2}, \frac{\text{Eq.\eqref{eq:d3-phase4}}}{p_3}\right). \label{eq:total-phase4}
\end{align}
number of time slots to finish. More precisely, as time proceeds, we need to gradually switch to a degenerate $T$. For example, if the $Q_{2;13}$ packets are used up first, then we set the new $T=\{1,3\}$ and focus on mixing the remaining $Q_{1;23}$ and $Q_{3;12}$ packets. After \eqref{eq:total-phase4} number of time slots, it is thus guaranteed that for sufficiently large $n$, all information packets $X_{k,j}$, $k\in\{1,2,3\}$, and $j\in [nR_k]$ satisfy $k\in S(X_{k,j})$. By Corollary~\ref{cor:dec}, all $d_k$ can decode the desired packets $X_{k,j}$, $j\in [nR_k]$ with close-to-one probability.

{\bf Quantify the throughput of the 4-phase scheme:} The remaining task is to show that if $(R_1,R_2,R_3)$ is in the interior of the outer bound in Proposition~\ref{prop:outer}, then the total number of time slots used by the above 4-Phase PE scheme is within the time budget $n$ time slots.  That is, we need to prove that
\begin{align}
&\eqref{eq:phase11}+\eqref{eq:phase12}+\eqref{eq:phase13}+\eqref{eq:d3-phase2.1}+
\eqref{eq:d3-phase2.2}\nonumber\\
&+\eqref{eq:d2-phase2.3}+\eqref{eq:Phase3-1}+\eqref{eq:Phase3-2}+\eqref{eq:Phase3-3}+\eqref{eq:total-phase4}\leq n.\label{eq:final1}
\end{align}
The summation of the first nine terms of the left-hand side of \eqref{eq:final1} can be simplified to
\begin{align}
A_{\text{1.1--3.3}}\stackrel{\Delta}{=}&nR_1\left(\frac{1}{p_{\cup\{1,2\}}}+\frac{1}{p_{\cup\{1,3\}}}-\frac{1}{p_{\cup\{1,2,3\}}}\right)\nonumber\\
&+nR_2\frac{1}{p_{\cup\{1,2\}}}+nR_3\frac{1}{p_{\cup\{1,2,3\}}}, \nonumber
\end{align}
where $A_{\text{1.1--3.3}}$ is the total number of time slots in Phases~1.1 to 3.3. Since \eqref{eq:total-phase4} is the maximum of three terms, proving \eqref{eq:final1} is thus equivalent to proving that the following three inequality hold simultaneously.
\begin{align}
&A_{\text{1.1--3.3}}+\frac{\eqref{eq:d1-phase4}}{p_1}\leq n,\nonumber\\
&A_{\text{1.1--3.3}}+\frac{\eqref{eq:d2-phase4}}{p_2}\leq n,\nonumber\\
\text{and }&A_{\text{1.1--3.3}}+\frac{\eqref{eq:d3-phase4}}{p_3}\leq n.\nonumber
\end{align}
With direct simplification of the expressions, proving the above three inequalities is equivalent to proving
\begin{align}
&\frac{nR_1}{p_1}+\frac{nR_2}{p_{\cup\{1,2\}}}+\frac{nR_3}{p_{\cup\{1,2,3\}}}\leq n,\nonumber\\
&\frac{nR_1}{p_{\cup\{1,2\}}}+\frac{nR_2}{p_{2}}+\frac{nR_3}{p_{\cup\{1,2,3\}}}\leq n,\nonumber\\
\text{and }&\frac{nR_1}{p_{\cup\{1,3\}}}+\frac{nR_2}{p_{\cup\{1,2,3\}}}+\frac{nR_3}{p_{3}}\leq n, \nonumber
\end{align}
which hold for any $(R_1,R_2,R_3)$ in the interior of the capacity outer bound in Proposition~\ref{prop:outer}. More rigorously, by the law of large numbers, the expressions of the numbers of time slots in Phase~1.1 to Phase 4:
\eqref{eq:phase11}, \eqref{eq:phase12}, \eqref{eq:phase13}, \eqref{eq:d3-phase2.1},
\eqref{eq:d3-phase2.2}, \eqref{eq:d2-phase2.3}, \eqref{eq:Phase3-1}, \eqref{eq:Phase3-2}, \eqref{eq:Phase3-3}, and \eqref{eq:total-phase4}, are all of precision $o(n)$. Since $(R_1,R_2,R_3)$ is in the interior of the capacity outer bound in Proposition~\ref{prop:outer}, the last three inequalities hold with arbitrarily close to one probability for sufficiently large $n$. The proof of Proposition~\ref{prop:cap3} is thus complete.

\subsection{Achieving the Capacity of 1-to-3 Broadcast PECs With COF --- High-Level Discussion\label{subsec:high-level-cap3}}

As discussed in Section~\ref{subsec:detailed3}, one advantage of a PE scheme is that although different packets $X_{k,j_k}$  and $X_{i,j_i}$ with $k\neq i$ may be mixed together, the corresponding evolution of $X_{k,j_k}$ (the changes of $S(X_{k,j_k})$ and $\vv(X_{k,j_k})$)  are independent from the evolution of $X_{i,j_i}$. Also by Lemma~\ref{lem:non-interfering}, two different packets $X_{k,j_k}$  and $X_{i,j_i}$ can share the same time slot without interfering each other as long as $i\in S(X_{k,j_k})$ and $k\in S(X_{i,j_i})$. These two observations enable us to convert the achievability problem of a PE scheme to the following ``time slot packing problem."


\begin{figure}
\centering
\includegraphics[width=8cm]{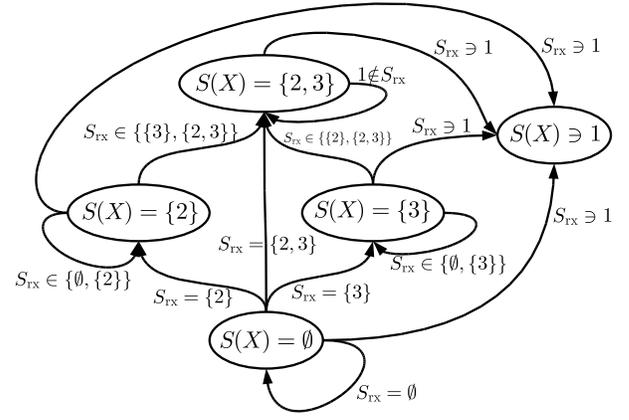}
\caption{The state transition diagram for destination $d_1$ when applying the packet evolution scheme to a 1-to-3 broadcast PEC.\label{fig:trans}}
\end{figure}

Let us focus on the $(s,d_1)$ session. 
For any $X_{1,j}$ packet, initially $S(X_{1,j})=\emptyset$. Then as time proceeds, each $X_{1,j}$ starts to participate in packet transmission. The corresponding $S(X_{1,j})$ evolves to different values, depending on the set of destinations that {receive} the transmitted packet in which $X_{1,j}$ participates. Since in this subsection we focus mostly on $S(X_{1,j})$, we sometimes use $S(X)$ as shorthand if it is unambiguous from the context.
 Fig.~\ref{fig:trans} describes how $S(X)$ evolves between different values.  In Fig.~\ref{fig:trans}, we use circles to represent the five different states according to the $S(X)$ value. Recall that $S_\text{rx}$ is the set of destinations who successfully receive the transmitted coded packet. The receiving set $S_\text{rx}$ decides the transition between different states. In Fig.~\ref{fig:trans}, we thus mark each transition arrow (between different states) by the value(s) of $S_\text{rx}$ that enables the transition. For example, by Line~\ref{line:S-update} of the {\sc Update}, when the initial state is $S(X)=\emptyset$, if the receiving set $S_\text{rx}\ni 1$, then the new set satisfies $S(X)\ni 1$. Similarly, when the initial state is $S(X)=\emptyset$, if $S_\text{rx}=\{2,3\}$, then the new $S(X)$ becomes $S(X)=\{2,3\}$.
(Note that the corresponding $\vv(X_{1,j})$ also evolves over time to maintain the non-interfering property in Lemma~\ref{lem:non-interfering}, which is not illustrated in Fig.~\ref{fig:trans}.)

Since $S(X_{1,j})\ni 1$ if and only if $d_1$ {\sf receives} $X_{1,j}$, it thus takes $\frac{nR_1}{p_1}$ {\em logical} time slots to finish the transmission of $nR_1$ information packets. On the other hand, some logical time slots for the $(s,d_1)$ session can be ``packed/shared" jointly with the logical time slots for the $(s,d_{k})$ session, $k\neq 1$, or, equivalently, one physical time slot can serve two sessions simultaneously. For the following, we quantify how many logical time slots of the $(s,d_1)$ session are {\em compatible} to those of other sessions. For any $S_0\in 2^{\{1,2,3\}}$, let $A_{1;S_0}$ denote the number of logical time slots (out of the total $\frac{nR_1}{p_1}$ time slots) such that during those time slots, the transmitted $X_{1,j}$ has $S(X_{1,j})=S_0$. Initially, there are $nR_1$ packets $X_{1,j}$. If any one of $\{d_1,d_2,d_3\}$ {\sf receives} the transmitted packet (equivalently $S_\text{rx}\neq \emptyset$), $S(X_{1,j})$ becomes non-empty. Therefore, each $X_{1,j}$ contributes to $\frac{1}{p_{\cup\{1,2,3\}}}$ logical time slots with $S(X_{1,j})=\emptyset$. We thus have
\begin{align}
A_{1;\emptyset}= nR_1\left(\frac{1}{p_{\cup\{1,2,3\}}}\right).\label{eq:A0}
\end{align}
We also note that during the evolution process of $X_{1,j}$, if any one of $\{d_1,d_3\}$ {\sf receives} the transmitted packet (equivalently $S_\text{rx}\cap \{1,3\}\neq \emptyset$), then $S(X)$ value will move from one of the two states ``$S(X)=\emptyset$" and ``$S(X)=\{2\}$" to one of the three states ``$S(X)=\{3\}$," ``$S(X)=\{2,3\}$," and ``$S(X)\ni 1$."
Therefore, each $X_{1,j}$ contributes to $\frac{1}{p_{\cup\{1,3\}}}$ logical time slots for which we either have $S(X_{1,j})=\emptyset$ or $S(X_{1,j})=\{2\}$. By the above reasoning, we have
\begin{align}
A_{1;\{2\}}+A_{1;\emptyset}= nR_1\left(\frac{1}{p_{\cup\{1,3\}}}\right).\label{eq:A2}
\end{align}
Similarly, during the evolution process of $X_{1,j}$, if any one of $\{d_1,d_2\}$ {\sf receives} the transmitted packet (equivalently $S_\text{rx}\cap \{1,2\}\neq \emptyset$), then $S(X)$ value will move from one of the two states ``$S(X)=\emptyset$" and ``$S(X)=\{3\}$" to one of the three states ``$S(X)=\{2\}$," ``$S(X)=\{2,3\}$," and ``$S(X)\ni 1$."
 Therefore, each $X_{1,j}$ contributes to $\frac{1}{p_{\cup\{1,2\}}}$ logical time slots for which either $S(X_{1,j})=\emptyset$ or $S(X_{1,j})=\{3\}$.
By the above reasoning, we have
\begin{align}
A_{1;\{3\}}+A_{1;\emptyset}= nR_1\left(\frac{1}{p_{\cup\{1,2\}}}\right).\label{eq:A3}
\end{align}
Before $S(X)$ evolves to the state ``$S(X)\ni 1$," any logical time slot contributed by such an $X$ must have one of the following four states:
``$S(X)=\emptyset$," ``$S(X)=\{2\}$," ``$S(X)=\{3\}$," and ``$S(X)=\{2,3\}$."  As a result, we must have
\begin{align}
A_{1;\{2,3\}}+A_{1;\{2\}}+A_{1;\{3\}}+A_{1;\emptyset}= nR_1\left(\frac{1}{p_{1}}\right).\label{eq:A23}
\end{align}
Solving \eqref{eq:A0}, \eqref{eq:A2}, \eqref{eq:A3}, and \eqref{eq:A23}, we have
\begin{align}
A_{1;\emptyset}&=nR_1\left(\frac{1}{p_{\cup\{1,2,3\}}}\right)\label{eq:A1-0-num}\\
A_{1;\{2\}}&=nR_1\left(\frac{1}{p_{\cup\{1,3\}}}-\frac{1}{p_{\cup\{1,2,3\}}}\right)\\
A_{1;\{3\}}&=nR_1\left(\frac{1}{p_{\cup\{1,2\}}}-\frac{1}{p_{\cup\{1,2,3\}}}\right)\\
A_{1;\{2,3\}}&=nR_1\left(\frac{1}{p_1}-\frac{1}{p_{\cup\{1,2\}}}-\frac{1}{p_{\cup\{1,3\}}}+\frac{1}{p_{\cup\{1,2,3\}}}\right).
\end{align}
We can also define $A_{k;S_0}$ as the number of logical time slots of the $(s,d_k)$ session with $S(X_{k,j_k})=S_0$. 
By similar derivation arguments, we have
\begin{align}
A_{2;\emptyset}&=nR_2\left(\frac{1}{p_{\cup\{1,2,3\}}}\right)\\
A_{2;\{1\}}&=nR_2\left(\frac{1}{p_{\cup\{2,3\}}}-\frac{1}{p_{\cup\{1,2,3\}}}\right)\\
A_{2;\{3\}}&=nR_2\left(\frac{1}{p_{\cup\{1,2\}}}-\frac{1}{p_{\cup\{1,2,3\}}}\right)\\
A_{2;\{1,3\}}&=nR_2\left(\frac{1}{p_2}-\frac{1}{p_{\cup\{1,2\}}}-\frac{1}{p_{\cup\{2,3\}}}+\frac{1}{p_{\cup\{1,2,3\}}}\right).
\end{align}
and
\begin{align}
A_{3;\emptyset}&=nR_3\left(\frac{1}{p_{\cup\{1,2,3\}}}\right)\\
A_{3;\{1\}}&=nR_3\left(\frac{1}{p_{\cup\{2,3\}}}-\frac{1}{p_{\cup\{1,2,3\}}}\right)\\
A_{3;\{2\}}&=nR_3\left(\frac{1}{p_{\cup\{1,3\}}}-\frac{1}{p_{\cup\{1,2,3\}}}\right)\\
A_{3;\{1,2\}}&=nR_3\left(\frac{1}{p_3}-\frac{1}{p_{\cup\{1,3\}}}-\frac{1}{p_{\cup\{2,3\}}}+\frac{1}{p_{\cup\{1,2,3\}}}\right).\label{eq:A3-12-num}
\end{align}

Recall that by definition, $A_{k;S_0}$ is the number of logical time slots of the $(s,d_k)$ session that is compatible to the logical time slots of $(s,d_i)$ session with $i\in S_0$.
The achievability problem of a PE scheme thus becomes the following time slot packing problem.

\begin{quote}
Consider 12 types of logical time slots and each type is denoted by $(k;S_0)$ for some $k\in\{1,2,3\}$, $S_0\in 2^{\{1,2,3\}}$, and $k\notin S_0$. The numbers of logical time slots of each type are described in \eqref{eq:A1-0-num} to \eqref{eq:A3-12-num}. Two logical time slots of types $(k_1;S_1)$ and $(k_2;S_2)$ are {\em compatible} if $k_1\neq k_2$, $k_1\in S_2$, and $k_2\in S_1$. Any compatible logical time slots can be packed together in the same physical time slot. For example, consider the following types of logical time slots: $(1;\{2,3\})$, $(2;\{1,3\})$, and $(3;\{1,2\})$. Three logical time slots, one from each type, can occupy the same physical time slot since any two of them are compatible to each other. The time slot packing problem is thus: {\em Can we pack all the logical time slots within $n$ physical time slots?}
\end{quote}

\begin{figure*}
\centering
\includegraphics[width=14cm]{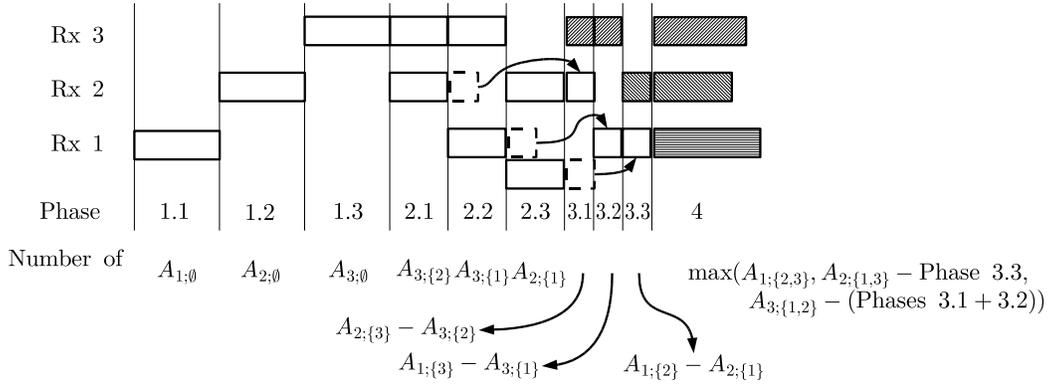}
\caption{The time-slot packing policy that corresponds to the 4-Phase solution for 1-to-3 broadcast PECs. The shaded rectangles represent the logical time slots of types $(1;\{2,3\})$, $(2;\{1,3\})$, and $(3;\{1,2\})$. \label{fig:pack}}
\end{figure*}

The detailed 4-phase PE scheme in Section~\ref{subsec:detailed3} thus corresponds to the time-slot-packing policy depicted in Fig.~\ref{fig:pack}. Namely, we first use Phases~1.1 to 1.3 send all the logical time slots that cannot be packed with any other logical time slots. Totally, it takes $A_{1;\emptyset}+A_{2;\emptyset}+A_{3;\emptyset}$ number of time slots to finish Phases~1.1 to 1.3.  We then use Phases~2.1 to 2.3 to pack those logical time slots that can be packed with exactly one other logical time slot from a different session. By the assumption that $d_1$ dominates $d_2$ and $d_3$, and $d_2$ dominates $d_3$, we have $A_{1;\{2\}}\geq A_{2;\{1\}}$, $A_{1;\{3\}}\geq A_{3;\{1\}}$, and $A_{2;\{3\}}\geq A_{3;\{2\}}$. Therefore, it takes $A_{3;\{2\}}+A_{3;\{1\}}+A_{2;\{1\}}$ number of physical time slots to finish Phases~2.1 to 2.3.

Phases~3.1 to 3.3 are to clean up the remaining logical time slots of types $(2;\{3\})$, $(1;\{3\})$, and $(1;\{2\})$. We notice that in Phase~3.1 when sending a logical time slot of type $(2;\{3\})$, there is no type-$(3,\{2\})$ logical time slot that can be mixed together. On the other hand, there are still some type-$(3,\{1,2\})$ logical time slots, which can also be mixed with the logical time slots of the $(s,d_2)$ session. Therefore, when we send a logical time slot of type $(2;\{3\})$, the optimal way is to pack it with a type-$(3,\{1,2\})$ logical time slots together as illustrated in Phase~3.1 of Fig.~\ref{fig:pack}. It is worth emphasizing that although those type-$(3,\{1,2\})$ logical time slots can be packed with two other logical time slots simultaneously, there is no point to save the type-$(3,\{1,2\})$ logical time slots for future mixing. The reason is that when Phase~3.1 cleans up the remaining type-$(2;\{3\})$ logical time slots, it actually  provides a zero-cost free ride for any logical time slot that is compatible to a type-$(2;\{3\})$ logical time slot. Therefore, piggybacking a type-$(2;\{3\})$ logical time slot with a type-$(3,\{1,2\})$ logical time slot is optimal.
 Similarly, we also take advantage of the free ride by packing logical time slots of type-$(1;\{3\})$ with that of type-$(3;\{1,2\})$  in Phases~3.2, and by  packing logical time slots of type-$(1;\{2\})$ with that of type-$(2;\{1,3\})$  in Phases~3.3.
 It thus takes \begin{align}
(A_{2;\{3\}}-A_{3;\{2\}})+(A_{1;\{3\}}-A_{3;\{1\}})+(A_{1;\{2\}}-A_{2;\{1\}})\nonumber
\end{align}
number of time slots to finish Phases~3.1 to 3.3.

In Phase~4, we clean up and pack together all the remaining logical time slots of types $(1;\{2,3\})$, $(2;\{1,3\})$, and $(3;\{1,2\})$. We thus need
\begin{align}
\max&\left((A_{3;\{1,2\}}-(A_{2;\{3\}}-A_{3;\{2\}})-(A_{1;\{3\}}-A_{3;\{1\}}))^+,\right.\nonumber\\
&\left.(A_{2;\{1,3\}}-(A_{1;\{2\}}-A_{2;\{1\}}))^+, A_{1;\{2,3\}}\right)\label{eq:pack-max}
\end{align}
number of time slots to finish Phase~4. Depending on which of the three terms in \eqref{eq:pack-max} is the largest, the total number of physical time slots is one of the following three expressions:
\begin{align}
&A_{3;\emptyset}+A_{3;\{1\}}+A_{3;\{2\}}+A_{3;\{1,2\}}+A_{1;\emptyset}+A_{1;\{2\}}+A_{2;\emptyset},\nonumber\\
&A_{2;\emptyset}+A_{2;\{1\}}+A_{2;\{3\}}+A_{2;\{1,3\}}+A_{1;\emptyset}+A_{1;\{3\}}+A_{3;\emptyset},\nonumber\\
&\text{or }A_{1;\emptyset}+A_{1;\{2\}}+A_{1;\{3\}}+A_{1;\{2,3\}}+A_{2;\emptyset}+A_{2;\{3\}}+A_{3;\emptyset}.\nonumber
\end{align}
By \eqref{eq:A1-0-num} to \eqref{eq:A3-12-num}, one can easily check that all three equations are less than $n$ for any $(R_1,R_2,R_3)$ in the interior of the outer bound of Proposition~\ref{prop:outer}, which answers the time-slot-packing problem in an affirmative way.
One can also show that the packing policy in Fig.~\ref{fig:pack} is the tightest among any other packing policy, which indeed corresponds to the capacity-achieving PE scheme described in Section~\ref{subsec:detailed3}.

\subsection{The Achievability Results of General 1-to-$M$ Broadcast PECs With COF\label{subsec:1-to-M-ach}}

In Section~\ref{subsec:high-level-cap3}, we show how to reduce the achievability problem of a PE scheme to a time-slot-packing problem. However, the converse may not hold due to the causality constraint of the PE scheme. By taking into account the causality constraint, the time-slot-packing arguments can be used to generate new achievable rate inner bounds for general 1-to-$M$ broadcast PECs with COF, which will be discussed in this subsection.

One major difference between the tightest solution of the time-slot-packing problem in Fig.~\ref{fig:pack} and the detailed PE scheme in Section~\ref{subsec:detailed3} is that for the former, we can pack the time slots in any order. There is no need to first pack those logical time slots that cannot be shared with any other time slots. Any packing order will result in the same amount of physical time slots in the end. On the other hand, for the PE scheme it is critical to perform the 4 phases (10 sub-phases) in sequence since many packets used in the later phase are generated by the previous phases. For example, all the packets in Phases~2 to~4 are generated in Phases~1.1 to 1.3. Therefore it is imperative to conduct Phase~1 first before Phases~2 to~4. Similarly, the  $Q_{3;\{1,2\}}$ packets used in Phases~3.1 and~3.2 are generated in Phases~1.3, 2.1, and 2.2. Therefore, the number of  $Q_{3;\{1,2\}}$ packets in the end of Phase~1.3 (without those generated in Phases~2.1 and 2.2) may not be sufficient for mixing with $Q_{1;\{3\}}$ packets. As a result, it can be suboptimal to perform Phase~3.1 before Phases~2.1 and~2.2.

The causality constraints for a 1-to-$M$ PEC with $M\geq 4$ quickly become complicated due to the potential {\em cyclic} dependence\footnote{For general 1-to-$M$ PECs with $M\geq 4$, we may have the following cyclic dependence relationship: Packet mixing in Phase~A needs to use the packets generated by the packet mixing during Phase~B. Packing mixing in Phase~B needs the packets resulted from the packet mixing during Phase~C. But the packing mixing of Phase~C also needs the packets resulted from the packing mixing in Phase~A. Quantifying such a cyclic dependence relationship with causality constraints is a complicated problem.} of the problem. To simplify the derivation, we consider the following {\em sequential acyclic construction of PE schemes}, which allows tractable performance analysis but at the cost of potentially being  throughput suboptimal. As will be seen in Section~\ref{subsec:sim}, for most PEC channel parameters, the proposed sequential acyclic PE schemes are sufficient to achieve the channel capacity.

For the following, we describe the sequential PE schemes. The main feature of the sequential PE scheme is that we choose the mixing set $T$ in a sequential, acyclic fashion. For comparison, the $T$ parameters used in the capacity-achieving 1-to-3 PE scheme of Section~\ref{subsec:detailed3} are $\{1\}$, $\{2\}$, $\{3\}$, $\{2,3\}$, $\{1,3\}$, $\{1,2\}$, $\{2,3\}$, $\{1,3\}$, $\{1,2\}$, and $\{1,2,3\}$ in Phases~1.1 to 4, respectively. We notice that $T=\{1,2\}$ is visited twice in Phases 2.3 and 3.3. We thus call the capacity-achieving PE scheme a {\em cyclic} PE scheme. For the sequential PE schemes, we never revisit any $T$ value during all the phases.

To design a sequential PE scheme, we first observe that in the capacity-achieving 4-Phase PE scheme in Section~\ref{subsec:detailed3}, we always start from mixing a small subset $T$ then gradually move to mixing a larger subset $T$. The intuition behind is that when mixing a small set, say $T=\{2,3\}$ in Phase~2.1, we can create more coding opportunities in the later Phase~4 when $T=\{1,2,3\}$. Recall the definition of {\em cardinality-compatible} total ordering $\prec$ on $2^{[K]}$ in \eqref{eq:card-ordering}. For a sequential PE scheme, we thus choose the mixing set $T$ from the smallest to the largest according to the given cardinality-compatible total ordering. The detailed algorithm of choosing $T$ and the target packets $X_{k,j_k}$, $k\in T$, is described as follows.

There are $(2^K-1)$ phases and each phase is indexed by a non-empty subset $T\subseteq [K]$. We move sequentially between phases according to the cardinality-compatible total ordering $\prec$. That is, if $T_1\prec T_2$ and there is no other subset $T_3$ satisfying $T_1\prec T_3\prec T_2$, then after the completion of Phase~$T_1$, we move to Phase~$T_2$.

Consider the operation in Phase~$T$. Recall that the basic properties of the PE scheme allow us to choose the target packets $X_{k,j_k}$ independently for all $k\in T$. In Phase~$T$, consider a fixed $k\in T$. Let $S_{k}=T\backslash k$. We first choose a $Q_{k;S_k}$ packet $X_{k,j_k}$, i.e., those with $S(X_{k,j_k})=S_k$, and keep using this packet for transmission, which will be mixed with packets from other sessions according to Line~\ref{line:vv-construct} of the {PE} scheme. Whenever the current $X_{k,j_k}$ packet evolves (the corresponding $S(X_{k,j_k})$ changes), we move to the next $Q_{k;S_k}$ packet $X_{k,j_k'}$. Continue this process for a pre-defined amount of time slots. We use $w_{k;S_k\rightarrow S_k}$ to denote the number of time slots in which we choose a $Q_{k;S_k}$ packet. After $w_{k;S_k\rightarrow S_k}$ number of time slots, we are still in Phase~$T$ but we will start to choose a different $Q_{k;\tilde{S}_{k}}$ packet $X_{k,j_k}$ (i.e., with $S(X_{k,j_k})=\tilde{S}_{k}$), which will be mixed with packets from other sessions in $T$.
More explicitly, we choose a sequence of $\tilde{S}_{k}$ such that all $\tilde{S}_{k}$ satisfy $S_k\subseteq \tilde{S}_{k}\subseteq ([K]\backslash k)$, which guarantees that such new $X_{k,j_k}$ with $S(X_{k,j_k})=\tilde{S}_{k}$ is still non-interfering from the perspectives of all other sessions in $T$. The order we choose the $\tilde{S}_{k}$ follows that of the total ordering $\prec$. The closer $\tilde{S}_{k}$ is to $S_k$, the earlier we use such $\tilde{S}_{k}$.

For any chosen $\tilde{S}_{k}$, we choose a $Q_{k;\tilde{S}_{k}}$ packet $X_{k,j_k}$, i.e., those with $S(X_{k,j_k})=\tilde{S}_{k}$, and keep using this packet to generate coded packets for transmission. Whenever the current $X_{k,j_k}$ packet evolves (the corresponding $S(X_{k,j_k})$ changes), we move to the next $Q_{k;\tilde{S}_{k}}$ packet $X_{k,j_k'}$. Continue this process for a pre-defined amount of time slots. We use $w_{k;\tilde{S}_{k}\rightarrow S_k}$ to denote the number of time slots in which we choose a $Q_{k;\tilde{S}_{k}}$ packet. That is, $w_{k;\tilde{S}_{k}\rightarrow S_k}$ is the number of time slots that we are using a $Q_{k;\tilde{S}_{k}}$ packet in substitute for a $Q_{k;S_{k}}$ packet, which is similar to the operations in Phases~3.1 to~3.3. After $w_{k;\tilde{S}_{k}\rightarrow S_k}$ number of time slots, we are still in Phase~$T$ but we will move to the next eligible $\tilde{S}_{k}$ according to the total ordering $\prec$. Continue this process until all $\tilde{S}_{k}$ have been used.

Since we choose the target packet $X_{k,j_k}$ independently for all $k$, Phase~$T$ thus takes
\begin{align}
x_T\stackrel{\Delta}{=}\max_{\forall k\in T}\left(\sum_{\forall S:(T\backslash k)\subseteq S\subseteq([K]\backslash k)}w_{k;S\rightarrow (T\backslash k)}\right)\label{eq:phase-xT}
\end{align}
number of time slots to finish. Since we have totally $(2^{K}-1)$  different phases, it thus takes $
\sum_{\forall T\in 2^{[K]}: T\neq \emptyset}x_T$
to finish all the phases.

For the following, we will show that there exists a feasible sequential PE scheme if the choices of $\{x_T:\forall T\in 2^{[K]}, T\neq \emptyset\}$ and $\{w_{k;S\rightarrow T}:\forall k\in[K], \forall T\subseteq S\subseteq ([K]\backslash k)\}$ satisfy \eqref{eq:phase-xT} and the following equations:
\begin{align}
&\sum_{\forall T\in 2^{[K]}: T\neq \emptyset}x_T\leq n(1-\epsilon) \text{ for some $\epsilon>0$}\label{eq:pf-total}\\
&\forall k\in[K],\quad w_{k;\emptyset\rightarrow \emptyset}\cdot p_{\cup [K]}=n R_k\label{eq:pf-ind-length-0}\\
&\forall k\in[K], \forall S\subseteq ([K]\backslash k), S\neq \emptyset, \nonumber\\
&\hspace{0cm}
\left(\sum_{\forall T_1: T_1\subseteq S} w_{k;S\rightarrow T_1}\right) p_{\cup ([K]\backslash S)}= \nonumber\\
&\hspace{.7cm}\sum_{\scriptsize \begin{array}{c}\forall S_1,T_1:\text{such that}\\
T_1\subseteq S_1\subseteq ([K]\backslash k),\\
T_1\subseteq S,S\nsubseteq S_1 \end{array}}w_{k;S_1\rightarrow T_1}\cdot f_p\left((S\backslash T_1)\overline{([K]\backslash S)}\right)\label{eq:pf-ind-length-1}\\
&\forall k\in[K], S,T\in 2^{[K]} \text{ satisfying } T\subseteq S\subseteq ([K]\backslash k), T\neq S,\nonumber\\
&\hspace{0cm}\left(w_{k;S\rightarrow T}+\sum_{\scriptsize\begin{array}{c}\forall T_1\subseteq S:\\
(T_1\cup\{k\})\prec (T\cup\{k\})\end{array}}w_{k; S\rightarrow T_1}\right)p_{\cup ([K]\backslash S)} \leq\nonumber\\
&\hspace{.5cm}
\sum_{\scriptsize\begin{array}{c}\forall S_1: S_1\prec S,\\
T\subseteq S_1\subseteq ([K]\backslash k)\end{array}}w_{k;S_1\rightarrow T} \cdot f_p\left((S\backslash T)\overline{([K]\backslash S)}\right)+\nonumber\\
&\hspace{.7cm}\sum_{\scriptsize \begin{array}{c}\forall S_1,T_1:\text{such that}\\
T_1\subseteq S_1\subseteq ([K]\backslash k),\\
(T_1\cup\{k\})\prec (T\cup\{k\}),\\
T_1\subseteq S,S\nsubseteq S_1 \end{array}}w_{k;S_1\rightarrow T_1}\cdot f_p\left((S\backslash T_1)\overline{([K]\backslash S)}\right).\label{eq:pf-ind-length-2}
\end{align}
Note that \eqref{eq:phase-xT} to \eqref{eq:pf-ind-length-2} are similar to \eqref{eq:total-x} to \eqref{eq:ind-length-2} of the achievability inner bound in Proposition~\ref{prop:ach2}. The only differences are (i) The new scaling factor $n$ in \eqref{eq:pf-total} and \eqref{eq:pf-ind-length-0} when compared to \eqref{eq:total-x} and \eqref{eq:ind-length-0}; (ii) The use of the max operation in \eqref{eq:phase-xT} when compared to \eqref{eq:coding-len}; and (iii) The equality ``$=$" in \eqref{eq:pf-ind-length-0} and \eqref{eq:pf-ind-length-1} instead of the inequality ``$\geq$" in \eqref{eq:ind-length-0} and \eqref{eq:ind-length-1}. The first two differences (i) and (ii) are simple restatements and do not change the feasibility region. The third difference (iii) can be reconciled by sending auxiliary dummy (all-zero) packets in the PE scheme as will be clear in the following proof. As a result, we focus on proving the existence of a feasible sequential PE scheme provided the new inequalities \eqref{eq:phase-xT} to \eqref{eq:pf-ind-length-2} are satisfied.

Assuming sufficiently large $n$, the law of large numbers ensures that all the following discussion are accurate within the precision $o(n)$, which is thus ignored for simplicity.
\eqref{eq:pf-total} implies that we can finish all the phases within $n$ time slots. Since each $Q_{k;\emptyset}$ packet $X_{k,j_k}$ in average needs $\frac{1}{p_{\cup[K]}}$ time slots before its $S(X_{k,j_k})$ evolves to another value, \eqref{eq:pf-ind-length-0} ensures that after Phase~$\{k\}$, all $Q_{k;\emptyset}$ packets have been used up and evolved to a different $Q_{k;S}$ packet.

Suppose that we are currently in Phase~$(T\cup\{k\})$ for some $k\notin T$, and suppose that we just finished choosing the $Q_{k;S'}$ packet for some old $S'$ and are in the beginning of choosing a new $Q_{k;S}$ packet (with a new $S\neq S'$) that will subsequently be mixed with packets from other sessions. By Line~\ref{line:S-update}, each $Q_{k;S}$ packet evolves to a different packet if and only if one of the $d_i$ with $i\in ([K]\backslash S)$ receives the coded transmission. Therefore, sending $Q_{k;S}$ packets for  $w_{k;S\rightarrow T}$ number of time slots will consume additional $w_{k;S\rightarrow T}\cdot p_{\cup([K]\backslash S)}$ number of $Q_{k;S}$ packets. Similarly, the previous phases~$(T_1\cup\{k\})$ such that $T_1\subseteq S$ and $(T_1\cup\{k\})\prec (T\cup\{k\})$, have consumed totally
\begin{align}
\sum_{\scriptsize \begin{array}{c}\forall T_1\subseteq S:\\
 (T_1\cup\{k\})\prec (T\cup\{k\})\end{array}}\left(w_{k; S\rightarrow T_1}\cdot p_{\cup ([K]\backslash S)}\right)\nonumber
\end{align}
number of $Q_{k;S}$ packets. The left-hand side of \eqref{eq:pf-ind-length-2} thus represents the total number of $Q_{k;S}$ packets that have been consumed after finishing the $w_{k;S\rightarrow T}$ number of time slots of Phase~$(T\cup\{k\})$ sending $Q_{k;S}$ packets.
 As will be shown short after, the right-hand side of \eqref{eq:pf-ind-length-2} represents the total number of $Q_{k;S}$ packets that have been created until the current time slot.  As a result, \eqref{eq:pf-ind-length-2} corresponds to a packet-conservation law that limits the largest number of $Q_{k;S}$ packets that can be used in Phase~$(T\cup\{k\})$.

To show that  the right-hand side of \eqref{eq:pf-ind-length-2} represents the total number of $Q_{k;S}$ packets that have been created, we notice that the $Q_{k;S}$ packets can either be created within the current Phase~$(T\cup\{k\})$ but during the previous attempts of sending $Q_{k;S_1}$ packets in Phase~$(T\cup\{k\})$ with $S_1\prec S$; or be created in the previous phases~$(T_1\cup\{k\})$ with $(T_1\cup\{k\})\prec (T\cup\{k\})$. The former case corresponds to the first term on the right-hand side  of \eqref{eq:pf-ind-length-2} and the latter case corresponds to the second term on the right-hand side of \eqref{eq:pf-ind-length-2}.

For the former case, for each time slot in which we transmit a $Q_{k;S_1}$ packet in Phase~$(T\cup\{k\})$, there is some chance that the packet will evolve into a $Q_{k;S}$ packet. More explicitly, by Line~\ref{line:S-update} of the {\sc Update}, a $Q_{k;S_1}$ packet in Phase~$(T\cup\{k\})$ evolves into a $Q_{k;S}$ packet if and only if the packet is received by all $d_i$ with $i\in (S\backslash T)$ and not by any $d_i$ with $i\in ([K]\backslash S)$. As a result, each such time slot will create $f_p((S\backslash T)\overline{([K]\backslash S)})$ number of $Q_{k;S}$ packet in average. Since we previously sent $Q_{k;S_1}$ packets for a total $w_{k;S_1\rightarrow T}$ number of time slots, the first term of the right-hand side of \eqref{eq:pf-ind-length-2} is indeed the number of $Q_{k;S}$ packets created within the current Phase~$(T\cup\{k\})$ but during the previous attempts of sending $Q_{k;S_1}$ packets.

For the latter case, for each time slot in which we transmit a $Q_{k;S_1}$ packet in Phase~$(T_1\cup\{k\})$, there is some chance that the packet will evolve into a $Q_{k;S}$ packet, provided we have $T_1\subseteq S$ and $S\nsubseteq S_1$. More explicitly, by Line~\ref{line:S-update} of the {\sc Update}, a $Q_{k;S_1}$ packet in Phase~$(T_1\cup\{k\})$ evolves into a $Q_{k;S}$ packet if and only if
\begin{align}\begin{cases}
T_1\subseteq S\\
S_\text{rx}\nsubseteq S_1\\
(S\backslash T_1)\subseteq S_{\text{rx}}\\
([K]\backslash S)\subseteq ([K]\backslash S_\text{rx})
 \end{cases}\hspace{-.7cm}\text{or equivalently }\begin{cases}
T_1\subseteq S\\
S\nsubseteq S_1 \\
(S\backslash T_1)\subseteq S_{\text{rx}}\\
([K]\backslash S)\subseteq ([K]\backslash S_\text{rx})
\end{cases}\hspace{-.5cm}.
\nonumber \end{align}
Therefore, for any $(S_1,T_1)$ pair satisfying $T_1\subseteq S$ and $S\nsubseteq S_1$, a $Q_{k;S_1}$ packet in Phase~$(T_1\cup\{k\})$ will have $f_p((S\backslash T_1)\overline{([K]\backslash S)})$ probability to evolve into a $Q_{k;S}$ packet. Since we previously sent $Q_{k;S_1}$ packets in Phase~$(T_1\cup\{k\})$ for a total $w_{k;S_1\rightarrow T_1}$ number of time slots, the second term of the right-hand side of \eqref{eq:pf-ind-length-2} is indeed the number of $Q_{k;S}$ packets created during the attempts of sending $Q_{k;S_1}$ packets in the previous Phase~$(T_1\cup\{k\})$.

Suppose that we are currently in Phase~$(S\cup\{k\})$ for some $k\notin S$.
To justify \eqref{eq:pf-ind-length-1}, we first note that in the sequential PE construction we only select the packets $X_{k,j}$ with $k\notin S(X_{k,j})$. By Line~\ref{line:S-update} of the {\sc Update}, each packet $X_{k,j}$ transmitted in Phase~$T$ is either received by the intended destination $d_k$, or it will evolve into a new $S(X_{k,j})$ that is a proper superset of $(T\backslash k)$. As a result, the cardinality-compatible total ordering ``$\prec$" ensures that once we are in Phase~$(S\cup\{k\})$, any subsequent Phase~$T$ with $(S\cup\{k\})\prec T$ will not create any new $Q_{k;S}$ packets. Therefore, if we can clean up all $Q_{k;S}$ packets in Phase~$(S\cup\{k\})$ for all $S\subseteq ([K]\backslash k)$, then in the end of the sequential PE scheme, there will be no $Q_{k;S}$ packets for any $S\subseteq ([K]\backslash k)$. This thus implies that all $X_{k,j}$ packets in the end must have $S(X_{k,j})\ni k$. By Lemma~\ref{lem:decodability}, decodability is thus guaranteed.
\eqref{eq:pf-ind-length-1} is the equation that guarantees that we can clean up all $Q_{k;S}$ packets in Phase~$(S\cup\{k\})$.

By similar computation as in the discussion of the right-hand side of \eqref{eq:pf-ind-length-2}, the right-hand side of \eqref{eq:pf-ind-length-1} is the total number of $Q_{k;S}$ packets generated during the attempts of sending $Q_{k;S_1}$ packets in the previous Phase~$(T_1\cup\{k\})$ with $(T_1\cup\{k\})\prec (S\cup\{k\})$. Similar to the computation in the discussion of the left-hand side of \eqref{eq:pf-ind-length-2}, there is
\begin{align}
\left(\sum_{\forall T_1: T_1\subseteq S, T_1\neq S} w_{k;S\rightarrow T_1}\right) p_{\cup ([K]\backslash S)}\label{eq:clean-up1}
\end{align}
number of $Q_{k;S}$ packets that have been used during the previous Phases $(T_1\cup\{k\})$. In the beginning of this phase, we send $Q_{k;S\rightarrow S}$ packets for $w_{k;S\rightarrow S}$ number of time slots, which can clean up additional \begin{align}w_{k;S\rightarrow S}\cdot p_{\cup ([K]\backslash S)}\label{eq:clean-up2}
\end{align}
number of $Q_{k;S}$ packets. Jointly, \eqref{eq:clean-up1}, \eqref{eq:clean-up2}, and \eqref{eq:pf-ind-length-1} ensures that we can use up all $Q_{k;S}$ packets in Phase~$(S\cup\{k\})$.

The above reasonings show that we can finish the transmission in $n$ time slots, make all $X_{k,j}$ have $S(X_{k,j})\ni k$, and obey the causality constraints. Therefore, the corresponding sequential PE scheme is indeed a feasible solution. The proof of Proposition~\ref{prop:ach2} is thus complete.

\subsection{Attaining The Capacity Of Two Classes of PECs \label{subsec:2spec-class} }

In this section, we prove the capacity results for symmetric 1-to-$K$ broadcast PECs in Proposition~\ref{prop:cap-sym} and for spatially independent broadcast PECs with one-sided fairness constraints in Proposition~\ref{prop:cap-osf}.

\begin{thmproof}{Proof of Proposition~\ref{prop:cap-sym}:}  Since the broadcast channel is symmetric, for any $S_1,S_2\in 2^{[K]}$, we have
\begin{align}
p_{\cup S_1}=p_{\cup S_2} \text{ if } |S_1|=|S_2|.\nonumber
\end{align}
Without loss of generality, also assume that $R_1\geq R_2\geq \cdots \geq R_K$. By the above simplification, the outer bound in Proposition~\ref{prop:outer} collapses to the following single linear inequality:
\begin{align}
\sum_{k=1}^K\frac{R_k}{p_{\cup[k]}}\leq 1.\label{eq:sym-outer}
\end{align}

We use the results in Proposition~\ref{prop:ach2} to prove that \eqref{eq:sym-outer} is indeed the capacity region. To that end, we first fix an arbitrary cardinality-compatible total ordering. Then for any $S\subseteq ([K]\backslash k)$, we choose
\begin{align}
&w_{k;S\rightarrow S}=R_k\cdot \nonumber\\
&\hspace{1.5cm}
\sum_{i=K-|S|}^{K}\left(\sum_{\scriptsize\begin{array}{c}\forall S_1: |S_1|=i\\ ([K]\backslash S)\subseteq S_1\subseteq [K]\end{array}}\frac{\left(-1\right)^{i-(K-|S|)}}{p_{\cup S_1}}\right),\nonumber
\end{align}
and $w_{k;S\rightarrow T}=0$ for all $T$ being a proper subset of $S$.  The symmetry of the broadcast PEC, the assumption that $R_1\geq R_2\geq \cdots \geq R_K$, and \eqref{eq:phase-xT} jointly imply that
\begin{align}
x_T&=w_{k^*;(T\backslash k^*)\rightarrow (T\backslash k^*)}\text{ where }k^*\stackrel{\Delta}{=}\min\{i:i\in T\}\label{eq:sym-ch-prf}
\end{align}
for all $T\neq \emptyset$. For completeness, we set $x_{\emptyset}=0$.

By simple probability arguments as first described\footnote{Some detailed discussion can also be found in the proof of Lemma~\ref{lem:osf-main} in Appendix~\ref{app:osf-lemma}.} in Section~\ref{subsec:high-level-cap3}, we can show that the above choices of $w_{k;S\rightarrow T}$ and $x_T$ are all non-negative and jointly satisfy the inequalities \eqref{eq:coding-len} to \eqref{eq:ind-length-2}.

The remaining task is to show that inequality
\eqref{eq:total-x} is satisfied for any $(R_1,\cdots, R_K)$ in the interior of the capacity outer bound \eqref{eq:sym-outer}. To that end, we simply need to verify the following equalities by some simple arithmetic computation.
\begin{align}
\forall k\in [K],~&\sum_{\forall T\in 2^{[K]}: k\in T, [k-1]\cap T=\emptyset} x_T\nonumber\\
&= \sum_{\forall T\in 2^{[K]}: k\in T, [k-1]\cap T=\emptyset}w_{k;(T\backslash k)\rightarrow (T\backslash k)}\nonumber\\
&= \frac{R_k}{p_{\cup{[k]}}}.\label{eq:prf-sum}
\end{align}
Summing \eqref{eq:prf-sum} over different $k$ values, we thus show that
any $(R_1,\cdots, R_K)$ in the interior of the capacity outer bound \eqref{eq:sym-outer} indeed satisfies
\eqref{eq:total-x}. The proof of Proposition~\ref{prop:cap-sym} is complete.
\end{thmproof}

\begin{thmproof}{Proof of Proposition~\ref{prop:cap-osf}:}
Consider an arbitrary spatially independent broadcast PEC with $0<p_1\leq p_2\leq \cdots\leq p_K$. The capacity outer bound in Proposition~\ref{prop:outer} implies that any achievable rate vector $(R_1,\cdots, R_K)$ must satisfy
\begin{align}
\sum_{k=1}^K\frac{R_k}{1-\prod_{l=1}^k(1-p_l)}\leq 1.\label{eq:osf-cap-2}
\end{align}

We use the results in Proposition~\ref{prop:ach2} to prove that any one-sidedly fair rate vector $(R_1,\cdots, R_K)\in \Lambda_\text{osf}$ that is in the interior of
\eqref{eq:osf-cap-2} is indeed achievable. To that end, we first fix an arbitrary cardinality-compatible total ordering. Then for any $S\subseteq ([K]\backslash k)$, we choose
\begin{align}
&w_{k;S\rightarrow S}=R_k\cdot \nonumber\\
&\hspace{1.5cm}
\sum_{i=K-|S|}^{K}\left(\sum_{\scriptsize\begin{array}{c}\forall S_1: |S_1|=i\\ ([K]\backslash S)\subseteq S_1\subseteq [K]\end{array}}\frac{\left(-1\right)^{i-(K-|S|)}}{p_{\cup S_1}}\right),\nonumber
\end{align}
and $w_{k;S\rightarrow T}=0$ for all $T$ being a proper subset of $S$.
By Lemma~\ref{lem:osf-main} in Appendix~\ref{app:osf-lemma} and by \eqref{eq:phase-xT}, we have
\begin{align}
x_T&=\max_{\forall k\in T}\left(w_{k;(T\backslash k)\rightarrow (T\backslash k)}\right)\nonumber\\
&=w_{k^*;(T\backslash k^*)\rightarrow (T\backslash k^*)}\text{ where }k^*\stackrel{\Delta}{=}\min\{i:i\in T\}\nonumber
\end{align}
for all $T\neq \emptyset$. For completeness, we set $x_{\emptyset}=0$.

The remaining proof of Proposition~\ref{prop:cap-osf} can be completed by following the same steps after \eqref{eq:sym-ch-prf} of the proof of Proposition~\ref{prop:cap-sym}.

\end{thmproof}

\section{Further Discussion of The Main Results\label{sec:discussion}}

We provide some further discussion of the main results in this section. In particular, we focus on the accounting overhead of the PE schemes, the minimum finite field size of the PE schemes, the sum rate performance of asymptotically large $M$ values, and numerical evaluations of the outer and inner bounds for general 1-to-$K$ broadcast PECs. 

\subsection{Accounting Overhead\label{subsec:accounting}}
Thus far we assume that the individual destination $d_k$ knows the global coding vector $\vv_\text{tx}$ that is used to generate the coded symbols (see Line~\ref{line:transmit-v} of the main PE scheme). Since the coding vector $\vv_\text{tx}$ is generated randomly, this assumption generally does not hold, and the coding vector $\vv_\text{tx}$ also needs to be conveyed to the destinations. Otherwise, destinations $d_k$ cannot decode the original information symbols $X_{k,j}$ for the received coded symbols $Z_k(t)$, $t\in[n]$. The cost of sending the coding vector $\vv_\text{tx}$ is termed the {\em coding overhead} or the {\em accounting overhead.}

We use the generation-based scheme in \cite{ChouWuJain03} to average out and absorb the accounting overheard. Namely, we first choose sufficiently large $n$ and finite field size $q$ such that the PE scheme can achieve $(1-\epsilon)$-portion of the capacity with arbitrarily close-to-one probability when assuming there is no accounting overhead. Once the $n$ and $q$ values are fixed, we choose an even larger finite field $\GF(q^{M+\sum_{k=1}^KnR_k})$ for some large integer $M$. The large finite field is then treated as a vector of dimension $M+\sum_{k=1}^KnR_k$.
Although each information symbol (vector) is chosen from $X_{k,j}\in \GF(q^{M+\sum_{k=1}^KnR_k})$, we limit the range of the $X_{k,j}$ vector value such that the first $\sum_{k=1}^KnR_k$ coordinates are always zero, i.e., no information is carried in the first $\sum_{k=1}^KnR_k$ coordinates. We can thus view the entire systems as sending $M$ coordinates in each vector. During the transmission of the PE scheme, we focus on coding over
each coordinate, respectively, rather than jointly coding over the entire vector.
The
same coding vector $\vv_{\text{tx}}$ is used repeatedly to encode the last $M$ coordinates. And we use the first $\sum_{k=1}^KnR_k$ coordinates to store the coding vector $\vv_{\text{tx}}$.

Since only the last $M$ coordinates are used to carry information, overall the transmission rate is reduced by a factor $\frac{M}{M+\sum_{k=1}^KnR_k}$. By choosing a sufficiently large $M$, we have averaged out and absorbed the accounting overhead.

\subsection{Minimum Finite Field Size}
The PE scheme in Section~\ref{sec:achievability} is presented in the context of random linear network coding, which uses a sufficiently large finite field size $\GF(q)$ and proves that the desired properties hold with close-to-one probability. The main advantage of this random-coding-based description is that the entire algorithm can be carried out in a very efficient and distributed fashion. For example, with a sufficiently large $q$, the source $s$ only needs to bookkeep the $S(X_{k,j})$ and $\vv(X_{k,j})$ values of all the information packets $X_{k,j}$. All the coding and update computations are of linear complexity. On the other hand, the drawback of a randomized algorithm is that even with very large $\GF(q)$, there is still a small probability that after the termination of the PE algorithm, some destination $d_k$ has not accumulated enough linearly independent packets to decode the desired symbols $X_{k,1}$ to $X_{k,nR_k}$. For the following, we discuss how to covert the randomized PE scheme into a deterministic algorithm by quantifying the corresponding minimum size of the finite field.

\begin{proposition}\label{prop:finite-field} Consider the 1-to-$K$ broadcast PEC problem with COF. For any fixed finite field $\GF(q_0)$ satisfying $q_0>K$, all the achievability results in Propositions~\ref{prop:cap3}, \ref{prop:ach2}, \ref{prop:cap-sym}, and~\ref{prop:cap-osf} can be attained by a deterministic PE algorithm on $\GF(q_0)$ that deterministically computes the mixing coefficients
$\{c_k:\forall k\in T\}$ in Line~\ref{line:vv-construct} of the PE scheme.

\end{proposition}

The proof of Proposition~\ref{prop:finite-field} is relegated to Appendix~\ref{app:finite-field}.


{\em Remark 1:} In practice, the most commonly used finite field is $\GF(2^8)$. Proposition~\ref{prop:finite-field} guarantees that $\GF(2^8)$ is sufficient for coding over $K\leq 255$ sessions together.

{\em Remark 2:} On the other hand, the construction of good mixing coefficients $\{c_k:\forall k\in T\}$ in Proposition~\ref{prop:finite-field} is computationally intensive. The randomized PE scheme has substantial complexity advantage over the deterministic PE scheme.


\subsection{The Asymptotic Sum-Rate Capacity of Large $M$ Values}
We first define the sum-rate capacity as follows:
\begin{definition}
The sum-rate capacity $R_\text{sum}^*$ is defined as
\begin{align}
R_\text{sum}^*=\sup\left\{\sum_{k=1}^KR_k:(R_1,\cdots, R_K)\text{ is achievable}\right\}.\nonumber
\end{align}
\end{definition}

Proposition~\ref{prop:cap-osf} quickly implies the following corollary.

\begin{corollary}\label{cor:sum-rate}Consider any spatially independent 1-to-$K$ broadcast PECs with marginal success probabilities $0< p_1\leq p_2\leq \cdots\leq p_K<1$. With COF, the sum-rate capacity satisfies
\begin{align}
\frac{\sum_{k=1}^K\frac{1}{1-p_k}}{\sum_{k=1}^K\frac{1}{(1-p_k)(1-\prod_{l=1}^k(1-p_l))}}\leq R_\text{sum}^*\leq 1.\nonumber
\end{align}
If we further enforce perfect fairness, i.e., $R_1=R_2=\cdots=R_K$, then the corresponding sum-rate capacity $R_\text{sum,perf.fair}^*$ becomes
\begin{align}
R_\text{sum,perf.fair}^*=\frac{K}{\sum_{k=1}^K\frac{1}{(1-\prod_{l=1}^k(1-p_l))}}.\nonumber
\end{align}
\end{corollary}
\begin{proof}
Since the sum-rate capacity $nR_\text{sum}^*$ is no larger than the total available time slots $n$, we have the upper bound $R_\text{sum}^*\leq 1$. Since the  rate vector $\left(\frac{R}{1-p_1},\frac{R}{1-p_2},\cdots, \frac{R}{1-p_K}\right)$ is one-sidedly fair, Proposition~\ref{prop:cap-osf} leads to the lower bound of $R_\text{sum}^*$.
Since a perfectly fair rate vector $(R,R,\cdots, R)$ is also one-sidedly fair, Proposition~\ref{prop:cap-osf} gives the exact value of $R_\text{sum,perf.fair}^*$.
\end{proof}

Corollary~\ref{cor:sum-rate} implies the following. Consider any fixed $p>0$. Consider a symmetric, spatially independent 1-to-$K$ broadcast PEC with marginal success probability $p_1=p_2=\cdots=p_K=p$. When $K$ is sufficiently large, both the sum-rate capacities $R_\text{sum}^*$ and $R_\text{sum,perf.fair}^*$ approach one. That is, for sufficiently large $K$, network coding completely removes all the channel uncertainty by taking advantage of the spatial diversity among different destinations $d_i$. Therefore, each $(s,d_k)$ session can sustain rate  $\frac{1-\epsilon}{K}$ for some $\epsilon>0$ where $\epsilon\rightarrow 0$ when $K\rightarrow\infty$.
 Note that when compared to the MIMO capacity gain, the setting in this paper is more conservative in a sense that it assumes that the channel gains change independently from time slot to time slot (instead of block fading)  while no coordination is allowed among destinations.

This relationship was first observed and proven in \cite{LarssonJohansson06} by identifying a lower bound of $R_\text{sum,perf.fair}^*$
for symmetric, spatially independent PECs. Compared to the results in \cite{LarssonJohansson06},  Corollary~\ref{cor:sum-rate} characterizes the exact value of $R_\text{sum,perf.fair}^*$ and provides a tighter lower bound on  $R_\text{sum}^*$
for non-symmetric spatially independent PECs. The $R_\text{sum,perf.fair}^*$ will later be evaluated numerically in Section~\ref{subsec:sim} for non-symmetric spatially independent PECs.  

\subsection{Numerical Evaluation\label{subsec:sim}}

\begin{figure}
\centering
\includegraphics[width=7cm]{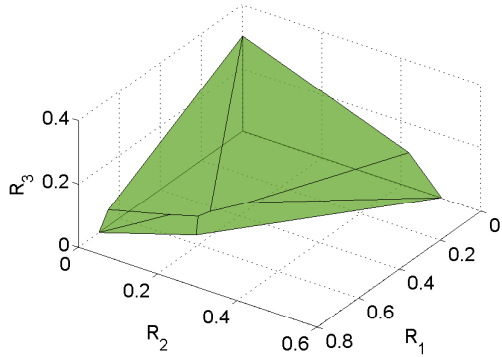}
\caption{The 3-D capacity region of a 1-to-3 spatially independent broadcast PEC with marginal success probabilities $p_1=0.7$, $p_2=0.5$, and $p_3=0.3$.\label{fig:3-by-3}}
\end{figure}

Fig.~\ref{fig:3-by-3} illustrates the 3-dimensional capacity region of $(R_1,R_2,R_3)$ of a spatially independent, 1-to-3 broadcast PEC with COF. The corresponding marginal probabilities are $p_1=0.7$, $p_2=0.5$, and $p_3=0.3$. The six facets in Fig.~\ref{fig:3-by-3} correspond to the six different permutations used in Proposition~\ref{prop:outer}.

For general 1-to-$K$ PECs with $K\geq 4$, we can use the outer and inner bounds in Propositions~\ref{prop:outer} and~\ref{prop:ach2} to bracket the actual capacity region. Since there is no tightness guarantee for $K\geq 4$ except for the two special classes of channels in Section~\ref{subsec:special}, we use computer to  numerically evaluate the tightness of the outer and inner bound pairs. To that end, for any fixed $K$ value,
we consider spatially independent 1-to-$K$ broadcast PEC with the marginal success probabilities $p_k$ chosen randomly from $(0,1)$. To capture the $K$-dimensional capacity {region}, we first choose a search direction $\vec{v}=(v_1,\cdots,v_K)$ uniformly randomly from a $K$-dimensional unit ball. With the chosen values of $p_{k}$ and $\vec{v}$, we use a linear programming (LP) solver to find the largest $t_{\text{outer}}$ such that $(R_1,\cdots, R_K)=(v_1\cdot t_{\text{outer}},\cdots,v_K\cdot t_{\text{outer}})$ satisfies the capacity outer bound in
Proposition~\ref{prop:outer}.

To evaluate the capacity inner bound, we need to choose a cardinality-compatible total ordering. For any set $S\subseteq [K]$, the corresponding incidence vector ${\bf 1}_S$ is a $K$-dimensional binary vector with the $i$-th coordinate being one if and only if $i\in S$. We can also view ${\bf 1}_S$ as a binary number, where the first coordinate is the most significant bit and the $K$-th coordinate is the least significant bit. For example, for $K=4$, $S=\{1,2,4\}$ has ${\bf 1}_S=(1,1,0,1)=13$. For two sets $S_1\neq S_2$, we say $S_1\prec S_2$ if and only if either (i) $|S_1|=|S_2|$ and ${\bf 1}_{S_1}<{\bf 1}_{S_2}$, or (ii) $|S_1|<|S_2|$. Based on this cardinality-compatible total ordering, we again use the LP solver to find the
largest $t_{\text{inner}}$ such that
 $(R_1,\cdots, R_K)=(v_1\cdot t_{\text{inner}},\cdots,v_K\cdot t_{\text{inner}})$ satisfies the capacity inner bound in
Proposition~\ref{prop:ach2}. The deficiency is then defined as ${\mathsf{defi}}\stackrel{\Delta}{=}\frac{t_{\text{outer}}-t_{\text{inner}}}{t_{\text{outer}}}$. We then repeat the above experiment for $10^4$ times for $K=4$, 5, and 6, respectively.

Note that although there is no tightness guarantee  for $K\geq 4$ except in the one-sidedly fair rate region, all our numerical experiments (totally $3\times 10^4$) have ${\mathsf{defi}}\leq 0.1\%$. Actually, in our experiments with $K\leq 6$, we have not found any instance of the input parameters $(p_1,\cdots, p_K)$ and $\vec{v}$, for which ${\mathsf{defi}}$ is greater than the numerical precision of the LP solver. This shows that Propositions~\ref{prop:outer} and~\ref{prop:ach2} indeed describe the capacity region from the practical perspective.

\begin{figure}
\centering
\includegraphics[width=7cm]{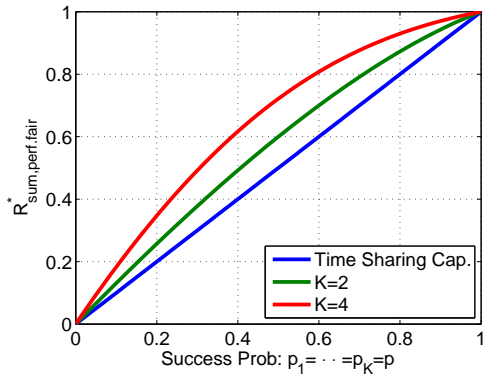}
\caption{The sum-rate capacity $R^*_{\text{sum,perf.fair}}$ in a perfectly fair system versus the marginal success probability $p$ of a symmetric, spatially independent 1-to-$K$ broadcast PEC, $K=2$ and $4$.  \label{fig:sym124}}
\end{figure}

\begin{figure}
\centering
\includegraphics[width=7cm]{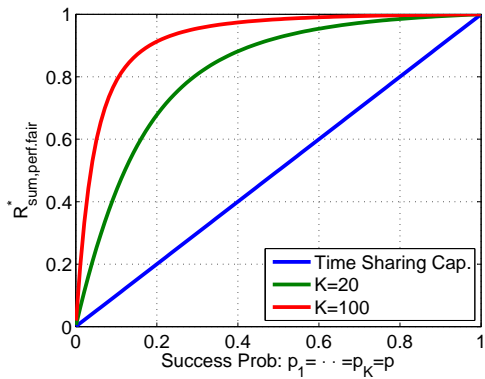}
\caption{The sum-rate capacity $R^*_{\text{sum,perf.fair}}$ in a perfectly fair system versus the marginal success probability $p$ of a symmetric, spatially independent 1-to-$K$ broadcast PEC, $K=20$ and $100$.\label{fig:sym120100}}
\end{figure}

To illustrate the broadcast network coding gain, we compare the sum-rate capacity versus the sum rate achievable by time sharing. Figs.~\ref{fig:sym124} and~\ref{fig:sym120100} consider symmetric, spatially independent PECs with marginal success probabilities $p_1=\cdots=p_K=p$. We plot the sum rate capacity $R_{\text{sum,perf.fair}}^*$ versus $p$ for a perfectly fair system. The baseline is the largest sum rate that can be achieved by time sharing for a perfectly fair system. As seen in Figs.~\ref{fig:sym124} and~\ref{fig:sym120100}, the network coding gains are substantial even when we only have $K=4$ destinations.   We also note that $R_{\text{sum,perf.fair}}^*$ approaches one for all $p\in(0,1]$ as predicted by Corollary~\ref{cor:sum-rate}.

 We are also interested in the sum rate capacity under asymmetric channel profiles (also known as heterogeneous channel profiles). Consider asymmetric, spatially independent PECs. For each $p$ value, we let the channel gains $p_1$ to $p_K$ be equally spaced between $(p,1)$, i.e., $p_k=p+(k-1)\frac{1-p}{K-1}$. We then plot the sum rate capacities for different $p$ values. Fig.~\ref{fig:asym6} describes the case for $K=6$. The sum rate capacities are depicted by solid curves, which is obtained by solving the linear inequalities in the outer and inner bounds of Propositions~\ref{prop:outer} and~\ref{prop:ach2}. For all the parameter values used to plot Fig.~\ref{fig:asym6}, the outer and inner bounds meet and we thus have the exact sum rate capacities for the case of $K=6$. The best achievable rate of time sharing are depicted by dashed curves in Fig.~\ref{fig:asym6}. We consider both a perfectly fair system $(R,R,\cdots, R)$ or a proportionally fair system $(p_1R, p_2R,\cdots, p_KR)$ for which the rate of the $(s,d_k)$ session is proportional to the marginal success probability $p_k$ (the optimal rate when all other sessions are silent). To highlight the impact of channel heterogeneity, we also redraw the curves of perfectly symmetric PECs with $p_1=\cdots=p_K=p$.

 \begin{figure}
 \centering
\includegraphics[width=7cm]{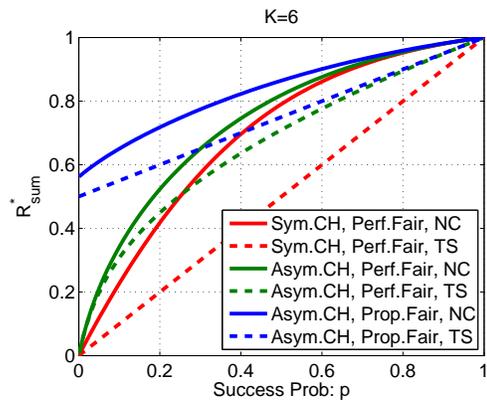}
\caption{The sum-rate capacities for a 6-destination heterogenous channel profiles with the success probabilities $p_1$ to $p_6$ evenly spaced between $(p,1)$. \label{fig:asym6}}
\end{figure}

As seen in Fig.~\ref{fig:asym6}, for perfectly fair systems, the sum-rate capacity gain does not increase much when moving from symmetric PECs $p_1=\cdots=p_K=p$ to
the heterogeneous channel profile with $p_1$ to $p_K$ evenly spaced between $(p,1)$.
The reason is due to that the worst user $d_1$ (with the smallest $p_1$) dominates the system performance in a perfectly fair system. When we allow proportional fairness, network coding again provides substantial improvement for all $p$ values. However, the gain is not as large as the case of symmetric channels. For example, when $p_1$ to $p_K$ are evenly spaced between $(0,1)$. The sum rate capacity of a proportionally fair system is 0.56 ($p=0$). However, if all $p_1$ to $p_K$ are concentrated on their mean $0.5$, then the sum rate capacity of the symmetric channel  ($p=0.5$) is 0.79. The results show that for practical implementation, it is better to group together all the sessions of similar marginal success rates and perform intersession network coding within the same group.

\begin{figure}
 \centering
\includegraphics[width=7cm]{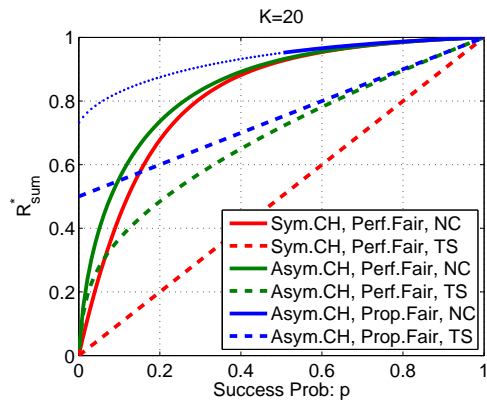}
\caption{The sum-rate capacities for a 20-destination heterogenous channel profiles with the success probabilities $p_1$ to $p_{20}$ evenly spaced between $(p,1)$. \label{fig:asym20}}
\end{figure}

We also repeat the same experiment of Fig.~\ref{fig:asym6} but for the case $K=20$ in Fig.~\ref{fig:asym20}. In this case of a moderate-sized $K=20$, the sum-rate capacity of a perfectly fair system is characterized by Proposition~\ref{prop:cap-osf}. On the other hand, the sum-rate capacity of a proportionally fair system are characterized by Proposition~\ref{prop:cap-osf} only when all $p_1$ to $p_K$ are in the range of $[0.5, 1]$ (see the discussion of one-sidedly fair systems in Section~\ref{subsec:2spec-class}). Since the evaluations of both the outer and inner bounds have prohibitively high complexity for the case $K=20$, we use the capacity formula of Proposition~\ref{prop:cap-osf} as a substitute\footnote{When all $p_1$ to $p_K$ are in $[0.5,1]$, the formula in Proposition~\ref{prop:cap-osf} describes the capacity. When some $p_1$ to $p_K$ is outside $[0.5,1]$, the formula in Proposition~\ref{prop:cap-osf} describes an outer bound of the capacity.} of the sum-rate capacity for  $p<0.5$, which is illustrated in Fig.~\ref{fig:asym20} by the fine dotted extension of  the solid curve for the region of $p\in[0.5,1]$. Again, the more sessions ($K=20$) to be encoded together, the higher the network coding gain over the best time sharing rate.

\section{Conclusion\label{sec:conclusion}}

The recent development of practical network coding schemes \cite{ChouWuJain03} has brought attentions back to the study of packet erasure channels (PECs), which is a generalization of the classic binary erasure channels. Since per-packet feedback (such as ARQ) is widely used in today's network protocols, it is thus of critical importance to study PECs with channel output feedback (COF). This work have focused on deriving the capacity of general 1-to-$K$ broadcast PECs with COF, which was previously known only for the case $K=2$.

In this work, we have proposed a new class of intersession network coding schemes, termed the packet evolution (PE) schemes, for the broadcast PECs. Based on the PE schemes, we have
derived the capacity region for general 1-to-3 broadcast PECs, and a pair of capacity outer and inner bounds for general 1-to-$K$ broadcast PECs, both of which can be easily evaluated by any linear programming solver for the cases $K\leq 6$. It has also been proven that the outer and inner bounds meet for two classes of 1-to-$K$ broadcast PECs: the symmetric broadcast PECs, and the spatially independent broadcast PECs with the one-sided fairness rate constraints. Extensive numerical experiments have shown that the outer and inner bounds meet for almost all broadcast PECs encountered in practical scenarios. Therefore, we can effectively use the outer/inner bounds as the substitute for the capacity region in practical applications.
 The capacity  results in this paper also show that for large $K$ values, the noise of the broadcast PECs can be effectively removed by exploiting the inherent spatial diversity of the system, even without any coordination between the destinations.

For practical implementation, the COF usually arrives in batches. That is, instead of instant per-packet COF, we usually have periodic, per-batch COF. The PE scheme can be modified to incorporate periodic COF as well. The corresponding discussion and some precursory empirical implementation of the revised PE scheme can be found in \cite{KoutsonikolasWangHu10a}.

\section*{Acknowledgment}

This work was supported in parts by NSF grants CCF-0845968 and CNS-0905331. The author would also like to thank for the insightful suggestions of Profs.\ Anant Sahai and David Tse.

\appendices

\section{A Proof Of Lemma~\ref{lem:non-interfering}\label{app:lem-non-interfering}}

\begin{proof} We prove Lemma~\ref{lem:non-interfering} by induction. First consider the end of the 0-th time slot (before any transmission).   Since $S_0(X_{k,j})=\emptyset$ for all $X_{k,j}$ and the only $d_i$ satisfying $i\in ( S_0(X_{k,j})\cup\{k\})$ is $d_k$, we only need to check whether $\vv_0(X_{k,j})$ is in the linear space $\linsp(\Omega_{Z,k}(0),\Omega_{M,k})$.
Note that in the end of time 0,  $\vv_0(X_{k,j})$ is the elementary vector $\delta_{k,j}\in \Omega_{M,k}$.
Lemma~\ref{lem:non-interfering} thus holds in the end of time 0.

Suppose Lemma~\ref{lem:non-interfering} is satisfied in the end of time $(t-1)$. Consider the end of time $t$.  We use $T$ to denote the subset chosen in the beginning of time $t$ and use $\{X_{k,j_k}:\forall k\in T\}$ to denote the corresponding target packets. Consider the following cases:

{\bf Case 1:} Consider those $X_{k,j_k}$ such that $S_{t}(X_{k,j_k})=S_{t-1}(X_{k,j_k})$. We first note that if Line~\ref{line:S-update} of the {\sc Update} is executed, then $S_{t}(X_{k,j_k})\neq S_{t-1}(X_{k,j_k})$. Therefore, for those $X_{k,j_k}$ such that $S_{t}(X_{k,j_k})=S_{t-1}(X_{k,j_k})$, we must have that Lines~\ref{line:S-update} and~\ref{line:v-update} of the {\sc Update} are not executed, which implies that
$\vv_{t}(X_{k,j_k})=\vv_{t-1}(X_{k,j_k})$.

By definition, $\Omega_{Z,i}(t-1)\subseteq \Omega_{Z,i}(t)$ for all $i\in[K]$ and $t\in[n]$. By the induction assumption, we thus have that  for all $d_i$ with $i\in (S_{t}(X_{k,j_k})\cup\{k\})=(S_{t-1}(X_{k,j_k})\cup\{k\})$,
\begin{align}
&\vv_{t}(X_{k,j_k})=\vv_{t-1}(X_{k,j_k})\nonumber\\
&\in \linsp(\Omega_{Z,i}(t-1),\Omega_{M,i})\subseteq \linsp(\Omega_{Z,i}(t),\Omega_{M,i}).\nonumber
\end{align}
Vector $\vv_{t}(X_{k,j_k})$ is thus non-interfering from the perspectives of all $d_i$, $i\in (S_{t}(X_{k,j_k})\cup\{k\})$.

{\bf Case 2:} Consider those $X_{k',j'}$ that are not a target packet. Since those packets do not participate in time $t$ and their $S(X_{k',j'})$ and $\vv(X_{k',j'})$ do not change from time $(t-1)$ to time $t$. The same arguments of Case~1 hold verbatim for this case.

{\bf Case 3:} Consider those target packets $X_{k,j_k}$ such that $S_{t}(X_{k,j_k})\neq S_{t-1}(X_{k,j_k})$. For those target packets $X_{k,j_k}$ with $S_{t}(X_{k,j_k})\neq S_{t-1}(X_{k,j_k})$, we must have $S_{t}(X_{k,j_k})=(T\cap S_{t-1}(X_{k,j_k}))\cup S_\text{rx}$ and  $\vv_t(X_k,j_k)=\vv_\text{tx}$ by Lines~\ref{line:S-update} and~\ref{line:v-update} of the {\sc Update}, respectively.
 Consider any $d_i$ such that $i\in (S_{t}(X_{k,j_k})\cup\{k\})$. We have two subcases: {\bf Case~3.1:} $i\in S_{\text{rx}}$. 
 Since  all such $d_i$ must explicitly receive the new $\vv_t(X_{k,j_k})=\vv_\text{tx}$ in the end of time $t$, we must have \begin{align}
\vv_t(X_{k,j_k})&\in\linsp(\vv_\text{tx})=\linsp(Z_i(t))\nonumber\\
&\subseteq \Omega_{Z,i}(t)\subseteq \linsp(\Omega_{Z,i}(t),\Omega_{M,i}).\nonumber
\end{align}
Such $\vv_t(X_{k,j_k})$ is thus non-interfering from $d_i$'s perspective.
{\bf Case~3.2:} $i\in (S_{t}(X_{k,j_k})\cup\{k\})\backslash S_\text{rx}$. We first notice that
\begin{align}
S_{t}(X_{k,j_k})\cup\{k\}&=(T\cap S_{t-1}(X_{k,j_k}))\cup S_\text{rx}\cup\{k\}\nonumber\\
&=((T\cup \{k\})\cap (S_{t-1}(X_{k,j_k})\cup\{k\}))\cup S_\text{rx}\nonumber\\
&=(T\cap (S_{t-1}(X_{k,j_k})\cup\{k\}))\cup S_\text{rx}\label{eq:TT1}\\
&=T\cup S_\text{rx},\label{eq:TT2}
\end{align}
where \eqref{eq:TT1} follows from that $k\in T$ since $X_{k,j_k}$ is a target packet. \eqref{eq:TT2} follows from that $(S_{t-1}(X_{k,j_k})\cup \{k\})\supseteq T$ by Line~\ref{line:targeting-set} of the main structure of the PE scheme. From \eqref{eq:TT2}, the $i$ value in this case must satisfy
\begin{align}i\in (S_{t}(X_{k,j_k})\cup\{k\})\backslash S_\text{rx}=(T\cup S_{\text{rx}})\backslash S_\text{rx}=T\backslash S_\text{rx}.\label{eq:irx}
\end{align}
Also by Line~\ref{line:targeting-set} of the main structure of the PE scheme, for all $i$ satisfy \eqref{eq:irx} we must have $i\in (T\backslash S_\text{rx})\subseteq T\subseteq (S_{t-1}(X_{l,j_l})\cup \{l\})$ for all $l\in T$. By induction, the $\vv_{t-1}(X_{l,j_l})$ vectors used to generate the new $\vv_\text{tx}$ (totally $|T|$ of them) must all be non-interfering from $d_i$'s perspective. Therefore
\begin{align}
\forall l\in T,~ \vv_{t-1}(X_{l,j_l})&\in \linsp(\Omega_{Z,i}(t-1),\Omega_{M,i})\\
&=\linsp(\Omega_{Z,i}(t),\Omega_{M,i}),\nonumber
\end{align}
where the last equality follows from that $d_i$, $i\in T\backslash S_\text{rx}$, does not receive any packet in time $t$. Since $\vv_\text{tx}$ is a linear combination of $\vv_{t-1}(X_{l,j_l})$ for all $l\in T$, we thus have
\begin{align}
\vv_t(X_{k,j_k})=\vv_\text{tx}\in \linsp(\Omega_{Z,i}(t),\Omega_{M,i}).\nonumber
\end{align}
Based on the above reasoning, $\vv_t(X_{k,j_k})$ is non-interfering for all $d_i$ with $i\in(S_{t}(X_{k,j_k})\cup\{k\})\backslash S_\text{rx}$.

The proof is completed by induction on the time index $t$.

\end{proof}

\section{A Proof Of Lemma~\ref{lem:decodability}\label{app:lem-decodability}}

\begin{thmproof}{Proof of Lemma~\ref{lem:decodability}:} We prove this lemma by induction on time $t$. In the end of time $t=0$, since
\begin{align}
\Omega_{R,k}(0)&=\linsp(\vv_0(X_{k,j}):\forall j\in [nR_k], k\notin S_0(X_{k,j})=\emptyset)\nonumber\\
&=\linsp(\delta_{k,j}:\forall k\in[K],j\in [nR_k])=\Omega_{M,k},\nonumber
\end{align}
We thus have
\begin{align}\prop\left(\linsp(\Omega_{Z,k}(0),\Omega_{R,k}(0))
=\linsp(\Omega_{Z,k}(0),\Omega_{M,k})\right)=1.\nonumber
\end{align}
Lemma~\ref{lem:decodability} is satisfied.

Consider the end of time $t>0$. By induction, the following event is of close-to-one probability:
\begin{align}
&\linsp(\Omega_{Z,k}(t-1),\Omega_{R,k}(t-1))\nonumber\\
&=\linsp(\Omega_{Z,k}(t-1),\Omega_{M,k}).\label{eq:cond-prev}
\end{align}
The following proofs are conditioned on the event that \eqref{eq:cond-prev} is satisfied.

 We use $T$ to denote the subset chosen in the beginning of time $t$ and use $\{X_{k,j_k}\}$ to denote the corresponding target packets. Consider the following cases:

{\bf Case 1:} Consider those $k\in T$ such that the corresponding target packet $X_{k,j_k}$ either has $S_{t}(X_{k,j_k})=S_{t-1}(X_{k,j_k})$ or has $k\in S_{t-1}(X_{k,j_k})$. For the former subcase $S_{t}(X_{k,j_k})=S_{t-1}(X_{k,j_k})$, by Line~\ref{line:S-update} of the {\sc Update}, we must have $\vv_{t}(X_{k,j_k})=\vv_{t-1}(X_{k,j_k})$. Since $X_{k,j_k}$ is the only packet among $\{X_{k,j}:\forall j\in[nR_k]\}$ that participate in time $t$, for which the corresponding $\vv(X_{k,j})$ coding vector may change, we must have
$\vv_t(X_{k,j})=\vv_{t-1}(X_{k,j})$ for all $j\in[nR_k]$. We then have
 \begin{align}
\Omega_{R,k}(t)&=\linsp(\vv_{t}(X_{k,j}):\forall j\in[nR_k], k\notin S_t(X_{k,j})\nonumber\\
&=\linsp(\vv_{t-1}(X_{k,j}):\forall j\in[nR_k], k\notin S_{t-1}(X_{k,j}))\nonumber\\
&=\Omega_{R,k}(t-1).\label{eq:dec-pf-rm}
\end{align}
We note that for the latter subcase $k\in S_{t-1}(X_{k,j_k})$, we must have $T\subseteq (S_{t-1}(X_{k,j_k})\cup\{k\})=S_{t-1}(X_{k,j_k})$ by Line~\ref{line:targeting-set} of the main PE scheme. Therefore
Line~\ref{line:S-update} of the {\sc Update} implies that $k\in S_{t}(X_{k,j_k})$ as well. Since the remaining space $\Omega_{R,k}$ only counts the vectors $\vv(X_{k,j})$ with $k\notin S(X_{k,j})$, \eqref{eq:dec-pf-rm} holds for the latter subcase as well.
For both subcases, let $\w_k(t)$ denote the corresponding coding vector of $Z_k(t)$, which may or may not be an erasure. We then have
\begin{align}&\linsp(\Omega_{Z,k}(t), \Omega_{R,k}(t))\nonumber\\
&=\linsp(\w_k(t),\Omega_{Z,k}(t-1), \Omega_{R,k}(t))\nonumber\\
&=\linsp(\w_k(t),\Omega_{Z,k}(t-1), \Omega_{R,k}(t-1))\nonumber\\
&=\linsp(\w_k(t),\Omega_{Z,k}(t-1), \Omega_{M,k})\label{eq:lem-dec-ind1}\\
&=\linsp(\Omega_{Z,k}(t), \Omega_{M,k}),\nonumber
\end{align}
where \eqref{eq:lem-dec-ind1} is obtained by the induction condition \eqref{eq:cond-prev}. Lemma~\ref{lem:decodability} thus holds for the $k$ values satisfying Case~1.

{\bf Case 2:} Consider those $d_l$ with $l\notin T$. Since no $X_{l,j}$ packets  participate in time $t$ and their $S(X_{l,j})$ and $\vv(X_{l,j})$ do not change in time $t$. The same arguments of Case~1 thus hold verbatim for this case.

{\bf Case 3:} Consider those $k\in T$ such that the corresponding target packet $X_{k,j_k}$ has $S_{t}(X_{k,j_k})\neq S_{t-1}(X_{k,j_k})$ and $k\notin S_{t-1}(X_{k,j_k})$.  Define $\Omega'_R$ as
\begin{align}
\Omega_R'\stackrel{\Delta}{=}\linsp\{\vv_{t-1}(X_{k,j}):\forall j\in [nR_k]\backslash j_k, k\notin S_{t-1}(X_{k,j})\}.\label{eq:omega-R-prime}
\end{align}
Note that the conditions of Case~3 and \eqref{eq:omega-R-prime} jointly imply that $\Omega_{R,k}(t-1)=\linsp(\vv_{t-1}(X_{k,j_k}),\Omega_R')$. We have two subcases
{\bf Case 3.1:} $k\notin S_{t}(X_{k,j_k})$ and {\bf Case 3.2:} $k\in S_{t}(X_{k,j_k})$.


{\bf Case 3.1:} $k\notin S_{t}(X_{k,j_k})$. By Line~\ref{line:S-update}
of the {\sc Update}, we have $k\notin S_\text{rx}$, i.e., $d_k$ receives an erasure in time $t$. Therefore $\Omega_{Z,k}(t)=\Omega_{Z,k}(t-1)$.
We will first show that $\linsp\left(\Omega_{Z,k}(t),\Omega_{R,k}(t)\right)\subseteq \linsp(\Omega_{Z,k}(t),\Omega_{M,k})$.

Since the target $d_k$ satisfies $k\in T\subseteq (S_{t-1}(X_{l,j_l})\cup\{l\})$, for all $l\in T$, by Lemma~\ref{lem:non-interfering}, all those $\vv_{t-1}(X_{l,j_l})$ are non-interfering from $d_k$'s perspective. That is,
\begin{align}
\forall l\in T,~ \vv_{t-1}(X_{l,j_l})&\in \linsp(\Omega_{Z,k}(t-1),\Omega_{M,k})\nonumber\\
&=\linsp(\Omega_{Z,k}(t),\Omega_{M,k}).\label{eq:vL}
\end{align}
As a result, we have $\vv_{t}(X_{k,j_k})=\vv_{\text{tx}}\in \linsp(\Omega_{Z,k}(t),\Omega_{M,k})$ since $\vv_\text{tx}$ is a linear combination of all $\vv_{t-1}(X_{l,j_l})$ for all $l\in T$.
Therefore, we have
\begin{align}
&\linsp\left(\Omega_{Z,k}(t),\Omega_{R,k}(t)\right)\nonumber\\
&=\linsp\left(\Omega_{Z,k}(t),\vv_{t}(X_{k,j_k}),\Omega_{R}'\right)\nonumber\\
&\subseteq \linsp(\Omega_{Z,t}(t), \Omega_{Z,t}(t),\Omega_{M,k},\Omega_R')\nonumber\\ &=\linsp(\Omega_{Z,k}(t),\Omega_{M,k},\Omega_R'). \label{eq:v-vs-t-4}
\end{align}
Since we condition on the event that \eqref{eq:cond-prev} holds, we have
\begin{align}
\linsp(\Omega_{Z,k}(t),\Omega_R')
&\subseteq \linsp\left(\Omega_{Z,k}(t),\vv_{t-1}(X_{k,j_k}),\Omega_{R}'\right)\nonumber\\
&=\linsp\left(\Omega_{Z,k}(t),\Omega_{R,k}(t-1)\right)\nonumber\\
&=\linsp\left(\Omega_{Z,k}(t-1),\Omega_{R,k}(t-1)\right)\nonumber\\
&=\linsp\left(\Omega_{Z,k}(t-1),\Omega_{M,k}\right)\nonumber\\
&=\linsp(\Omega_{Z,k}(t),\Omega_{M,k}).\label{eq:v-vs-t-3}
\end{align}
Joint \eqref{eq:v-vs-t-4} and \eqref{eq:v-vs-t-3} show that $\linsp\left(\Omega_{Z,k}(t),\Omega_{R,k}(t)\right)\subseteq \linsp(\Omega_{Z,k}(t),\Omega_{M,k})$.

To prove Lemma~\ref{lem:decodability} for Case 3.1, it remains to show that the event $\linsp\left(\Omega_{Z,k}(t),\Omega_{R,k}(t)\right)\supseteq \linsp(\Omega_{Z,k}(t),\Omega_{M,k})$ is of close-to-one probability, conditioning on \eqref{eq:cond-prev} being true. We consider two subcases: depending on whether the following equation is satisfied.
\begin{align}
\vv_{t-1}(X_{k,j_k})&\in\linsp\left(\Omega_{Z,k}(t-1), \Omega_R'\right)\nonumber\\
&=\linsp\left(\Omega_{Z,k}(t), \Omega_R'\right).\label{eq:v-vs-t-1}
\end{align}
Case 3.1.1: If \eqref{eq:v-vs-t-1} is satisfied, then we have
\begin{align}
&\linsp\left(\Omega_{Z,k}(t),\Omega_{R,k}(t)\right)\nonumber\\
&=\linsp\left(\Omega_{Z,k}(t),\vv_{t}(X_{k,j_k}),\Omega_{R}'\right)\nonumber\\
&\supseteq\linsp\left(\Omega_{Z,k}(t),\Omega_{R}'\right)\nonumber\\
&=\linsp\left(\Omega_{Z,k}(t),\vv_{t-1}(X_{k,j_k}),\Omega_{R}'\right)\label{eq:v-vs-t-2}\\
&=\linsp\left(\Omega_{Z,k}(t),\Omega_{R,k}(t-1)\right)\nonumber\\
&=\linsp\left(\Omega_{Z,k}(t-1),\Omega_{R,k}(t-1)\right)\nonumber\\
&=\linsp(\Omega_{Z,k}(t-1),\Omega_{M,k})\label{eq:v-vs-t-5}\\
&=\linsp(\Omega_{Z,k}(t),\Omega_{M,k}),\label{eq:v-vs-t-7}
\end{align}
where \eqref{eq:v-vs-t-2} follows from \eqref{eq:v-vs-t-1}, and \eqref{eq:v-vs-t-5} follows from the induction condition \eqref{eq:cond-prev}.

Case 3.1.2:
\eqref{eq:v-vs-t-1} is not satisfied.
By the equality between \eqref{eq:v-vs-t-2} and \eqref{eq:v-vs-t-7}, we have
\begin{align}
\linsp\left(\Omega_{Z,k}(t),\vv_{t-1}(X_{k,j_k}),\Omega_{R}'\right)
=\linsp(\Omega_{Z,k}(t),\Omega_{M,k}).\label{eq:v-vs-t-8}
\end{align}
Recall that $\vv_t(X_{k,j_k})=\vv_\text{tx}$ is a linear combination of $\vv_{t-1}(X_{l,j_l})$ satisfying in \eqref{eq:vL}. By \eqref{eq:v-vs-t-8} and the assumption that \eqref{eq:v-vs-t-1} is not satisfied,
 we thus have that each $\vv_{t-1}(X_{l,j_l})$ can be written as a unique linear combination of $\alpha \vv_{t-1}(X_{k,j_k})+\w$ where $\alpha$ is a $\GF(q)$ coefficient and $\w$ is a vector satisfying $\w\in \linsp\left(\Omega_{Z,k}(t),\Omega_{R}'\right)$. By the same reasoning, we can  rewrite $\vv_t(X_{k,j_k})$ as
\begin{align}
\vv_t(X_{k,j_k})&=c_{k} \vv_{t-1}(X_{k,j_k})+\sum_{\forall l\in T\backslash k}c_l\vv_{t-1}(X_{l,j_l})\nonumber\\
&=c_{k} \vv_{t-1}(X_{k,j_k})+(\alpha \vv_{t-1}(X_{k,j_k})+\w)\nonumber\\
&=(c_{k}+\alpha) \vv_{t-1}(X_{k,j_k})+\w.\label{eq:alpha-1}
\end{align}
where $\alpha$ is a $\GF(q)$ coefficient, $\w$ is a vector satisfying $\w\in \linsp\left(\Omega_{Z,k}(t),\Omega_{R}'\right)$, and the values of $\alpha$ and $\w$ depend on the random coefficients $c_l$ for all $l\neq k$. As a result, we have
\begin{align}
&\linsp\left(\Omega_{Z,k}(t),\Omega_{R,k}(t)\right)\nonumber\\
&=\linsp\left(\Omega_{Z,k}(t),\vv_{t}(X_{k,j_k}),\Omega_{R}'\right)\nonumber\\
&=\linsp\left(\Omega_{Z,k}(t),\left((c_{k}+\alpha) \vv_{t-1}(X_{k,j_k})+\w\right),\Omega_{R}'\right).\nonumber
\end{align}
Since \eqref{eq:v-vs-t-1} is not satisfied and $\w\in \linsp\left(\Omega_{Z,k}(t),\Omega_{R}'\right)$, we have
\begin{align}
&\linsp\left(\Omega_{Z,k}(t),\left((c_{k}+\alpha) \vv_{t-1}(X_{k,j_k})+\w\right),\Omega_{R}'\right)\nonumber\\
&=\linsp\left(\Omega_{Z,k}(t),\vv_{t-1}(X_{k,j_k}),\Omega_{R}'\right)\nonumber\\
&=\linsp(\Omega_{Z,k}(t),\Omega_{M,k})\label{eq:v-vs-t-6}
\end{align}
if and only if $(c_k+\alpha)\neq 0$. Since $c_k$ is uniformly distributed in $\GF(q)$ and the random variables $c_k$ and $\alpha$ are independent, the event that \eqref{eq:v-vs-t-6} is true has the conditional probability $\frac{q-1}{q}$, conditioning on \eqref{eq:cond-prev} being true. For sufficiently large $q$ values, the conditional probability approaches one. 

{\bf Case 3.2:} $k\in S_{t}(X_{k,j_k})$. Recall  that for Case~3, we consider those $k$ such that $k\notin S_{t-1}(X_{k,j_k})$.  By Line~\ref{line:S-update}
of the {\sc Update}, we have $k\in S_\text{rx}$, i.e., $d_k$ receives the transmitted packet perfectly in time $t$. Therefore, in the end of time $t$, $\Omega_{R,k}(t)=\Omega_R'$, which was first defined in \eqref{eq:omega-R-prime}.

 We consider two subcases: depending on whether the following equation is satisfied.
\begin{align}
\vv_{t-1}(X_{k,j_k})&\in\linsp\left(\Omega_{Z,k}(t-1), \Omega_R'\right).\label{eq:v-vs-t-1-22}
\end{align}

Case 3.2.1: If \eqref{eq:v-vs-t-1-22} is satisfied, then we have
\begin{align}
&\linsp\left(\Omega_{Z,k}(t),\Omega_{R,k}(t)\right)=\linsp\left(\Omega_{Z,k}(t),\Omega_{R}'\right)\nonumber\\
&=\linsp\left(\vv_t(X_{k,j_k}),\Omega_{Z,k}(t-1),\Omega_{R}'\right)\nonumber\\
&=\linsp\left(\vv_t(X_{k,j_k}),\Omega_{Z,k}(t-1),\vv_{t-1}(X_{k,j_k}),\Omega_{R}'\right)\label{eq:v-vs-t-2-22}\\
&=\linsp\left(\vv_t(X_{k,j_k}),\Omega_{Z,k}(t-1),\Omega_{R,k}(t-1)\right)\nonumber\\
&=\linsp(\vv_t(X_{k,j_k}),\Omega_{Z,k}(t-1),\Omega_{M,k})
\label{eq:v-vs-t-5-22}\\
&=\linsp(\Omega_{Z,k}(t),\Omega_{M,k}),\nonumber
\end{align}
where \eqref{eq:v-vs-t-2-22} follows from \eqref{eq:v-vs-t-1-22}, and \eqref{eq:v-vs-t-5-22} follows from the induction assumption \eqref{eq:cond-prev}.

Case 3.2.2:
\eqref{eq:v-vs-t-1-22} is not satisfied. By the induction assumption \eqref{eq:cond-prev}, we have
\begin{align}
&\linsp\left(\Omega_{Z,k}(t-1),\vv_{t-1}(X_{k,j_k}),\Omega_{R}'\right)\nonumber\\
&=\linsp(\Omega_{Z,k}(t-1),\Omega_{M,k}).\label{eq:v-vs-t-9}
\end{align}

Since the target $d_k$ satisfies $k\in T\subseteq (S_{t-1}(X_{l,j_l})\cup\{l\})$, for all $l\in T$, by Lemma~\ref{lem:non-interfering}, all those $\vv_{t-1}(X_{l,j_l})$ are non-interfering from $d_k$'s perspective. That is,
\begin{align}
\forall l\in T,~ \vv_{t-1}(X_{l,j_l})&\in \linsp(\Omega_{Z,k}(t-1),\Omega_{M,k}).\label{eq:v-vs-t-10}
\end{align}
By \eqref{eq:v-vs-t-9}, \eqref{eq:v-vs-t-10}, and the assumption that \eqref{eq:v-vs-t-1-22} is not satisfied,
 each $\vv_{t-1}(X_{l,j_l})$ can thus be written as a unique linear combination of $\alpha \vv_{t-1}(X_{k,j_k})+\w$ where $\alpha$ is a $\GF(q)$ coefficient and $\w$ is a vector satisfying $\w\in \linsp\left(\Omega_{Z,k}(t-1),\Omega_{R}'\right)$. Since $\vv_t(X_{k,j_k})=\vv_\text{tx}$ is a linear combination of $\vv_{t-1}(X_{l,j_l})$,
by the same reasoning, we can  rewrite $\vv_t(X_{k,j_k})$ as
\begin{align}
\vv_t(X_{k,j_k})&=c_{k} \vv_{t-1}(X_{k,j_k})+\sum_{\forall l\in T\backslash k}c_l\vv_{t-1}(X_{l,j_l})\nonumber\\
&=c_{k} \vv_{t-1}(X_{k,j_k})+(\alpha \vv_{t-1}(X_{k,j_k})+\w)\nonumber\\
&=(c_{k}+\alpha) \vv_{t-1}(X_{k,j_k})+\w.\label{eq:alpha-2}
\end{align}
where $\alpha$ is a $\GF(q)$ coefficient, $\w$ is a vector satisfying $\w\in \linsp\left(\Omega_{Z,k}(t-1),\Omega_{R}'\right)$, and the values of $\alpha$ and $\w$ depend on the random coefficients $c_l$ for all $l\neq k$. As a result, we have
\begin{align}
&\linsp\left(\Omega_{Z,k}(t),\Omega_{R,k}(t)\right)=\linsp\left(\Omega_{Z,k}(t),\Omega_{R}'\right)\nonumber\\
&=\linsp\left(\vv_t(X_{k,j_k}),\Omega_{Z,k}(t-1),\Omega_{R}'\right)\label{eq:v-vs-t-7-22}\\
&=\linsp\left(\vv_t(X_{k,j_k}),\Omega_{Z,k}(t-1),\vv_{t-1}(X_{k,j_k}),\Omega_{R}'\right)\label{eq:v-vs-t-6-22}\\
&=\linsp\left(\vv_t(X_{k,j_k}),\Omega_{Z,k}(t-1),\Omega_{R,k}(t-1)\right)\nonumber\\
&=\linsp\left(\vv_t(X_{k,j_k}),\Omega_{Z,k}(t-1),\Omega_{M,k}\right)\nonumber\\
&=\linsp\left(\Omega_{Z,k}(t),\Omega_{M,k}\right),\nonumber
\end{align}
where the equality from \eqref{eq:v-vs-t-7-22} to \eqref{eq:v-vs-t-6-22} is true
if and only if the $(c_k+\alpha)$ in \eqref{eq:alpha-2} is not zero, since \eqref{eq:v-vs-t-1-22} is not satisfied and $\w\in \linsp\left(\Omega_{Z,k}(t-1),\Omega_{R}'\right)$.

Since $c_k$ is uniformly distributed in $\GF(q)$ and the random variables $c_k$ and $\alpha$ are independent, the event that $\linsp\left(\Omega_{Z,k}(t),\Omega_{R,k}(t)\right)=\linsp\left(\Omega_{Z,k}(t),\Omega_{R,M}\right)$ has the conditional probability $\frac{q-1}{q}$, conditioning on \eqref{eq:cond-prev} being true. For sufficiently large $q$ values, the conditional probability approaches one. 

{\bf Combining all cases:}
Let $\mathcal{A}_t$ denote the event that $\linsp(\Omega_{Z,k}(t), \Omega_{R,k}(t))=\linsp(\Omega_{Z,k}(t), \Omega_{M,k})$ and let $T$ denote the target set chosen in time $t$. Since for Cases 3.1.2 and 3.2.2 the conditional probability of $\mathcal{A}_t$ given $\mathcal{A}_{t-1}$ is lower bounded by $\frac{q-1}{q}$ and for all other cases the conditional probability is one, the discussion of Cases~1 to 3.2 thus proves the following inequalities:
\begin{align}
&\prop\left(\mathcal{A}_t|\mathcal{A}_{t-1},T\right)\geq \left(1-\frac{1}{q}\right)^{|T|}\nonumber.
\end{align}
Since for any $T\subseteq[K]$ we must have $|T|\leq K$, we then have
\begin{align}
\prop\left(\mathcal{A}_t|\mathcal{A}_{t-1}\right)\geq \left(1-\frac{1}{q}\right)^{K}.\nonumber
\end{align}

By concatenating the conditional probabilities, we thus have
\begin{align}
&\prop\left(\linsp(\Omega_{Z,k}(t), \Omega_{R,k}(t))=\linsp(\Omega_{Z,k}(t), \Omega_{M,k})\right)\nonumber\\
&\geq\left(1-\frac{1}{q}\right)^{tK}\geq \left(1-\frac{1}{q}\right)^{nK}.\label{eq:dec-final}
\end{align}
As a result, for any fixed $K$ and $n$ values, we can choose a sufficiently large finite field $\GF(q)$ such that \eqref{eq:dec-final} approaches one. Lemma~\ref{lem:decodability} thus holds for all $k\in[K]$ and $t\in[n]$.

\end{thmproof}

\section{A Proof Of Proposition~\ref{prop:finite-field}\label{app:finite-field}}

\begin{thmproof}{Proof of Proposition~\ref{prop:finite-field}:} To prove this proposition, we will show that for any $q_0>K$, the source $s$ can always {\em compute} the mixing coefficients $\{c_k:\forall k\in T\}$ in Line~\ref{line:vv-construct} of the PE scheme, such that the key properties in Lemmas~\ref{lem:non-interfering} and~\ref{lem:decodability} hold with probability one. Then for any PE scheme, we can use the computed mixing coefficients $\{c_k:\forall k\in T\}$ instead of the randomly chosen ones, while attaining the same desired throughput performance.



We first notice that the proof of Lemma~\ref{lem:non-interfering} does not involve any probabilistic arguments. Therefore, Lemma~\ref{lem:non-interfering} holds for any choices of the mixing coefficients with probability one.

We use induction to prove that when using carefully computed  mixing coefficients $\{c_k:\forall k\in T\}$, Lemma~\ref{lem:decodability} holds with probability one. We use the same notation of $S_t(X_{k,j})$, $\vv_t(X_{k,j})$, $\Omega_{R,k}(t)$, $\Omega_{Z,k}(t)$, $\Omega_{M,k}$, $\Omega_R'$ as defined in Lemma~\ref{lem:decodability} and its proof.\footnote{We note that $\Omega_R'$ in \eqref{eq:omega-R-prime} actually depends on the value of $k$ and the time index  $(t-1)$. }

 In the end of time $t=0$, since
\begin{align}
\Omega_{R,k}(0)&=\linsp(\vv_0(X_{k,j}):\forall j\in [nR_k], k\notin S_0(X_{k,j})=\emptyset)\nonumber\\
&=\linsp(\delta_{k,j}:\forall k\in[K],j\in [nR_k])=\Omega_{M,k},\nonumber
\end{align}
we have
\begin{align}\prop\left(\linsp(\Omega_{Z,k}(0),\Omega_{R,k}(0))
=\linsp(\Omega_{Z,k}(0),\Omega_{M,k})\right)=1.\nonumber
\end{align}
Lemma~\ref{lem:decodability} holds with probability one for any finite field $\GF(q_0)$.

Assume that in the end of time $(t-1)$, Lemma~\ref{lem:decodability} holds with probability one.
Suppose $T$ is chosen in the beginning of time $t$. Define $B_t$ as the set of $k$ values satisfying:
\begin{align}
B_t\stackrel{\Delta}{=}\{\forall k\in T:& k\notin S_{t-1}(X_{k,j_k}) \text{ and}\nonumber\\
&\vv_{t-1}(X_{k,j_k})\notin \linsp(\Omega_{Z,k}(t-1),\Omega_{R'})  \}.\nonumber
\end{align}
Note that this $B_t$ can be computed in the beginning of time $t$. Once $B_t$ is computed, we would like to choose the mixing coefficients $\{c_l:\forall l\in T\}$ such that the following equation is satisfied.
\begin{align}
\forall k\in B_t,\quad \vv_\text{tx}&=\sum_{\forall l\in T}c_l\vv_{t-1}(X_{l,j_l})\nonumber\\
&=c_{k} \vv_{t-1}(X_{k,j_k})+\sum_{\forall l\in T\backslash k}c_l\vv_{t-1}(X_{l,j_l})\nonumber\\
&\notin  \linsp(\Omega_{Z,k}(t-1),\Omega_{R'}).\label{eq:ck-all-k}
\end{align}
Note that for any $k\in B_t$, we have $\vv_{t-1}(X_{k,j_k})\notin \linsp(\Omega_{Z,k}(t-1),\Omega_{R'})$. Therefore if we choose the coefficients  $\{c_l:\forall l\in T\}$ uniformly randomly, the probability that
\begin{align}
&c_{k} \vv_{t-1}(X_{k,j_k})+\sum_{\forall l\in T\backslash k}c_l\vv_{t-1}(X_{l,j_l})\nonumber\\
&\in  \linsp(\Omega_{Z,k}(t-1),\Omega_{R'})\label{eq:ck-give-k}
\end{align}
is at most $\frac{1}{q_0}$. The probability that there is at least one $k\in T$ satisfying \eqref{eq:ck-give-k} has probability at most $\frac{|B_t|}{q_0}\leq \frac{K}{q_0}$. For any $q_0>K$, we thus have a non-zero probability $\geq (1-\frac{K}{q_0})$ such that the uniformly random choice of $\{c_l:\forall l\in T\}$ will satisfy \eqref{eq:ck-all-k}. Therefore, there must exist at least one $\{c_l:\forall l\in T\}$ satisfying \eqref{eq:ck-all-k}.
 In the beginning of time $t$, we arbitrarily choose any such mixing coefficients  $\{c_l:\forall l\in T\}$ that satisfy \eqref{eq:ck-all-k}.

The remaining task is to show that the above construction of $\{c_k:\forall k\in T\}$ guarantees that Lemma~\ref{lem:decodability} holds in the end of time $t$ with probability one, regardless the channel realization of time $t$.

For those $k\notin T$, such $k$ falls into Case~2 of the proof of Lemma~\ref{lem:decodability}. Since Case~2 holds with probability one, Lemma~\ref{lem:decodability} is true for those $k\notin T$ with probability one.
For those $k\in T$ and $k\in S_{t-1}(X_{k,j_k})$, then such $k$ falls into Case 1 of the proof of Lemma~\ref{lem:decodability}. Since Case~1 holds with probability one, Lemma~\ref{lem:decodability} is true for those $k\in T$ and $k\in S_{t-1}(X_{k,j_k})$ with probability one.

For those $k$ satisfying: $k\in T$, $
k\notin S_{t-1}(X_{k,j_k})$, and $
\vv_{t-1}(X_{k,j_k})\in \linsp(\Omega_{Z,k}(t-1),\Omega_{R'})$, such $k$ must fall into Case~1, Case~3.1.1, or Case~3.2.1, depending on whether $S_{t}(X_{k,j_k})=S_{t-1}(X_{k,j_k})$ and whether $k\in S_{t}(X_{k,j_k})$, respectively. Since Cases~1, 3.1.1, and 3.2.1 hold with probability one, Lemma~\ref{lem:decodability} is true for those $k$ with probability one.

The remaining $k$'s to consider are those $k\in B_t$.  If the random channel realization leads to $S_{t}(X_{k,j_k})=S_{t-1}(X_{k,j_k})$, then by Case~1 of the proof of Lemma~\ref{lem:decodability}, we must have Lemma~\ref{lem:decodability} holds with conditional probability one. If the random channel realization leads to $S_{t}(X_{k,j_k})\neq S_{t-1}(X_{k,j_k})$ and $k\notin S_t(X_{k,j_k})$, then we are in Case~3.1.2. Since for those $k\in B_t$ we have chosen the mixing coefficients $\{c_l:\forall l\in T\}$ satisfying \eqref{eq:ck-all-k}, following the same arguments as in \eqref{eq:alpha-1} we must be able to rewrite $\vv_t(X_{k,j_k})$ as follows.
\begin{align}
\vv_t(X_{k,j_k})&=\vv_\text{tx}\nonumber\\
&=(c_{k}+\alpha) \vv_{t-1}(X_{k,j_k})+\w\nonumber
\end{align}
where $(c_k+\alpha)$ is a non-zero $\GF(q)$ coefficient, and $\w$ is a vector satisfying \begin{align}\w\in \linsp\left(\Omega_{Z,k}(t-1),\Omega_{R}'\right)=\linsp\left(\Omega_{Z,k}(t),\Omega_{R}'\right).\nonumber\end{align} Following the same proof of Case~3.1.2 of Lemma~\ref{lem:decodability}, we must have Lemma~\ref{lem:decodability} holds with conditional probability one.
If the random channel realization leads to $S_{t}(X_{k,j_k})\neq S_{t-1}(X_{k,j_k})$ and $k\in S_t(X_{k,j_k})$, then we are in Case~3.2.2. Since for those $k\in B_t$ we have chosen the mixing coefficients $\{c_l:\forall l\in T\}$ satisfying \eqref{eq:ck-all-k},
following the same arguments as in \eqref{eq:alpha-2} we must be able to rewrite $\vv_t(X_{k,j_k})$ as follows.
\begin{align}
\vv_t(X_{k,j_k})&=\vv_\text{tx}\nonumber\\
&=(c_{k}+\alpha) \vv_{t-1}(X_{k,j_k})+\w'\nonumber
\end{align}
where $(c_k+\alpha)$ is a non-zero $\GF(q)$ coefficient, and $\w'$ is a vector satisfying $\w'\in \linsp\left(\Omega_{Z,k}(t-1),\Omega_{R}'\right)$. Following the same proof of Case~3.2.2 of Lemma~\ref{lem:decodability}, we must have Lemma~\ref{lem:decodability} holds with conditional probability one. Since regardless of the random channel realization, Lemma~\ref{lem:decodability} holds with probability one, we have thus shown that one can always construct the desired mixing coefficients $\{c_l:\forall l\in T\}$ provided the finite field  $\GF(q_0)$ satisfying $q_0>K$. By induction on $t$, the proof is complete.

\end{thmproof}

\section{A Key Lemma For The Proof Of Proposition~\ref{prop:cap-osf}\label{app:osf-lemma}}

Consider an arbitrary spatially independent 1-to-$K$ broadcast PEC with marginal success probabilities $0<p_1\leq p_2\leq\cdots \leq p_K$. For any $S\subseteq [K]$ and $S\neq [K]$, define
\begin{align}
L_{S}\stackrel{\Delta}{=}
\sum_{i=K-|S|}^{K}\left(\sum_{\scriptsize\begin{array}{c}\forall S_1: |S_1|=i\\ ([K]\backslash S)\subseteq S_1\subseteq [K]\end{array}}\frac{\left(-1\right)^{i-(K-|S|)}}{p_{\cup S_1}}\right).\nonumber
\end{align}

We then have the following lemma:
\begin{lemma}\label{lem:osf-main} Suppose the 1-to-$K$ broadcast PEC is spatially independent with marginal success probabilities $0<p_1\leq \cdots\leq p_K$. Consider any one-sidedly fair rate vector $(R_1,\cdots, R_K)\in\Lambda_\text{osf}$, and any non-empty subset $T\subseteq [K]$. For any $k_1, k_2\in T$ with $k_1< k_2$, we have
\begin{align}
R_{k_1}\cdot L_{T\backslash k_1}\geq R_{k_2}\cdot L_{T\backslash k_2}.\nonumber
\end{align}
\end{lemma}

\begin{proof}
Consider $K$ independent geometric random variables $X_1$ to $X_K$ with success probability $p_1$ to $p_K$. That is, the probability mass function $F_k(t)$ of any $X_k$ satisfies
\begin{align}
F_k(t)\stackrel{\Delta}{=}\prop(X_k=t)=p_k(1-p_k)^{t-1},\nonumber
\end{align}
for all strictly positive integer $t$. For the sake of simplicity, here we omit the discussion of the degenerate case in which $p_k=1$.  We say that the geometric random trial $X_k$ is finished at time $t$ if $X_k=t$. For any $S\subseteq [K]$ and $S\neq [K]$, define three random variables
\begin{align}
Y_{[K]\backslash S}&\stackrel{\Delta}{=}\min (X_i:i\in [K]\backslash S)\label{eq:big-Y}\\
W_{S}&\stackrel{\Delta}{=}\max (X_i:i\in S)\label{eq:big-W}\\
\Gamma_S&\stackrel{\Delta}{=}Y_{[K]\backslash S}-\min(Y_{[K]\backslash S},W_{S}).\label{eq:big-Gamma}
\end{align}

\noindent {\bf Intermediate Step~1:} We will first show that
\begin{align}
L_{S}=\EE\left\{\Gamma_S\right\}.\nonumber
\end{align}
To that end, for any time $t$, we mark time $t$ by a set $I_t\stackrel{\Delta}{=}\{i\in [K]: X_i<t\}$. We then have
\begin{align}
\Gamma_S=Y_{[K]\backslash S}-\min(Y_{[K]\backslash S},W_{S})=\sum_{t=1}^\infty 1_{\{I_t=S\}}.\nonumber
\end{align}
By noting that
\begin{align}
t\leq Y_{[K]\backslash S}\Longleftrightarrow I_t\subseteq S, \nonumber
\end{align}
we also have
\begin{align}
Y_{[K]\backslash S}=\sum_{t=1}^\infty 1_{\{t\leq Y_{[K]\backslash S}\}}=\sum_{t=1}^\infty 1_{\{I_t\subseteq S\}}=\sum_{\forall S': S'\subseteq S} \Gamma_{S'}.\label{eq:Y-sum}
\end{align}
Taking the expectation of \eqref{eq:Y-sum}, we then have
\begin{align}
\forall S\subsetneq [K],~\sum_{\forall S': S'\subseteq S} \EE\left\{\Gamma_{S'}\right\}=\EE\left\{Y_{[K]\backslash S}\right\}=\frac{1}{p_{\cup ([K]\backslash S)}}.\label{eq:Gamma-sum}
\end{align}
Solving the simultaneous equations \eqref{eq:Gamma-sum}, we have
\begin{align}
\EE\left\{\Gamma_{S'}\right\}&=\sum_{i=K-|S'|}^{K}\left(\sum_{\scriptsize\begin{array}{c}\forall S_1: |S_1|=i\\ ([K]\backslash S')\subseteq S_1\subseteq [K]\end{array}}\frac{\left(-1\right)^{i-(K-|S'|)}}{p_{\cup S_1}}\right)\nonumber\\
&=L_{S'},\nonumber
\end{align}
for all $S'\subseteq [K]$ and $S'\neq [K]$.

\noindent {\bf Intermediate Step~2:} We will show that
for any non-empty subset $T\subseteq [K]$ and any $k_1, k_2\in T$ with $k_1< k_2$, we have
\begin{align}
\frac{L_{T\backslash k_1}}{1- p_{k_1}} \geq \frac{L_{T\backslash k_2}}{1- p_{k_2}}.\label{eq:LLW}
\end{align}
For any realization $(X_1,\cdots, X_K)=(x_1,\cdots, x_K)$, we use $y_{[K]\backslash S}$, $w_S$, and $\gamma_S$ to denote the corresponding values of $Y_{[K]\backslash S}$, $W_S$, and $\Gamma_S$ according to \eqref{eq:big-Y}, \eqref{eq:big-W}, and \eqref{eq:big-Gamma}, respectively.
We then have
\begin{align}
\EE\left\{\Gamma_{T\backslash k_1}\right\}&=\sum_{\forall (x_1,\cdots, x_K)} \gamma_{T\backslash k_1} \prod_{k=1}^K F_k(x_k)\nonumber\\
&=\sum_{\forall (x_1,\cdots, x_K):\gamma_{T\backslash k_1}>0} \gamma_{T\backslash k_1} \prod_{k=1}^K F_k(x_k).\label{eq:m-change-1}
\end{align}
Note that the only difference between $\EE\left\{\Gamma_{T\backslash k_1}\right\}$ and $\EE\left\{\Gamma_{T\backslash k_k}\right\}$ is the underlying measures of $X_{k_1}$ and $X_{k_2}$. Therefore, by the change of measure formula, we have
\begin{align}
\EE\left\{\Gamma_{T\backslash k_2}\right\}&=\sum_{\forall (x_1,\cdots, x_K):\gamma_{T\backslash k_1}>0} \gamma_{T\backslash k_1}\cdot\nonumber\\
&\hspace{1cm} \left(\frac{F_{k_2}(x_{k_1})}{F_{k_1}(x_{k_1})}\frac{F_{k_1}(x_{k_2})}{F_{k_2}(x_{k_2})}\right)\prod_{k=1}^K F_k(x_k).\label{eq:m-change-2}
\end{align}
Note that when $\gamma_{T\backslash k_1}>0$, we must have
$y_{([K]\backslash T)\cup\{k_1\}}>w_{T\backslash k_1}$,
which in turn implies that $x_{k_1}\geq x_{k_2}+1$.
We then have
\begin{align}
\frac{F_{k_2}(x_{k_1})}{F_{k_1}(x_{k_1})}\frac{F_{k_1}(x_{k_2})}{F_{k_2}(x_{k_2})}&=\frac{p_{k_2}(1-p_{k_2})^{x_{k_1}}}{p_{k_1}(1-p_{k_1})^{x_{k_1}}}\frac{p_{k_1}(1-p_{k_1})^{x_{k_2}}}{p_{k_2}(1-p_{k_2})^{x_{k_2}}}\nonumber\\
&=\left(\frac{1-p_{k_2}}{1-p_{k_1}}\right)^{x_{k_1}-x_{k_2}}\nonumber\\
&\leq \left(\frac{1-p_{k_2}}{1-p_{k_1}}\right),\label{eq:m-change-3}
\end{align}
where the last inequality follows from $p_{k_1}\leq p_{k_2}$ and $x_{k_1}\geq x_{k_2}+1$. Combining \eqref{eq:m-change-1}, \eqref{eq:m-change-2}, and \eqref{eq:m-change-3}, we thus have
\begin{align}
\EE\{\Gamma_{T\backslash k_1}\}\left(\frac{1-p_{k_2}}{1-p_{k_1}}\right)\geq \EE\{\Gamma_{T\backslash k_2}\},\nonumber
\end{align}
which implies \eqref{eq:LLW}.

{\bf Final Step~3:} Since $(R_1,\cdots, R_K)\in \Lambda_\text{osf}$, by the definition of one-sided fairness, we have
\begin{align}
R_{k_1}(1-p_{k_1})\geq R_{k_2}(1-p_{k_2}). \label{eq:RRR}
\end{align}
Multiplying \eqref{eq:LLW} and \eqref{eq:RRR} together, the proof of Lemma~\ref{lem:osf-main} is complete.

\end{proof}

\bibliography{ntwkcoding,nc_sys,chihw}

\begin{thebibliography}{10}
\providecommand{\url}[1]{#1}
\csname url@samestyle\endcsname
\providecommand{\newblock}{\relax}
\providecommand{\bibinfo}[2]{#2}
\providecommand{\BIBentrySTDinterwordspacing}{\spaceskip=0pt\relax}
\providecommand{\BIBentryALTinterwordstretchfactor}{4}
\providecommand{\BIBentryALTinterwordspacing}{\spaceskip=\fontdimen2\font plus
\BIBentryALTinterwordstretchfactor\fontdimen3\font minus
  \fontdimen4\font\relax}
\providecommand{\BIBforeignlanguage}[2]{{%
\expandafter\ifx\csname l@#1\endcsname\relax
\typeout{** WARNING: IEEEtranS.bst: No hyphenation pattern has been}%
\typeout{** loaded for the language `#1'. Using the pattern for}%
\typeout{** the default language instead.}%
\else
\language=\csname l@#1\endcsname
\fi
#2}}
\providecommand{\BIBdecl}{\relax}
\BIBdecl

\bibitem{AvestimehrDiggaviTse07}
S.~Avestimehr, S.~Diggavi, and D.~Tse, ``A deterministic approach to wireless
  relay network,'' in \emph{Proc.\ 45th Annual Allerton Conf.\ on Comm.,
  Contr., and Computing}.\hskip 1em plus 0.5em minus 0.4em\relax Monticello,
  IL, September 2007.

\bibitem{Bergmans73}
P.~Bergmans, ``Random coding theorem for broadcast channels with degraded
  components,'' \emph{IEEE Trans.\ Inform.\ Theory}, vol.~19, pp. 197--207,
  March 1973.

\bibitem{CadambeJafar08}
V.~Cadambe and S.~Jafar, ``Interference alignment and degrees of freedom of the
  $k$-user interference channel,'' \emph{IEEE Trans.\ Inform.\ Theory},
  vol.~54, no.~8, pp. 3425--3441, August 2008.

\bibitem{ChouWuJain03}
P.~Chou, Y.~Wu, and K.~Jain, ``Practical network coding,'' in \emph{Proc.\ 41st
  Annual Allerton Conf.\ on Comm., Contr., and Computing}.\hskip 1em plus 0.5em
  minus 0.4em\relax Monticello, IL, October 2003.

\bibitem{Cover98}
T.~M. Cover, ``Comments on broadcast channels,'' \emph{IEEE Trans. Inform.
  Theory}, vol.~44, no.~6, pp. 2524--2530, Oct. 1998.

\bibitem{DanaGowaikarPalankiHassibiEffros06}
A.~Dana, R.~Gowaikar, R.~Palanki, B.~Hassibi, and M.~Effros, ``Capacity of
  wireless erasure networks,'' \emph{IEEE Trans.\ Inform.\ Theory}, vol.~52,
  no.~3, pp. 789--804, March 2006.

\bibitem{DasVishwanathJafarMarkopoulou10}
A.~Das, S.~Vishwanath, S.~Jafar, and A.~Markopoulou, ``Network coding for
  multiple unicasts: An interference alignment approach,'' in \emph{Proc.\ IEEE
  Int'l Symp.\ Inform.\ Theory}.\hskip 1em plus 0.5em minus 0.4em\relax Austin,
  Texas, USA, June 2010.

\bibitem{Gamal78}
A.~{El Gamal}, ``The feedback capacity of degraded broadcast channels,''
  \emph{IEEE Trans.\ Inform.\ Theory}, vol.~25, no.~2, pp. 379--381, March
  1978.

\bibitem{GeorgiadisTassiulas09}
L.~Georgiadis and L.~Tassiulas, ``Broadcast erasure channel with feedback ---
  capacity and algorithms,'' in \emph{Proc.\ 5th Workshop on Network Coding,
  Theory, \& Applications (NetCod)}.\hskip 1em plus 0.5em minus 0.4em\relax
  Lausanne, Switzerland, June 2009, pp. 54--61.

\bibitem{HoMedardKoetterKargerEffrosShiLeong06}
T.~Ho, M.~M\'{e}dard, R.~Koetter, D.~Karger, M.~Effros, J.~Shi, and B.~Leong,
  ``A random linear network coding approach to multicast,'' \emph{IEEE Trans.\
  Inform.\ Theory}, vol.~52, no.~10, pp. 4413--4430, October 2006.

\bibitem{KattiRahulHuKatabiMedardCrowcroft06}
S.~Katti, H.~Rahul, W.~Hu, D.~Katabi, M.~M\'{e}dard, and J.~Crowcroft, ``{XORs}
  in the air: Practical wireless network,'' in \emph{Proc.\ ACM Special
  Interest Group on Data Commun.\ (SIGCOMM)}, 2006.

\bibitem{KornerMarton77}
J.~{K\"{o}rner} and K.~Marton, ``General broadcast channels with degraded
  message sets,'' \emph{IEEE Trans.\ Inform.\ Theory}, vol.~23, no.~1, pp.
  60--64, January 1977.

\bibitem{KoutsonikolasWangHu10}
D.~Koutsonikolas, C.-C. Wang, and Y.~Hu, ``{CCACK}: Efficient network coding
  based opportunistic routing through cumulative coded acknowledgments,'' in
  \emph{Proc.\ 29th IEEE Conference on Computer Communications
  (INFOCOM)}.\hskip 1em plus 0.5em minus 0.4em\relax San Diego, USA, March
  2010, pp. 1--9.

\bibitem{KoutsonikolasWangHu10a}
------, ``{ECR:} enhanced coded retransmission for downlink access-point
  networks,'' in \emph{ACM CoNEXT}, 2010, to be submitted.

\bibitem{LarssonJohansson06}
P.~Larsson and N.~Johansson, ``Multi-user {ARQ},'' in \emph{Proc.\ IEEE
  Vehicular Technology Conference}, 2006.

\bibitem{LiYeungCai03}
S.-Y. Li, R.~Yeung, and N.~Cai, ``Linear network coding,'' \emph{IEEE Trans.\
  Inform.\ Theory}, vol.~49, no.~2, pp. 371--381, February 2003.

\bibitem{LiWangLin10}
X.~Li, C.-C. Wang, and X.~Lin, ``Throughput and delay analysis on uncoded and
  coded wireless broadcast with hard deadline constraints,'' in \emph{Proc.\
  29th IEEE Conference on Computer Communications (INFOCOM)}.\hskip 1em plus
  0.5em minus 0.4em\relax San Diego, USA, March 2010, pp. 1--5.

\bibitem{OzarowCheong84}
L.~Ozarow and S.~Leung-Yan-Cheong, ``An achievable region and outer bound for
  the {Gaussian} broadcast channel with feedback,'' \emph{IEEE Trans.\ Inform.\
  Theory}, vol.~30, no.~4, July 1984.

\bibitem{RoznerIyerMehtaQiuJafry07}
E.~Rozner, A.~Iyer, Y.~Mehta, L.~Qiu, and M.~Jafry, ``{ER}: Efficient
  retranmission scheme for wireless {LANs},'' in \emph{Proc.\ ACM
  CoNEXT}.\hskip 1em plus 0.5em minus 0.4em\relax New York, USA, December 2007.

\bibitem{SundararajanShahMedard07}
J.~Sundararajan, D.~Shah, and M.~M\'{e}dard, ``{ARQ} for network coding,'' in
  \emph{Proc.\ IEEE Int'l Symp.\ Inform.\ Theory}.\hskip 1em plus 0.5em minus
  0.4em\relax Toronto, Canada, July 2008.

\bibitem{Wang10a}
C.-C. Wang, ``On the capacity of wireless 1-hop intersession network coding ---
  a broadcast packet erasure channel approach,'' in \emph{Proc.\ IEEE Int'l
  Symp.\ Inform.\ Theory}.\hskip 1em plus 0.5em minus 0.4em\relax Austin, TX,
  USA, June 2010.

\bibitem{WangKhreishahShroff09}
C.-C. Wang, A.~Khreishah, and N.~Shroff, ``Cross-layer optimizations for
  intersession network coding on practical 2-hop relay networks,'' in
  \emph{Proc.\ 43rd Asilomar Conference on Signals, Systems and
  Computers}.\hskip 1em plus 0.5em minus 0.4em\relax Pacific Grove, CA, USA,
  November 2009, pp. 771--775.

\bibitem{WeingartenSteinbergShamai06}
H.~Weingarten, Y.~Steinberg, and S.~Shamai, ``The capacity region of the
  {Gaussian} multiple-input multiple-output broadcast channel,'' \emph{IEEE
  Trans.\ Inform.\ Theory}, vol.~52, no.~9, pp. 3936--3964, September 2006.

\bibitem{Wu07}
Y.~Wu, ``Broadcasting when receivers know some messages a priori,'' in
  \emph{Proc.\ IEEE Int'l Symp.\ Inform.\ Theory}.\hskip 1em plus 0.5em minus
  0.4em\relax Nice, France, June 2007, pp. 1141--1145.

\bibitem{XueYang08}
F.~Xue and X.~Yang, ``Network coding and packet-erasure broadcast channel,'' in
  \emph{Proc. the 5th IEEE Annual Communications Society Conference on Sensor,
  Mesh and Ad Hoc Communications and Networks Workshops (SECON)}.\hskip 1em
  plus 0.5em minus 0.4em\relax San Francisco, CA, USA, July 2008.

\end{thebibliography}
\bibliographystyle{IEEEtranS}

\end{document}